\documentclass[10pt]{article}

\usepackage[T1]{fontenc}
\usepackage{amsmath,amssymb,amsfonts,amsthm,textcomp,graphicx,nicefrac,mathtools,mathrsfs, dsfont} 
\usepackage{tikz}
\usepackage{wrapfig}

\usepackage{bbm}
\usepackage[utf8]{inputenc}
\usepackage{enumitem}
\usepackage{mathrsfs}

\def\<{{\langle}}
\def\>{{\rangle}}

\newcommand{\tpose}{\top}
\renewcommand{\exp}[1]{\operatorname{exp}\left(#1\right)} 

\def\SL{\mathbb{SL}}
\def\gofbb{\mathbb{GOF}}
\def\eofbb{\mathbb{EOF}}

\newtheoremstyle{dotless}{}{}{\itshape}{}{\bfseries}{}{ }{}
\theoremstyle{dotless}

\newtheorem*{defi}{Definition}
\theoremstyle{plain}
\newtheorem{myth}{Theorem}

\newtheorem{myprop}[myth]{Proposition}
\newtheorem{mylem}[myth]{Lemma}

\newtheoremstyle{named}{}{}{\itshape}{}{\bfseries}{.}{.5em}{#1 #3}
\theoremstyle{named}
\newtheorem*{namthm*}{Theorem}

\usepackage{chngcntr}

\usepackage[colorlinks=true, citecolor = blue]{hyperref}
\usepackage{cleveref}
\crefname{myth}{Theorem}{Theorems} 

\newcounter{parentnumber}

\usepackage[
backend=biber,
style=alphabetic,
sorting=anyt,
maxbibnames=99
]{biblatex}
\newcommand{\edit}[1]{\textcolor{red}{#1}}

\newcommand{\ii}{\mathcal{I}}
\usepackage{enumitem}
\usepackage{graphicx}

\newcommand{\gof}{\mathrm{GoF}}
\newcommand{\eof}{\mathrm{EoF}}
\newcommand{\sll}{\mathrm{SL}}
\newcommand{\SNR}{\mathrm{SNR}}
\newcommand{\delete}{\mathrm{del}}

\newcommand{\qmix}{\left\langle \mathcal{Q} ^{\otimes n} \right\rangle}

\usepackage{csquotes}
\usepackage{todonotes}

\title{Limits on Testing Structural Changes in Ising Models}

\author{%
  Aditya Gangrade, Bobak Nazer, Venkatesh Saligrama\\ Boston University \\ \texttt{\{gangrade, bobak, srv\}@bu.edu}
}

\date{\vspace{-12pt}}

\usepackage[
    textheight=9in,
    textwidth=5.5in,
    top=1in,
    headheight=12pt,
    headsep=25pt,
    footskip=30pt
]{geometry}
  
\begin{document}

\maketitle
\begin{abstract}
We present novel information-theoretic limits on detecting sparse changes in Ising models, a problem that arises in many applications where network changes can occur due to some external stimuli. We show that the sample complexity for detecting sparse changes, in a minimax sense, is no better than learning the entire model even in settings with local sparsity. This is a surprising fact in light of prior work rooted in sparse recovery methods, which suggest that sample complexity in this context scales only with the number of network changes. To shed light on when change detection is easier than structured learning, we consider testing of edge deletion in forest-structured graphs, and high-temperature ferromagnets as case studies. We show for these that testing of small changes is similarly hard, but testing of \emph{large} changes is well-separated from structure learning. These results imply that testing of graphical models may not be amenable to concepts such as restricted strong convexity leveraged for sparsity pattern recovery, and algorithm development instead should be directed towards detection of large changes.
\end{abstract}

\section{Introduction}\label{sec:intro}
Recent technological advances have lead to the emergence of high-dimensional datasets in a wide range of scientific disciplines~\cite{yang2017vivo,costanzo2010genetic,phizicky1995protein,bresler2015structurelearning,lokhov2018optimal,wu2019sparse,bandeira18}, where the observations are modeled as arising from a probabilistic graphical model (GM), and the goal is to recover the network~\cite{neuralconnectomics}. While full network recovery is sometimes useful, and there has been a flurry of activity \cite{drton2017structure,SanWai} in this context, we are often interested in \emph{changes} in network structure in response to external stimuli, such as changes in protein-protein interactions across different disease states~\cite{ideker2012differential} or changes in neuronal connectivity as a subject learns a task~\cite{mohammed2016integrative}.

A baseline approach is to estimate the network at each stage, and then compare the differences. However, such observations exhibit significant variability, and the amount of data available may be too small for this approach to yield meaningful results. On the other hand, \emph{reliably recovering network changes should be easier than full reconstruction}. While prior works have proposed inference algorithms to explore this possibility~\cite{zhao2014direct,xia2015testing,FazBan,belilovsky2016testing,bodwin2018testing,zhang2019diffnetfdr,cai2019differential}, we do not have a good mathematical understanding of when this is indeed easier.

To shed light on this question, we propose to derive information-theoretic limits for two structural inference problems over degree-bounded Ising models. The first is goodness-of-fit testing ($\gofbb$). Let $G(P)$ be the network structure (see \S\ref{sec:def}) of an Ising model $P$. $\gofbb$ is posed as follows.\vspace{-5pt} \begin{displayquote} $\mathbf{\gofbb}:$ \emph{Given an Ising model $P$ and i.i.d.~samples from another Ising model $Q$, determine if $P = Q$ or if $G(P)$ and $G(Q)$ differ in at least $s$ edges.} \vspace{-7pt} \end{displayquote}  The second is a related estimation problem, termed error-of-fit ($\eofbb$), that demands localising differences in $G(P)$ and $G(Q)$ (if distinct).\vspace{-5pt}  \begin{displayquote} $\mathbf{\eofbb}$: \emph{Given an Ising model $P$ and i.i.d.~samples from another Ising model $Q$ that is either equal to $P$, or has a network structure that differs from that of $P$ in $s$ edges or more, determine the edges where $G(P)$ and $G(Q)$ differ.} \vspace{-7pt} \end{displayquote}

Notice that the above problems are restricted to models that are either identical, or significantly different. `Tolerant' versions (separating small changes from large) are not the focus for us (although we discuss this setting for a special case in \S\ref{sec:testing_deletions}). 
The main question of interest is: \emph{For what classes of Ising models is the sample complexity of the above inference problems significantly smaller than that of recovering the underlying graph directly?}

{\bf Contribution.} We prove the following surprising fact: up to relatively large values of $s,$ the sample complexities of $\gofbb$ and $\eofbb$ are \emph{not} appreciably separated from that of structure learning ($\SL$). Our bound is surprising in light of the fact that prior works \cite{liu2014direct, liu2017, FazBan, kim2019two, cai2019differential} propose algorithms for $\gofbb$ and $\eofbb$, and claim recovery of \emph{sparse} changes is possible with sample complexity much smaller than $\SL$. Concretely, for models with $p$ nodes, degrees bounded by $d,$ and non-zero edge weights satisfying $\alpha \le |\theta_{ij}| \le \beta$ (see \S\ref{sec:def}), the sample complexity of $\SL$ is bounded as $O(e^{2\beta d}\alpha^{-2} \log p)$. We show that if $s \ll \sqrt{p},$ then the sample complexity of $\gofbb$ is at least $e^{2\beta d - O(\log(d))} \alpha^{-2} \log p,$ and that if $s \ll p,$ then the sample complexity of $\eofbb$ has the same lower bound. We further show that the same effect occurs in the restricted setting of detecting edge deletions in forest-structured Ising models, and, to some extent, in detecting edge deletions in high-temperature ferromagnets. In the case of forests, we tightly characterise this behaviour of $\gofbb$, showing that for $s\ll \sqrt{p},$ $\gofbb$ has sample complexity comparable to $\SL$ of forests, while for $s\gg \sqrt{p},$ it is vanishingly small relative to $\SL$. {For high-temperature ferromagnets, we show that detecting changes is easier than $\SL$ if $s \gg \sqrt{pd},$ while this does not occur if $s \ll \sqrt{pd}$. These are the first structural testing results for edge edits in natural classes of Ising models that show a clear separation from $\SL$ in sample complexity.} 

\emph{Technical Novelty.} The lower bounds are shown by constructing explicit and flexible obstructions, utilising Le Cam's method and $\chi^2$-based Fano bounds. The combinatorial challenges arising in directly showing obstructions on large graphs are avoided by constructing obstructions with well-controlled $\chi^2$-divergence on small graphs, and then {\it lifting} these to $p$ nodes via tensorisation in a process that efficiently deals with combinatorial terms. The main challenge is obtaining precise control on the $\chi^2$-divergence between graphs based on cliques, which is attained by an elementary but careful analysis that exploits the symmetries inherent in Ising models on cliques. The most striking instance of this is the `Emmentaler clique' (Fig.~\ref{fig:emmen_main}), which is constructed by removing $\Theta(d^2)$ edges from a $d$-clique in a structured way. Despite this large edit, we show that it is exponentially hard (in low temperatures) to distinguish this clique with large holes from a full clique.

\subsection{Related Work}

\noindent \textbf{Statistical Divergence Based Testing.} Related to our problem, but different from our setup, $\gofbb$ of Ising models has been studied under various statistical metrics such as the symmetrised KL divergence \cite{daskalakis2019testing} and total variation \cite{bezakova_lowbd}. More refined results and extensions have appeared in \cite{gheissari2018concentration, daskalakis2017concentration,canonne2017testing,acharya-causal}. These are tests that certify whether or not a particular statistical distance between two distribution is larger than some threshold. In contrast, our focus is on \emph{structural} testing and estimation, namely, whether or not the change in the network is a result of edge-deletions or edge-additions. As such, statistically-based $\gofbb$ tests do not have a direct bearing on structural testing. Divergences can be large in structurally irrelevant ways, e.g., if a few isolated nodes in a large graph become strongly interacting, a large KL divergence is induced, but this is not a significant change in the network on the whole (Also see \S\ref{appx:stat_div}). In light of applications which demand structure testing as a means to interpret phenomena, and this misalignment of goals, testing in the parameter space is compelling, and testing the network is the simplest instance of this.

\noindent \textbf{Sparse-Recovery-Based Structural Testing Methods.} 
More directly related to our work, are those that are based on direct change estimation ({\it DCE}) \cite{FazBan, liu2014direct, liu2017, liu2017learning, kim2019two}, which attempt to directly characterize the difference of parameters $\delta^* = \theta_P - \theta_Q$ by leveraging sparsity of $\delta^*$. These works leverage the `KL Importance Estimation Procedure' (KLIEP), the key insight of which is that the log-likelihood ratios can be written in a form that is suggestive of expressions from sparse-pattern recovery methods, to define the empirical loss function \[ \mathcal{L}(\delta) = -\langle \delta, \hat{\mathbb{E}}_Q[XX^T]\rangle + \log \hat{\mathbb{E}}_P[ \exp{X^T \delta X}],\] where $\hat{\mathbb{E}}$ denotes an empirical mean, and $\delta$ is sparse. The second term, which is the only non-linear term, is reminiscent of normalization factors in graphical models. In this context, it is useful to recall the key ideas from high-dimensional sparse estimation theory (see ~\cite{negahban2012unified}), which has served as a powerful generic tool. At a high-level, these results show that for a loss function $\mathcal{L}(\delta)$ paired with a decomposable regulariser (such as an $\ell_1$ norm on $\delta$), if the loss function satisfies restricted strong convexity, namely, strong convexity only in a suitable descent error set, as characterised by the regulariser and  the optimal value $\delta^*$, minimising the penalised empirical loss leads to a non-trivial estimation error bound. Leveraging these concepts of high-dimensional estimation, and exploiting sparsity, the sparse DCE works show that testing can be done in $O(\mathrm{poly}(s) \log p)$ samples (for any $P,Q$!), which is further much smaller than the number needed for $\SL$, a result which contradicts bounds we derive in this paper. The situation warrants further discussion.

From a technical perspective, the sample complexity gains of these methods arise from assuming law-dependent quantities to be constants. For example, \cite{liu2014direct, liu2017} require that for $\|u\| \le \|\delta^*\|, \nabla^2\mathcal{L}(\delta^* + u) \preccurlyeq \lambda_1I,$ and that for $S$ the support of $\delta^*,$ the submatrix $(\nabla^2 \mathcal{L}(\delta^*) )_{S,S} \succcurlyeq \lambda_2 I$, where $\lambda_1, \lambda_2$ are constants independent of $P,Q$. \cite{FazBan} removes the second condition, and shows that $\mathcal{L}$ has the  $\lambda_2$-RSC property, where $\lambda_2$ is claimed to be independent of $P,Q$. In each case, sample costs increase with $\lambda_1$ and $\lambda_2^{-1}$. However, the assertion that $\lambda_1, \lambda_2$ are independent of $(P,Q)$ cannot hold in general -- the only non-linear part in $\mathcal{L}$ is $\log \hat{\mathbb{E}}_P[ \exp{ X^T \delta X} ],$ which clearly depends on $P$! This dependence also occurs if $P$ is known. Thus, the `constants' $\lambda_1, \lambda_2$ are affected by the properties of $P$. More generically, the efficacy of sparse recovery techniques is questionable in this scenario. Since the data is essentially distinct across samples, and internally dependent, and since the sparse changes, $\delta^*,$ and the underlying distributions interact, it is unclear if meaningful notions of design matrix that allow testing with sub-recovery sample costs can be developed. 

Nevertheless, it is an interesting question to understand what additional assumptions on $P,Q$ or topological restrictions are useful in terms of benefiting from sparsity. Our results suggest that these conditions are stronger than typical incoherence conditions such as high temperatures, and further that the topological restrictions demand more than just `simplicity' of the graphs.

\noindent \textbf{Other Methods.}\cite{cai2019differential} propose a method, whereby the parameters $\theta_P$ and $\theta_Q$ are only crudely estimated, and then tests using the biggest (normalised) deviations in the estimates as a statistic. The claims made in this paper are more modest, and do not show sample complexity below $n_{\sll}$. We point out, however, that $d$-dependent terms are treated as constants in this as well.

Much of the structural testing work studies Gaussian GMs instead of Ising (see the recent survey \cite{shojaie_survey}). We do not discuss these, but encourage the same careful examination of their assumptions. 

\noindent \textbf{Other Information-Theoretic Approaches.} We adopted a similar information-theoretic viewpoint in our earlier work \cite{gangrade2017lower, gangrade2018two}. Of these, the former only considers the restricted case of $s = 1$ (very sparse changes), and the bounds in the latter are very inefficient. As such, the present paper is a significant extension and generalization of this perspective. Our bounds further improve the approximate recovery lower bounds of \cite{scarlett2016difficulty}.

\noindent \textbf{Structural Testing Extensions.} A number of structural testing problems other than $\gofbb$ have been pursued. For instance, \cite{bresler-nagaraj} tests if the model is mean field or supported on a structured graph (sparse, etc.), \cite{bresler2019} tests mean-field models against those on an expander, \cite{cao2018high} tests independence against presence of structure in high temperatures, \cite{neykov2019property} tests combinatorial properties of the underlying graph such as whether it has cycles, or the largest clique it contains (also see \S\ref{appx:prop_test}).

\section{Problem Definitions and Notation} \label{sec:def}

The zero external field Ising Model specifies a law on a $p$-dimensional random vector $X = (X_1, \dots, X_p) \in \{ \pm 1\},$ parametrised by a symmetric matrix $\theta$ with $0$ diagonal, of the form 
\[ P_\theta(X = x) = \frac{\exp{ \sum_{i<j} \theta_{ij} x_i x_j} }{Z(\theta)},\] where $Z(\theta)$ is called the partition function. Notice that given $X_j$ for all $j \in \partial i:= \{ j: \theta_{ij} \neq 0\},$ $X_i$ is conditionally independent of $X_{[1:p] - \{i\} -\partial i}.$ Thus, the $\theta$ determine the local interactions of the model. With this intuition, one defines a simple, undirected graph $G(P_\theta) = ([1:p], E(P_\theta))$ with $E(P_\theta) = \{ (i,j): \theta_{ij} \neq 0\}.$ This graph is called the \emph{Markov network structure} of the Ising model, and $\theta$ can serves as a weighted adjacency matrix of $G(P_\theta)$. We often describe models by an unweighted graph, keeping weights implicit until required. 

The model above can display very rich behaviour as $\theta$ changes, and this strongly affects all inference problems on Ising models. With this in mind, we make two explicit parametrisations to help us track how $\theta$ affects the sample complexity of various inference problems. The first of these is degree control - we assume that the degree of every node is $G(P),G(Q)$ is at most $d$. The second is weight control - we assume that if $\theta_{ij} \neq 0,$ then $\alpha \le |\theta_{ij}| \le \beta$.

These are natural conditions: small weights are naturally difficult to detect, while large weights mask the nearby small-weight edges; degree control further sets up a local sparsity that tempers network effects in the models. The class of laws so obtained is denoted $\mathcal{I}_{d}(\alpha, \beta)$. We will usually work with a subclass $\mathcal{I} \subset \mathcal{I}_{d}$ which has \emph{unique network structures} (i.e., for $P, Q \in \ii, G(P) \neq G(Q))$. Note that we do not restrict $\alpha, \beta,d$ to have a particular behaviour - these are instead used as parametrisation to study how weights and degree affects sample complexity. In particular, they may vary with $p$ and each other. We do demand that $d \le p^{1-c}$ for some constant $c > 0$, and that $p$ is large ($\gg 1$).

We let $\mathcal{G}$ be the set of all graphs on $p$ nodes, and $\mathcal{G}_d \subset \mathcal{G}$ be those with degree at most $d.$ The symmetric difference of two graphs $G, H$ is denoted $G \triangle H,$ which is a graph with edge set consisting of those edges that appear in exactly one of $G$ and $H$.

Lastly, we say that two Ising models are \emph{$s$-separated} if their networks differ in at least $s$ edges. The `anti-ball' \( A_s(P) := \{Q \in \mathcal{I} : |G(Q) \triangle G(P)| \ge s\} \) is the set of $Q \in \ii$ $s$-separated from $P$.

\subsection{Problem Definitions}

Below we define three structural inference problems: goodness-of-fit testing, error-of-fit identification, and approximate structure learning. 

\noindent \textbf{Goodness-of-Fit Testing} Given $P$ and the dataset $X^n \sim Q^{\otimes n}$ where $Q \in \{P\} \cup A_s(P)$, we wish to distinguish between the case where the model is unchanged, $Q = P$, and the case where the network structure of the model differs in at least $s$ edges, $Q \in A_s(P).$ A goodness-of-fit test is a map $\Psi^{\gof}: \mathcal{I} \times \mathcal{X}^n \to \{0,1\}.$ The $n$-sample risk is defined as \[R^{\gof} (n, s, \mathcal{I}) :=  \adjustlimits \inf_{\Psi^{\gof}} \sup_{P \in \mathcal{I}} \left\{P^{\otimes n} (\Psi^{\gof}(P,X^n) = 1) + \sup_{Q \in A_s(P)} Q^{\otimes n}(\Psi^{\gof}(P,X^n) = 0) \right\}. \] 
\noindent \textbf{Error-of-Fit Recovery} Given $P$ and the dataset $X^n \sim Q^{\otimes n}$ where $Q \in \{P\} \cup A_s(P)$ we wish to identify where the structures of $P$ and $Q$ differ, if they do. 
    The error-of-fit learner is a graph-valued map $\Psi^{\eof}: \mathcal{I} \times \mathcal{X}^n \to \mathcal{G}.$ The $n$-sample risk is defined as \[R^{\eof} (n, s, \mathcal{I}) := \adjustlimits \inf_{\Psi^{\eof}} \sup_{P \in \mathcal{I}} \sup_{Q \in \{P\} \cup A_s(P)} Q^{\otimes n}\left( \left| \Psi^{\eof}(P,X^n) \triangle \left( G(P) \triangle G(Q)\right) \right| \ge (s-1)/2 \right). \] 
    In words, $\Psi^{\eof}$ attempts to recover $G(P) \triangle G(Q),$ and the risk penalises answers that get more than $(s-1)/2$ of the edges of this difference wrong. This problem is very similar to the following.
    
\noindent \textbf{s-Approximate Structure Learning} Given the dataset $X^n \sim Q^{\otimes n}$ we wish to determine the network structure of $Q$, with at most $s$ errors in the recovered structure. A structure learner is a graph-valued map $\Psi^{\textrm{SL}} : \mathcal{X}^n \to \mathcal{G},$ and the risk of structure learning is \[ R^{\textrm{SL}}(n, s, \mathcal{I}) :=  \adjustlimits \inf_{\Psi^{\textrm{SL}}} \sup_{Q \in \mathcal{I}} Q^{\otimes n}( | \Psi^{\textrm{SL}}(X^n) \triangle G(P)| \ge s ). \]
The sample complexity of the above problems is defined as the smallest $n$ necessary for the corresponding risk to be bounded above by $1/4,$ i.e. \[ n_{\gof} (s, \mathcal{I}) := \inf\{ n : R^{\gof} (n,s,\mathcal{I}) \le 1/4 \},\] and similarly $n_{\eof}$ and $n_{\sll}$ but with the risk lower bound of $1/8$.\footnote{$1/4$ is convenient for bounds for $\gofbb$, but any risk smaller than $1$ is of interest, and can be boosted to arbitrary accuracy by repeating trials and majority. For $\eofbb, \SL$ we use $1/8$ for ease of showing Prop.~\ref{prop:sample_comp_order}.}

The above problems are listed in increasing order of difficulty, in that methods for $\SL$ yield methods for $\eofbb$, which in turn solve $\gofbb$. This is captured by the following statement, proved in \S\ref{appx:sample_comp_order_pf}. \begin{myprop}\label{prop:sample_comp_order} \(n_{\mathrm{SL}}( (s-1)/2, \mathcal{I}) \ge n_{\eof}( s, \mathcal{I}) \ge n_{\gof}( s, \mathcal{I}). \) \end{myprop}


    

Our main point of comparison with the literature on $\SL$ is the following result, which (mildly) extends \cite[Thm 3a)]{SanWai} due to Santhanam \& Wainwright.  We leave the proof of this to Appx.~\ref{appx:pf_of_sl_upper}.

\begin{myth}\label{thm:sl_upper} If $\mathcal{I}\subset \ii_d(\alpha, \beta)$ has unique network structures, then for $s \le pd/2, \exists C \le 64$ such that 
\[ n_{\sll}(s, \ii) \le C \frac{d e^{2\beta d}}{ \sinh^2 (\alpha/4) } \left( 1 +  \log \frac{p^2  }{2s} + O(1/s) \right).\]
\end{myth}

\section[GoF and EoF Lower Bounds for degree-bounded Ising models]{Lower Bounds for $\gofbb$ and $\eofbb$ over $\ii_d(\alpha, \beta)$}\label{sec:Main_lowbs}

This section states our results, and discusses our proof strategy, but proofs for all statements are left to \S\ref{appx:lowbs}. The bound are generally stated in a weaker form to ease presentation, but the complete results are described in \S\ref{appx:lowbs}. We begin by stating lower bounds for the case of $s = O(p).$ Throughout $500 > K > 1$ is a constant independent of all parameters. 
\begin{myth}\label{thm:lb_small_s} If $20 \le d \le s \le \nicefrac{p}{K},$ then there exists a $C>0$ independent of $(s,p,d,\alpha, \beta)$ such that \begin{align*}
    n_{\gof}(s, \ii) &\ge C\max \left\{ \frac{e^{2\beta}}{\tanh^2\alpha} , \frac{e^{2\beta (d-3)}}{d^2 \min(1, \alpha^2 d^4)} \right\} \log\Big( 1 + C\frac{p}{s^2}\Big)\\
    n_{\eof}(s, \ii) &\ge C\max \left\{ \frac{e^{2\beta}}{\tanh^2\alpha} , \frac{e^{2\beta (d-3)}}{d^2 \min(1, \alpha^2 d^4)} \right\} \log \left( C\frac{p}{s} \right)
\end{align*}\end{myth}

This statement is enough to make our generic point - for small $s$ (i.e., if $s \le p^{\nicefrac{1}{2}-c}$ in $\gofbb$ and if $s \le p^{1-c}$ in $\eofbb$), the above bounds are uniformly within a $O(\mathrm{poly}(d))$ factor of the the upper bound on $n_{\sll}$ in Theorem \ref{thm:sl_upper}. Notice also that the $\max$-terms are uniformly $\tilde{\Omega}(d^2)$ in the above - if $\beta d \ge 2\log d,$ then the second term in the max is $\Omega(d^2),$ while if smaller, the first term is $\Omega( (\nicefrac{d}{\log d})^2)$ because $\alpha \le \beta$. Thus, over $\ii_d$, the best possible sample complexity of $\gofbb$ and $\eofbb$ scales as $\tilde{\Omega}(d^2 \log p),$ and in particular cannot be generally $d$-independent.

Of course, graphs in $\mathcal{G}_d$ have upto $\sim pd$ edges, and so many more changes can be made. Towards this, we provide the following bound for $\gofbb$. A similar result for $\eofbb$ is discussed in \S\ref{appx:lowbs}.
\begin{myth}\label{thm:lb_large_s}
If for some $\zeta >0, s \le pd^{1-\zeta}/K,$ and $d \ge 10$, then there exists a constant $C>0$ independent of $(s,p,d,\alpha, \beta)$  such that \begin{enumerate}[leftmargin = .5in, parsep = 3pt, topsep = 2pt] 
    \item If $\alpha d^{1-\zeta} \le 1/32$ then \(\displaystyle n_{\gof} \ge C \frac{1}{d^{2-2\zeta} \alpha^2} \log \Big(1 + C\frac{pd^{3-3\zeta}}{s^2}\Big).\)
    \item If $\beta d \ge 4\log(d-4)$ then \( \displaystyle n_{\gof} \ge C\frac{e^{2\beta d(1 - d^{-\zeta})}}{d^2 \min(1, \alpha^2 d^4)} \log \Big(1 + C\frac{pd^{2-3\zeta}}{s^2}\Big). \)
\end{enumerate}
\end{myth}

Thm.~\ref{thm:lb_large_s} leaves a (small) gap, since as $\zeta \to 0,$ $\alpha d^{1-\zeta} \le 1$ and $\beta d \ge 4\log(d)$ do not completely cover all possibilities. Barring this gap, we again notice that for $s\ll \sqrt{pd^{1-\zeta}},$ $n_{\gof}$ is separated from $n_{\sll}$ by at most a $\mathrm{poly}(d)$ factor. The first part of the above statement is derived using results of \cite{cao2018high}. For the limiting case of $\zeta = 0,$ i.e. when $s$ is linear in $pd$, we recover similar bounds, but with the distinction that the $2\beta d$ in the exponent is replaced by a $\beta d$. See \S\ref{appx:lowbs}.

Finally, since often the interest in DCE lies in \emph{very sparse} changes, we present the following - \begin{myth}\label{thm:very_small_s} 
    If $s \le d$, then there exists a $C > 0$ independent of $(s,p,d,\alpha, \beta)$ such that \begin{align*} n_{\gof}(s, \ii) &\ge C\max \left\{ \frac{e^{2\beta}}{\tanh^2\alpha} , \frac{e^{2\beta (d- 1- 2\sqrt{s})}}{d^6\sinh^2(\alpha \sqrt{s})} \right\} \log\Big( 1 + C\left(\frac{p}{s^2} \wedge \frac{p}{d}\right) \Big)\\
    n_{\eof}(s, \ii) &\ge C\max \left\{ \frac{e^{2\beta}}{\tanh^2\alpha} , \frac{e^{2\beta (d- 1- 2\sqrt{s})}}{d^6\sinh^2(\alpha \sqrt{s})} \right\} \log\left( C\frac{p}{d}\right) \end{align*}
\end{myth}


\noindent \textbf{Structure of the Bounds} Each of the bounds above can be viewed as of the form $(\SNR)^{-1} \log ( 1 + f(p,s,d) )$, where we call the premultiplying terms $\SNR$ since they naturally capture how much signal about the network structure of a law relative to its fluctuations is present in the samples. This $\SNR$ term in Thms.~\ref{thm:lb_small_s} and~\ref{thm:very_small_s} is developed as a max of two terms. The first of these is effective in the high temperature regime (where $\beta d$ is small), while the second takes over in the low temperature regime of large $\beta d$. Similarly, the first and second parts of Thm.~\ref{thm:lb_large_s} are high and low temperature settings, respectively, and have different $\SNR$ terms. The $\SNR$ in all of the above is within a $\mathrm{poly}(d)$ factor of the corresponding term in the upper bound for $n_{\sll}$.

The term $f(p,d,s)$ thus captures the hardness of testing/error localisation. For $\eofbb$, as long as $s$ is small, this term takes the form $p^c$ for some $c$. Thus, generically, localising sparse changes is nearly as hard as approximate recovery. This is to be expected from the form of the $\eofbb$ problem itself. More interestingly, for $\gofbb$, these take the form $ p d^{c}/s^2$. When $s \ll \sqrt{p d^c},$ this continues to look polynomial in $p$, and thus $\gofbb$ is as hard as recovery. On the other hand, for $s$ much larger than this, $f$ becomes $o(1)$ as $p$ grows, and so $\log(1 + f) \approx f$ itself and the resulting bounds look like $(\SNR)^{-1} pd^c/s^2$. In the setting of low temperatures with non-trivially large degree, these can still be super-polynomial in $p$, but relative to $n_{\sl}$ they are essentially vanishing.

Notice that in high temperatures ($\beta d \le 1$), the bounds of Thms.~\ref{thm:lb_small_s} and \ref{thm:very_small_s} are only $O(d)$ away from $n_{\sll}$ for small $s$, fortifying our claim that $\gofbb$ and $\eofbb$ are not separated from $\SL$ in this setting.

\noindent \textbf{Counterpoint to Sparse DCE efforts} The above bounds, especially Thm.~\ref{thm:very_small_s}, show that for small $s$ $\gofbb$ and $\eofbb$ are as hard as recovery of $G(Q)$ itself. A possible critique of these bounds when considering DCE is that the DCE schemes demand that the changes are smaller than $s,$ while our formulations only require the changes to have size at least $s$. To counter this, we point out that the constructions for Thms.~\ref{thm:lb_small_s}, \ref{thm:lb_large_s}, and \ref{thm:very_small_s} make at most $2s$ changes when computing bounds for any $s$ (in fact, smaller edits lead to stronger bounds). Thus, the above results catergorically contradict the claim that a generic $O(\mathrm{poly}(s) \log p)$ bound that is $d$ independent and much smaller than $n_{\sll}$ can hold for DCE methods on $\ii_d$. Since $\alpha, \beta, d$ are only parameters, and are not restricted in any way, this shows that the assumptions made for DCE cannot be reduced to some conditions on only $\alpha, \beta, d$, and further topological conditions must be implicit. In particular, these are stronger than typical incoherence conditions such as Dobrushin/high-temperature ($\beta d < 1$;e.g.,\cite{daskalakis2017concentration, gheissari2018concentration}).
 
\subsection{Proof Technique}\label{sec:pf_tech}
The above bounds are shown via Le Cam's method with control on the $\chi^2$-divergence of a mixture of alternatives for $\gofbb$, and via a Fano-type inequality for the $\chi^2$-divergence, due to Guntuboyina \cite{guntuboyina2011lower} for $\eofbb$. These methods allow us to argue the bounds above by explicit construction of distributions that are hard to distinguish. We briefly describe the technique used for $\gofbb$ below.
\begin{defi}
A $s$-change ensemble in $\mathcal{I}$ is a distribution $P$ and a set of distributions $\mathcal{Q}$, denoted $(P,\mathcal{Q}),$ such that $P \in \mathcal{I}, Q \subseteq \mathcal{I},$ and for every $Q \in \mathcal{Q},$ it holds that $|G(P) \triangle G(Q)| \ge s.$
\end{defi} 
Each of the testing bounds we show will involve a mixture of $n$-fold distributions over a class of distributions. For succinctness, we define the following symbol for a set of distibutions $\mathcal{Q}$ \[ \langle \mathcal{Q}^{\otimes n} \rangle := \frac{1}{|\mathcal{Q}|} \sum_{Q \in \mathcal{Q}} Q^{\otimes n}.\]
Le Cam's method (see e.g.~\cite{Yu1997, ingster_suslina}) shows that if $(P, \mathcal{Q})$ is a $s$-change ensemble in $\ii$, then \[R^{\gof}(n , s, \mathcal{I}) \ge 1 -\sqrt{\frac{1}{2} \log( 1 + \chi^2( \langle \mathcal{Q}^{ \otimes n} \rangle \| P^{\otimes n}) )}. \] As a consequence, if we find a change ensemble and an $n$ such that $1 + \chi^2(\qmix\|\mathcal{P}^{\otimes n}) \le 3,$ then we would have established that $n_{\gof}(s,\ii) \ge n.$ So, our task is set up as constructing appropriate change ensembles for which the $\chi^2$-divergence is controllable. 

Directly constructing such ensembles is difficult, essentially due to the combinatorial athletics involved in controlling the divergence. We instead proceed by constructing a pair of separated distributions $(P_0,Q_0)$ on a small number of nodes, and then `lifting' the resulting bounds to the $p$ nodes via tensorisation - $P$ is contructed by collecting disconnected copies of $P_0$, while $\mathcal{Q}$ is constructed by changing some of the $P_0$ copies to $Q_0$. The process is summarised as follows.

\begin{mylem}{(Lifting)}\label{lem:lifting} Let $P_0$ and $Q_0$ be Ising models with degree $\le d$ on $\nu \le p/2$ nodes such that $|G(P_0) \triangle G(Q_0)| = \sigma,$ and $\chi^2(Q_0^{\otimes n}\|P_0^{\otimes n}) \le a_n.$ Let $m := \lfloor p/\nu \rfloor.$ For $t < m/16e,$ there exists a $t\sigma$-change ensemble $(P,\mathcal{Q})$ over $p$ nodes such that $|\mathcal{Q}| = \binom{m}{t}$ and \[1 + \chi^2(\langle \mathcal{Q}^{ \otimes n} \rangle \| P^{\otimes n} ) \le \exp{ \frac{t^2}{m} a_n}. \] 
\end{mylem}\vspace{-6pt}
A similar argument is used for the $\eofbb$ bounds, along with a similar lifting trick, discussed in \S\ref{appx:lowbs}. Due to the tensorisation of the $\chi^2$-divergence, we obtain results of the form $a_n \le (1 + \kappa)^n -1,$ where $\kappa$ depends on $(P_0, Q_0)$ but not $n$. Plugging this into the above with $t = \lceil s/\sigma\rceil$ yields \[   n_{\gof}(s, \ii) \ge  \frac{1}{\log(1 + \kappa)} \log \left(1 + \frac{p\sigma^2}{8\nu s^2}\right).\] Notice that this $\kappa$ is an SNR term, while $\log(1 + p\sigma^2/8\nu s^2)$ captures combinatorial effects.

\begin{wrapfigure}[14]{r}{0.44\textwidth}
    \centering \vspace{-20pt}
    \includegraphics[width = .4\textwidth]{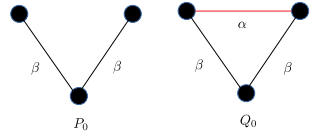} \vspace{-45pt}
    \caption{Graphs used to construct high-temperature obstructions. Labels indicate edge-weight, and the red edge is added in $Q_0$.}
    \label{fig:triangle_main}
\end{wrapfigure}
The procedure thus calls for strong $\chi^2$ bounds for various choices of small graphs, or `widgets'. We use two varieties of these - the first, `star-type' widgets, are variations on a star graph. These allow direct calculations in general, and provide bounds that extend to the high-temperature regime. The second variety is the `clique-type' widgets, that are variations on a clique, and provide low-temperature obstructions. Classical Curie-Weiss analysis shows that cliques tend to `freeze' - for Ising models on a $k$-clique with uniform weight $\lambda,$ the probability mass concentrates on the set $\{ (1)^{\otimes k}, (-1)^{\otimes k}\}$ w.p.~roughly $1 - e^{- \Theta(\lambda k)}.$ The clique-type obstructions implicitly argue that this effect is very robust.

\begin{wrapfigure}[11]{r}{0.44\textwidth}
    \centering \vspace{-65pt}
    \includegraphics[width = .2\textwidth]{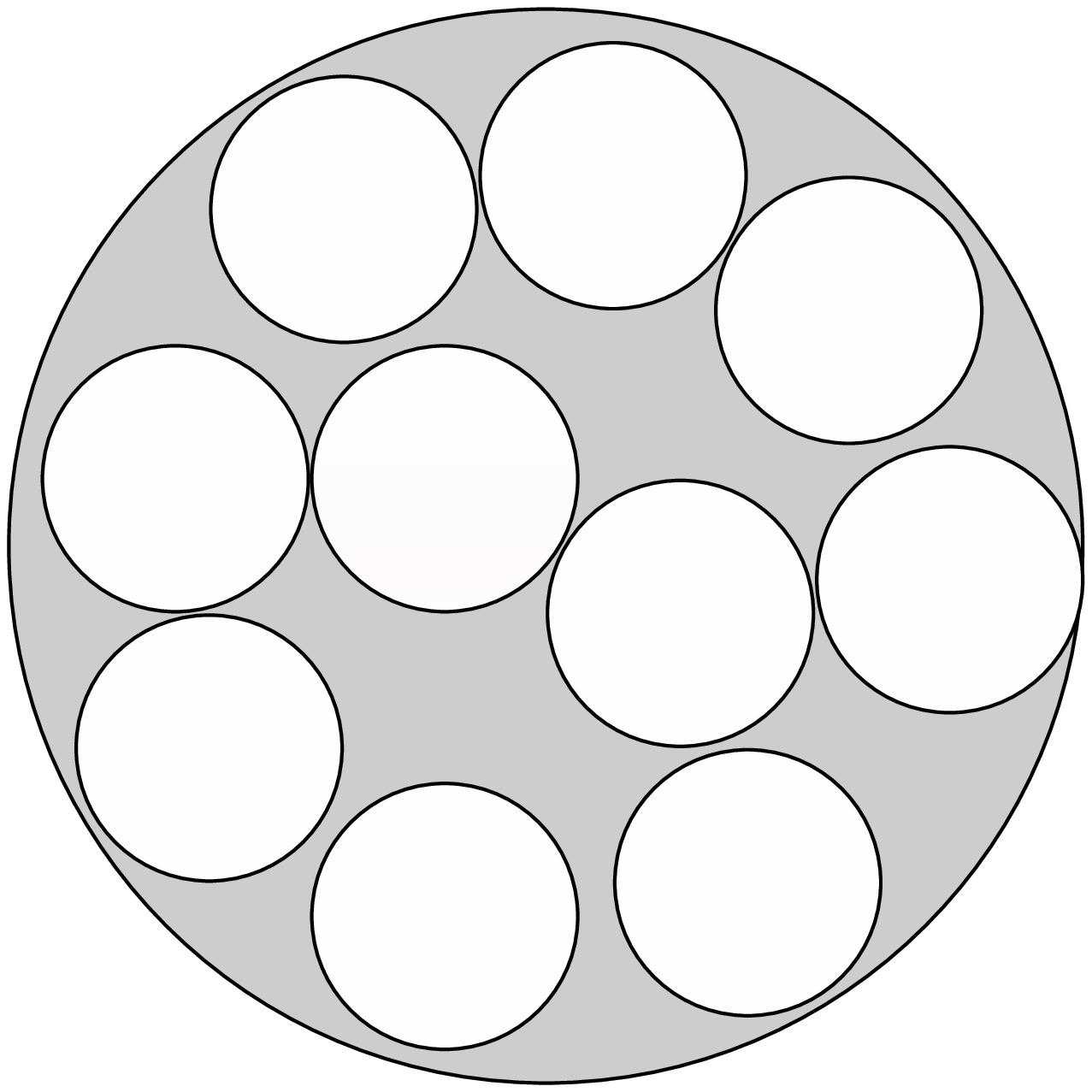}
    ~
    \includegraphics[width = .2\textwidth]{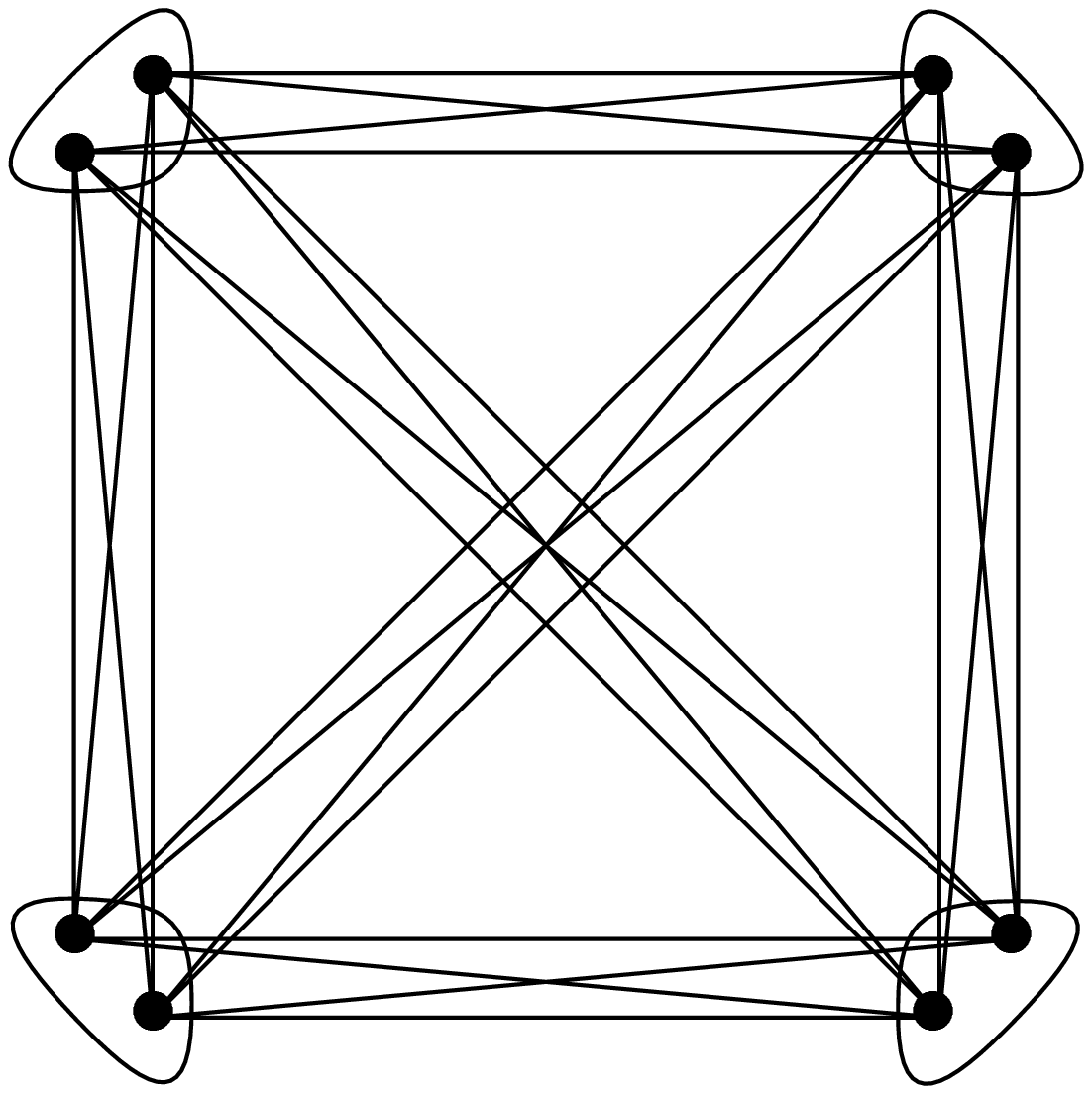} \vspace{-4pt}
    \caption{Two views of Emmentaler cliques. Left: the base clique is the large grey circle, uncoloured circles represent the groups with no edges within (this is $d,\ell \gg 1, \nicefrac{d+1}{\ell+1} = 10$); Right: Emmentaler as the graph $K_{\ell + 1, \ell + 1, \dots, \ell + 1}$ ($d = 7, \ell = 1$).}
    \label{fig:emmen_main}
\end{wrapfigure}
The particular graphs used to argue the high temperature bounds in Thms.~\ref{thm:lb_small_s},\ref{thm:very_small_s} are a `V' versus a triangle as seen in Fig.~\ref{fig:triangle_main}, while in Thm.~\ref{thm:lb_large_s} the empty graph is compared to a $d^{1- \zeta}$-clique. The low temperature obstructions of Thms.~\ref{thm:lb_small_s},\ref{thm:lb_large_s} compare a full $d+1$-clique as $P_0$ to an `Emmentaler' clique (Fig.~\ref{fig:emmen_main}). These are constructed by dividing the $d+1$ nodes into groups of size $\ell+1$, and removing the $\ell+1$-subclique within each group. The graph can thus be seen either as a clique with many large `holes' - corresponding to the deleted subcliques - which inspires the name, or as the complete $\nicefrac{d+1}{\ell+1}$-partite graph on $d+1$ nodes. Notice that in the Emmentaler clique we have deleted $ \approx \nicefrac{d\ell}{2}$ edges. We will show in \S\ref{appx:widgets} that this is still hard to distinguish from the full clique for $\ell \sim d/10$ - a deletion of $\Omega(d^2)$ edges! \vspace{2pt}\\
\noindent \textbf{On Tightness} Prima facie the above bounds suggest that one may find sample efficient schemes in, say, $\gofbb$ for $s \gg \sqrt{pd}.$ However, it is our opinion that these bounds are actually loose. Particularly, while the $\SNR$ terms are relatively tight, the behaviour of $f(p,d,s)$ is not. To justify this opinion, consider the setting of forest-structured graphs. By the same techniques, we show a similar bound with $f = p/s^2$ for $\gofbb$ in forests in \S\ref{sec:forests} - this is the best possible by the methods employed. For $s \gg \sqrt{p},$ the resulting overall lower bound is the trivial $n \ge 1$ unless $\alpha \le (p/s^2)^{1/2}$. On the other hand, \cite[Thm.~14]{daskalakis2019testing} can be adapted to show a lower bound for forests of $\Omega(\alpha^{-2}\wedge \alpha^{-4}/p)$ for the particular case of $s = p/2$, which is non-trivial for all $\alpha \lesssim p^{-1/4}$. Our results trivialise for $\alpha \gtrsim p^{-1/2}$ for this case, demonstrating looseness.\vspace{2pt}\\
The reason for this gap lies in the lifting trick used to show these bounds. The tensorisation step involved in this constricts the set of `alternates' one can consider, thus diminishing $f$. More concretely - there are about $\nicefrac{p^2-pd}{2}$ potential ways to add an edge (and $O(pd)$ to delete an edge), while the lifting process as implemented here restricts these to at most $O(pd)$. It is important to recognize this lossiness, particularly since \emph{most} lower bounds, for both testing and recovery, proceed via a similar trick, e.g.~\cite{SanWai, tandon2014information, scarlett2016difficulty, gangrade2017lower, neykov2019property, cao2018high}. \cite[Thm.~14]{daskalakis2019testing} is the only exception we know of. We conjecture that for $\gofbb$ in $\ii_d,$ $f$ should behave like $p^2/s^2,$ while for $\eofbb$, it should behave like $p^2/s$. Note that for $\gofbb$, since $s$ can be as big as $pd,$ this indicates that one should look for sample-efficient achievability schema in the setting of $s > pd^{c}$.\vspace{2pt}\\ 
However, for simpler settings this technique \emph{can} recover tight bounds. For instance, \S\ref{sec:forests} presents a matching upper bound for testing of edge-\emph{deletion} in a forest. Notice that in this case there are only $O(p)$ possible ways to edit. This raises the further question of if the same effect extends to $\ii_d$, i.e., can deletion of edges in $\ii_d$ be tested with $O(1 \vee e^{2\beta d}\alpha^{-2} (pd/s^2) )$ samples when $s \gg \sqrt{pd}$?  {\S\ref{sec:delete_high_temp} offers initial results in this direction in the high temperature regime.}

\section{Testing Edge Deletions}\label{sec:testing_deletions}

Continuing on the theme that concluded our discussion of the tightness of our lower bounds, we study the testing of edge deletions in two classes of Ising models - forests, and high-temperature ferromagnets - with the aim demonstrating natural settings in which the sample complexity of $\gofbb$ testing of Ising models is provably separated from that of the corresponding recovery problem.

In the deletion setting, we consider the same problems as in \S\ref{sec:def}, but with the additional constraint that if $Q \neq P,$ then $G(Q) \subset G(P)$, that is, the network structures of alternates can be obtained by dropping some edges in that of the null. For a class of Ising models $\mathcal{J},$ we thus define \[ R^{\gof,\delete}(n,s,\mathcal{J}) = \adjustlimits \inf_{\Psi} \sup_{P \in \mathcal{J}}  P^{\otimes n}(\Psi(P,X^n) = 1) + \sup_{\substack{ Q \in A_s(P) \cap \mathcal{J} \\ G(Q) \subset G(P)}} Q^{\otimes n}(\Psi(P,X^n) = 1), \] and, analogously define $R_{\eof, \delete}$, and the sample complexities $n_{\gof, \delete}(s, \mathcal{J})$ and $n_{\eof, \delete}(s, \mathcal{J})$.

We will  look at testing deletions for two choices of $\mathcal{J}$ which both have uniform edge weights\begin{itemize}[wide, nosep]
    \item \textbf{Forest-Structured Models} ($\mathcal{F}(\alpha)$) are Ising models with uniform weight $\alpha$ such that their network structure is a forest (i.e., has no cycles).
    \item \textbf{High-Temperature Ferromagnets} ($\mathcal{H}_d^\eta(\alpha)$) are models with max degree at most $d$, uniform \emph{positive} edge weights $\alpha,$ and further such that there is an $\eta < 1$ such that $\alpha d \le \eta$.
\end{itemize}

We note that while our motivation for the study of the above is technical, both of these subclasses of models have been utilised in practice, and indeed are the subclasses of $\mathcal{I}_d$ that are best understood.

\subsection{Testing Deletions in Forests}\label{sec:forests}

Forest-structured Ising models are known to be tractable, and have thus long served as the first setting to explore when trying to establish achievability statements. We show a tight characterisation of the sample complexity of testing deletions in forests for large changes, and also demonstrate the separation from the corresponding $\eofbb$ (and thus also $\SL$) problem. In addition, we also show that for the restricted subclass of trees, essentially the same characterisation follows for \emph{arbitrary} changes (i.e., not just deletions), and that the methods support some amount of tolerance directly. We begin with the main result for testing deletions in forests (all proofs are in \S\ref{appx:forests}).

It is worth noting that degrees are not assumed to be explicitly bounded in this section - i.e.~the results hold even if the max degree is $p-1$ (a star graph).

\begin{myth}\label{thm:forest_testing}
    There exists a constant $C$ independent of $(s,p,\alpha)$ such that the sample complexity of $\gofbb$ testing of forest-structured Ising models against deletions is bounded as \[  n_{\gof, \delete}(s, \mathcal{F}(\alpha)) \le C \max\left\{ 1, \frac{1}{\sinh^2(\alpha)} \frac{p}{s^2}\right\}.\]
    Conversely, for $s \le \nicefrac{p}{32e},$ there exists a constant $C'$ independent of $(s,p,\alpha),$ such that \begin{align*}  n_{\gof, \delete}(s, \mathcal{F}(\alpha)) &\ge \max\left\{1 , \frac{1}{C'} \frac{1}{\sinh^2 \alpha} \log \Big(1 + \frac{p}{C's^2}\Big)\right\}, \\ n_{\eof, \delete}(s, \mathcal{F}(\alpha) ) &\ge \frac{1}{C'\sinh^2\alpha} \log \left( \frac{p}{C's} \right).\end{align*}
\end{myth}
The upper bound is constructed by using the simple global statistic $\mathscr{T}_P = \sum_{(i,j) \in G(P)} X_iX_j$, averaged across the samples. Again, the behaviour of the lower bound shifts as $s$ crosses $\sqrt{p}$ - for larger $s$, it scales as $1 \vee \sinh^{-2}(\alpha) \nicefrac{p}{s^{2}},$ while for much smaller $s$ it is $1 \vee \sinh^{-2}(\alpha)\log p$. Further, for large changes, the lower bound is matched, up to constants, by the achievability statement above. For the smaller case, the same holds in the restricted setting of $\alpha < 1$, since exact recovery in $\mathcal{F}(\alpha)$ only needs $\tanh^{-2}(\alpha) \log p$ samples (Chow-Liu algorithm, as analysed in \cite{breslerkarzand}).\footnote{While the $\alpha<1$ regime is certainly more relevant in practice, it is an open question whether for larger $\alpha,$ and for small $s$, the correct SNR behaviour is $\sinh^{-2}$ or $\tanh^{-2}$ in testing.} Finally, the $\eofbb$ lower bound (which is also tight for $\alpha < 1$, show that the sample complexity of $\gofbb$ is separated from error of fit (and thus $\SL$) for large changes.

Fig.~\ref{fig:forest_plot} illustrates Thm.~\ref{thm:forest_testing} via a simulation for testing deletions in a binary tree (for $p = 127, \alpha = 0.1$), showing excellent agreement. In particular, observe the sharp drop in samples needed at $s = 21 \approx 2\sqrt{p}$ versus at $s < \sqrt{p} \approx 11.$ We note that $\SL$-based testing fails for all $s \le 60$ for this setting even with $1500$ samples (Fig.~\ref{fig:chow-liu-error} in \S\ref{appx:exp}), which is far beyond the scale of Fig.~\ref{fig:forest_plot}. See \S\ref{appx:exp} for details.

\begin{figure}[h]
    \centering
    \includegraphics[width = .5\textwidth]{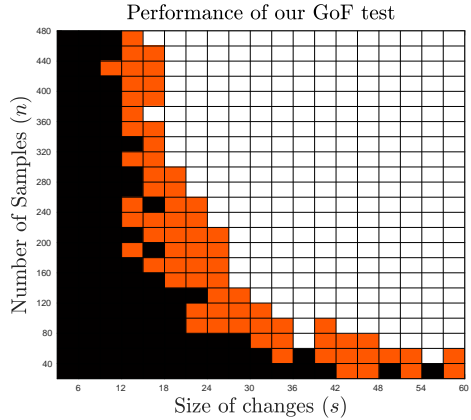}\vspace{-10pt}
    \caption{\small Testing deletions in binary trees for $p = 127, \alpha = 0.1$. Entries are coloured black if risk is $>0.35,$ white if $< 0.15,$ and orange otherwise.}
    \label{fig:forest_plot}
\end{figure}

\noindent \textbf{Testing arbitrary changes in trees} The statistic $\mathscr{T}$ is good at detecting deletions in edges, but is insensitive to edge additions, which prevents it from being effective in general for forests. However, if the forest-models $P$ and $Q$ are restricted to have the same \emph{number of edges}, then $\mathscr{T}$ should retain power, since any change of $s$ edges must delete $s/2$ edges. This, of course, naturally occurs for trees! Let $\mathcal{T}(\alpha) \subset \mathcal{F}(\alpha)$ denote tree-structured Ising models. \begin{myth}\label{thm:tree_testing}
    There exists a $C$ independent of $(p,s,\alpha)$ s.t. \[ n_{\gof}(s, \mathcal{T}(\alpha)) \le C \max\left( 1, \frac{1}{(1-\tanh(\alpha))^2 \sinh^2(\alpha)} \frac{p}{s^2} \right).\] 
    Conversely, there exists a $c$ independent of $(p,s,\alpha)$ such that \[ n_{\gof}(s,\mathcal{T}(\alpha)) \ge c\frac{1}{\tanh^2(\alpha)} \log \left( 1 + \frac{c p}{s^2}\right). \]
\end{myth}

\noindent \textbf{Tolerant Testing} The achievability results of Thm.s \ref{thm:forest_testing},\ref{thm:tree_testing} can be made `tolerant' without much effort (see \S\ref{appx:tolerant_testing}). `Tolerance' here refers to updating the task to separate models that are $\varepsilon s$-close to $P$ from those that are $s$-far from it. 

Concretely, let $\mathcal{J}$ be a class of Ising models, $s,p,n$ as before, and let $\varepsilon \in (0,1)$ be a tolerance parameter. We set up the following risks of tolerant testing of $s$ changes at tolerance $\varepsilon,$ and of tolerant testing of deletion at the same levels, as 

\begin{align*}
    R^{\gof}_{\mathrm{tol}}(n,s, \varepsilon, \mathcal{J}) &= \adjustlimits \inf_{\Psi} \sup_{P \in \mathcal{J}} \left\{ \sup_{\widetilde{P} \in A_{\varepsilon s}(P)^c \cap \mathcal{J}} \widetilde{P}^{\otimes n}(\Psi = 1) + \sup_{Q \in A_s(P) \cap \mathcal{J}} Q^{\otimes n}(\Psi = 0)\right\},\\
    R^{\gof, \delete}_{\mathrm{tol}}(n,s, \varepsilon, \mathcal{J}) &= \adjustlimits \inf_{\Psi} \sup_{P \in \mathcal{J}} \left\{ \sup_{ \substack{\widetilde{P} \in A_{\varepsilon s}(P)^c \cap \mathcal{J}\\ G(\widetilde{P})\subset G(P)}} \widetilde{P}^{\otimes n}(\Psi = 1) + \sup_{\substack{Q \in A_s(P) \cap \mathcal{J}\\ G(Q)\subset G(P)} } Q^{\otimes n}(\Psi = 0)\right\}.
\end{align*}

Analogously to \S\ref{sec:def}, the sample complexities $n_{\gof}^{\mathrm{tol}}(s, \varepsilon, \mathcal{J})$ and $n_{\gof,\delete}^{\mathrm{tol}}(s, \varepsilon, \mathcal{J})$ are the smallest $n$ required to drive the above risks below $1/4.$ Our claim in the above may be summarised as follows. 
\begin{myth}\label{thm:tolerant_tree}
    There exists a constant $C$ independent of $(s,p,\alpha, \varepsilon)$ such that \[ n_{\gof, \delete}^{\mathrm{tol}}(s, \varepsilon, \mathcal{F}(\alpha)) \le C\max \left\{ 1, \frac{1}{\sinh^2(\alpha)} \frac{p}{(1-\varepsilon)^2 s^2} , \frac{1}{(1-\varepsilon)^2 s}\right\}.\] Further, if $\varepsilon < \nicefrac{1 - \tanh(\alpha)}{2},$ then \[ n_{\gof}^{\mathrm{tol}}(s, \varepsilon, \mathcal{T}(\alpha)) \le C \max \left\{ 1, \frac{1}{\sinh^2(\alpha)} \frac{p}{(1-2\varepsilon - \tanh(\alpha))^2 s^2}, \frac{1}{(1-2\varepsilon - \tanh(\alpha))^2 s}\right\} .\]
\end{myth}

The key point for showing the above is that the mean of the statistic $\mathscr{T}$ doesn't move too much under small changes - for $\tau = \tanh(\alpha),$ changing $\varepsilon s$ edges reduces the mean of $\mathscr{T}_P$ by at most $\varepsilon s \tau$ in both cases, while changing $\ge s$ edges reduces it by at least $s\tau$ for forest deletion, and $\nicefrac{s\tau(1-\tau)}{2}$ for arbitrary changes in trees. Comparing this upper bound in the drop in the mean against the lower bound when $\ge s$ changes are made (along with the common noise scale of the problem) directly gives the above blowups in the costs of tolerant testing. This should be contrasted with statistical distance based formulations of testing, for which tolerant testing is a subtle question, and, at least in unstructured settings, requires using different divergences to define closeness and farness in order to show gains beyond learning \cite{daskalakis2018tolerant}. 

\subsection{Testing Deletions in High-Temperature Ferromagnets}\label{sec:delete_high_temp}

Testing deletions in ferromagnets is amenable due to two technical properties of the statistic $\mathscr{T}_P = \sum_{(i,j) \in G(P)} X_iX_j$. The first of these is that due to the ferromagneticity, deleting an edge can only reduce the correlations between the values that the variables take. Coupling this fact with a structural result that is derived using \cite[Lemma 6]{SanWai} yields that if $G(Q) \subset G(P)$ and $|G(P) \triangle G(Q)| \ge s,$ then $\mathbb{E}_P[\mathscr{T}_P] - \mathbb{E}_Q[\mathscr{T}_P] \gtrsim s\alpha.$ The second technical property is that bilinear functions of the variables, such as $\mathscr{T}_P,$ exhibit concentration in high-temperature Ising models. In particular, using the Hoeffding-type concentration of \cite[Ex.~2.5]{adamczak2019note}, $\mathscr{T}_P$ concentrates at the scale $O(\sqrt{pd})$ around its mean for all high-temperature ferromagnets. With means separated, and variances controlled, we can offer the following upper bound on the sample complexity, while the converse is derived using techniques of previous sections. See \S\ref{appx:high_temp_ferro} for proofs. \begin{myth}\label{thm:high_temp_ferr_del}
    There exists a constant $C_\eta$ depending only on $\eta$ and not on $(s,p,d,\alpha)$ such that \[ n_{\gof, \delete}(s\mathcal{H}_d^\eta(\alpha)) \le C_\eta \left(\frac{pd}{\alpha^2 s^2} \vee 1\right).\] Conversely, there exists a $c<1$ independent of $(s,p,d,\alpha)$ such that if $\eta \le 1/16, s \le cpd$ then\[ n_{\gof, \delete}(s, \mathcal{H}_d^\eta(\alpha)) \ge \frac{c}{\alpha^2d^2} \log\left( 1 + \frac{cpd^3}{s^2}\right) \,\,\textrm{\&}\,\, n_{\eof, \delete}(s, \mathcal{H}_d^\eta(\alpha)) \ge \frac{c}{\alpha^2d^2} \log\left( 1 + \frac{cpd}{s}\right)  \]
\end{myth}

Unlike in Thm. \ref{thm:forest_testing}, the lower bounds above are not very clean, and so our characterisation of the sample complexity is not tight. Nevertheless, we once again observe a clear separation between sample complexities of $\gofbb$ and of $\eofbb$ and a fortiori that of $\SL$. Concretely, our achievability upper bound and the $\eofbb$ lower bound show that for $s > \sqrt{pd^3},$ the sample complexity of testing deletions is far below that of structure learning in this class. Further, our testing lower bound tightly characterises the sample complexity for $s\ge \sqrt{pd^3}.$

As an aside, note that unlike in the forest setting, it is not clear if $\mathscr{T}$ is generically sensitive to edge deletions, since network effects due to cycles in a graph can bump up correlation even for deleted edges. However, we strongly suspect that a similar effect does hold in this setting, raising another open question - can testing of changes in the subclass of $\mathcal{H}^\eta_d$ with a fixed number of edges be performed with $O(\alpha^{-2} pd/s^2)$ samples for large $s$? A similar open question arises for tolerant testing, which requires us to show that small changes do not alter the mean of $\mathscr{T}$ too much.

\section{Discussion}

The paper was concerned with the structural goodness-of-fit testing problem for Ising models. We first argued that this is instrinsically motivated, and we distinguished this formulation from GoF testing under statistical measures that has been pursued in the recent literature. The main problem we studied was that of the sample complexity of GoF testing, with a refined question asking when this was significantly separated from that of structure learning. Alternatively, we can view this question as asking when testing via structure learning is suboptimal in sample costs. In addition, we considered the EoF estimation problem, which serves as a proxy for approximate structure recovery, and also aligns with the focus of the sparse DCE literature. We showed that quite generically, if the number of edge edits being tested are small, then the GoF testing and EoF testing problems are not separated from structure learning in sample complexity. This concretely rebuts the approach taken by the sparse DCE methods, and instead suggests that algorithmic work on structural testing should concentrate on large changes. In addition, we identified inefficiencies in our lower bound technique, namely that the number of changes the constructions allow is too small, which reduces the effectiveness of the lower bounds below the level we believe them to hold (in that the bounds trivialise for too small an $s$, in our opinion). In order to demonstrate that this is the only source of looseness, we demonstrated upper bounds for GoF testing in the deletion setting. This was helped by the fact that the deletion problem is much simpler than full testing, because the relevant test statistic is pretty obvious for this case, while it is unclear what statistic is appropriate to construct general tests. Along the way we controlled the sample costs of generic testing in tree structured models, and showed that the same tests easily admit some level of tolerance around the null model.

A number of questions are left open, and we point out a few here. From the perspective of lower bounds, the chief is to remove the inefficiencies in our lower bound technique. As a beginning towards this, it may be worth exploring if the methods used to show \cite[Thm.~14]{daskalakis2019testing} can be extended to deal with $s < p$ changes. In addition, we note that while the SNR terms in the lower bounds are relatively tight, there are still extraneous factors that need to be addressed. Coming around to upper bounds, the main open problem is that of constructing tests for degree bounded Ising models in the setting $s = pd^{c}$ for some $c > 0.$ Further, we ask if our bounds on testing deletions in high-temperature ferromagnets can be extended to generic ferromagnets (which would require replacing the concentration argument), or to generic changes in high-temperature ferromagnets (which would require development of new statistics that are sensitive to edge additions et c.). In addition, can the deletion result be extended to testing under the constraint the the null and alternate models have the same number of edges (analogously to how the forest deletion results extend to changes in trees), and can the deletion result be made tolerant?

\paragraph{Acknowledgements}  AG would like to thank Bodhi Vani and Anil Kag for discussions that helped with the simulations described in \S\ref{appx:exp}, on which Figure \ref{fig:forest_plot} is based. 

\paragraph{Funding Disclosure} This work was supported by the National Science Foundation grants CCF-2007350 (VS), DMS-2022446 (VS), CCF-1955981 (VS and BN) and CCF-1618800 (AG and BN). AG was funded in part by VS's data science faculty fellowship from the Rafik B. Hariri Institute at Boston University. We declare that we have no competing interests.

\printbibliography

@article{SanWai,
  title={Information-theoretic limits of selecting binary graphical models in high dimensions},
  author={Santhanam, Narayana P and Wainwright, Martin J},
  journal={IEEE Transactions on Information Theory},
  volume={58},
  number={7},
  pages={4117--4134},
  year={2012},
  publisher={IEEE}
}

@inproceedings{belilovsky2016testing,
  title={Testing for differences in {G}aussian graphical models: applications to brain connectivity},
  author={Belilovsky, Eugene and Varoquaux, Ga{\"e}l and Blaschko, Matthew B},
  booktitle={Advances in Neural Information Processing Systems (NIPS)},
  pages={595--603},
  year={2016}
}

@article{zhao2014direct,
  title={Direct estimation of differential networks},
  author={Zhao, Sihai Dave and Cai, T Tony and Li, Hongzhe},
  journal={Biometrika},
  volume={101},
  number={2},
  pages={253--268},
  year={2014},
}

@InProceedings{FazBan,
  title = 	 {Generalized Direct Change Estimation in {I}sing Model Structure},
  author = 	 {Farideh Fazayeli and Arindam Banerjee},
  booktitle = 	 {Proceedings of The 33rd International Conference on Machine Learning (ICML 2016)},
  pages = 	 {2281--2290},
  year = 	 {2016},
  volume = 	 {48},
}

@article{xia2015testing,
  title={Testing differential networks with applications to the detection of gene-gene interactions},
  author={Xia, Yin and Cai, Tianxi and Cai, T Tony},
  journal={Biometrika},
  volume={102},
  number={2},
  pages={247--266},
  year={2015},
  publisher={Oxford University Press}
}

@article{liu2014direct,
  title={Direct learning of sparse changes in {M}arkov networks by density ratio estimation},
  author={Liu, Song and Quinn, John A and Gutmann, Michael U and Suzuki, Taiji and Sugiyama, Masashi},
  journal={Neural computation},
  volume={26},
  number={6},
  pages={1169--1197},
  year={2014},
  publisher={MIT Press}
}

@article{liu2017,
author = "Liu, Song and Suzuki, Taiji and Relator, Raissa and Sese, Jun and Sugiyama, Masashi and Fukumizu, Kenji",
doi = "10.1214/16-AOS1470",
journal = "The Annals of Statistics",
number = "3",
pages = "959--990",
publisher = "The Institute of Mathematical Statistics",
title = "Support consistency of direct sparse-change learning in {M}arkov networks",
url = "http://dx.doi.org/10.1214/16-AOS1470",
volume = "45",
year = "2017"
}

@article{liu2017learning,
  title={Learning sparse structural changes in high-dimensional {M}arkov networks},
  author={Liu, Song and Fukumizu, Kenji and Suzuki, Taiji},
  journal={Behaviormetrika},
  volume={44},
  number={1},
  pages={265--286},
  year={2017},
  publisher={Springer}
}

@inproceedings{bresler2015structurelearning,
 author = {Bresler, Guy},
 title = {Efficiently Learning {I}sing Models on Arbitrary Graphs},
 booktitle = {Proceedings of the Forty-Seventh Annual ACM Symposium on Theory of Computing (STOC 2015)},
 year = {2015},
 location = {Portland, Oregon, USA}
}

@article{daskalakis2016testing,
  title={Testing {I}sing Models},
  author={Daskalakis, Constantinos and Dikkala, Nishanth and Kamath, Gautam},
  journal={arXiv preprint arXiv:1612.03147},
  year={2016}
}

@article{guntuboyina2011lower,
  title={Lower bounds for the minimax risk using $ f $-divergences, and applications},
  author={Guntuboyina, Adityanand},
  journal={IEEE Transactions on Information Theory},
  volume={57},
  number={4},
  pages={2386--2399},
  year={2011},
  publisher={IEEE}
}

@article{breslerkarzand,
  title={Learning a Tree-Structured {I}sing Model in Order to Make Predictions},
  author={Bresler, Guy and Karzand, Mina},
  journal={arXiv preprint arXiv:1604.06749},
  year={2016}
}

@Inbook{Yu1997,
author="Yu, Bin",
editor="Pollard, David
and Torgersen, Erik
and Yang, Grace L.",
title="Assouad, Fano, and Le Cam",
bookTitle="Festschrift for Lucien Le Cam: Research Papers in Probability and Statistics",
year="1997",
publisher="Springer New York",
address="New York, NY",
pages="423--435",
isbn="978-1-4612-1880-7",
doi="10.1007/978-1-4612-1880-7_29"
}

@article{daskalakis2019testing,
  title={Testing {I}sing models},
  author={Daskalakis, Constantinos and Dikkala, Nishanth and Kamath, Gautam},
  journal={IEEE Transactions on Information Theory},
  volume={65},
  number={11},
  pages={6829--6852},
  year={2019},
  publisher={IEEE}
}

@InProceedings{bezakova_lowbd,
  title = 	 {Lower bounds for testing graphical models: colorings and antiferromagnetic {I}sing models},
  author = 	 {Bez\'akov\'a, Ivona and Blanca, Antonio and Chen, Zongchen and {\v{S}}tefankovi{\v{c}}, Daniel and Vigoda, Eric},
  booktitle = 	 {Proceedings of the Thirty-Second Conference on Learning Theory},
  pages = 	 {283--298},
  year = 	 {2019}
}

@book{ingster_suslina,
  title={Nonparametric goodness-of-fit testing under Gaussian models},
  author={Ingster, Yuri and Suslina, Irina A},
  volume={169},
  year={2012},
  publisher={Springer Science \& Business Media}
}

@inproceedings{daskalakis2017concentration,
  title={Concentration of multilinear functions of the {I}sing model with applications to network data},
  author={Daskalakis, Constantinos and Dikkala, Nishanth and Kamath, Gautam},
  booktitle={Advances in Neural Information Processing Systems},
  pages={12--23},
  year={2017}
}

@article{gheissari2018concentration,
  title={Concentration inequalities for polynomials of contracting {I}sing models},
  author={Gheissari, Reza and Lubetzky, Eyal and Peres, Yuval},
  journal={Electronic Communications in Probability},
  volume={23},
  year={2018},
  publisher={The Institute of Mathematical Statistics and the Bernoulli Society}
}

@inproceedings{canonne2017testing,
  title={Testing Bayesian Networks},
  author={Canonne, Clement L and Diakonikolas, Ilias and Kane, Daniel M and Stewart, Alistair},
  booktitle={Conference on Learning Theory},
  pages={370--448},
  year={2017}
}

@incollection{acharya-causal,
title = {Learning and Testing Causal Models with Interventions},
author = {Acharya, Jayadev and Bhattacharyya, Arnab and Daskalakis, Constantinos and Kandasamy, Saravanan},
booktitle = {Advances in Neural Information Processing Systems 31},
editor = {S. Bengio and H. Wallach and H. Larochelle and K. Grauman and N. Cesa-Bianchi and R. Garnett},
pages = {9447--9460},
year = {2018},
publisher = {Curran Associates, Inc.}
}

@inproceedings{bresler-nagaraj,
  title={Optimal Single Sample Tests for Structured versus Unstructured Network Data},
  author={Bresler, Guy and Nagaraj, Dheeraj},
  booktitle={Conference On Learning Theory},
  pages={1657--1690},
  year={2018}
}

@article{bresler2019,
author = "Bresler, Guy and Nagaraj, Dheeraj",
doi = "10.1214/19-AAP1479",
fjournal = "Annals of Applied Probability",
journal = "Ann. Appl. Probab.",
month = "10",
number = "5",
pages = "3230--3265",
publisher = "The Institute of Mathematical Statistics",
title = "Stein’s method for stationary distributions of Markov chains and application to {I}sing models",
volume = "29",
year = "2019"
}

@article{neykov2019property,
  title={Property testing in high-dimensional {I}sing models},
  author={Neykov, Matey and Liu, Han},
  journal={The Annals of Statistics},
  volume={47},
  number={5},
  pages={2472--2503},
  year={2019},
  publisher={Institute of Mathematical Statistics}
}

@inproceedings{klivans2017learning,
  title={Learning graphical models using multiplicative weights},
  author={Klivans, Adam and Meka, Raghu},
  booktitle={2017 IEEE 58th Annual Symposium on Foundations of Computer Science (FOCS)},
  pages={343--354},
  year={2017},
  organization={IEEE}
}

@inproceedings{hamilton2017information,
  title={Information theoretic properties of Markov random fields, and their algorithmic applications},
  author={Hamilton, Linus and Koehler, Frederic and Moitra, Ankur},
  booktitle={Advances in Neural Information Processing Systems},
  pages={2463--2472},
  year={2017}
}

@article{lokhov2018optimal,
  title={Optimal structure and parameter learning of {I}sing models},
  author={Lokhov, Andrey Y and Vuffray, Marc and Misra, Sidhant and Chertkov, Michael},
  journal={Science advances},
  volume={4},
  number={3},
  pages={e1700791},
  year={2018},
  publisher={American Association for the Advancement of Science}
}

@inproceedings{wu2019sparse,
  title={Sparse logistic regression learns all discrete pairwise graphical models},
  author={Wu, Shanshan and Sanghavi, Sujay and Dimakis, Alexandros G},
  booktitle={Advances in Neural Information Processing Systems},
  pages={8069--8079},
  year={2019}
}

@inproceedings{goel2019learning,
  title={Learning {I}sing Models with Independent Failures},
  author={Goel, Surbhi and Kane, Daniel and Klivans, Adam},
  booktitle={Conference on Learning Theory},
  volume={99},
  year={2019}
}

@inproceedings{tandon2014information,
  title={On the information theoretic limits of learning {I}sing models},
  author={Tandon, Rashish and Shanmugam, Karthikeyan and Ravikumar, Pradeep K and Dimakis, Alexandros G},
  booktitle={Advances in Neural Information Processing Systems},
  pages={2303--2311},
  year={2014}
}

@article{scarlett2016difficulty,
  title={On the difficulty of selecting {I}sing models with approximate recovery},
  author={Scarlett, Jonathan and Cevher, Volkan},
  journal={IEEE Transactions on Signal and Information Processing over Networks},
  volume={2},
  number={4},
  pages={625--638},
  year={2016},
  publisher={IEEE}
}

@article{cao2018high,
  title={High Temperature Structure Detection in Ferromagnets},
  author={Cao, Yuan and Neykov, Matey and Liu, Han},
  journal={arXiv preprint arXiv:1809.08204},
  year={2018}
}

@inproceedings{gangrade2017lower,
  title={Lower bounds for two-sample structural change detection in {I}sing and {G}aussian models},
  author={Gangrade, Aditya and Nazer, Bobak and Saligrama, Venkatesh},
  booktitle={2017 55th Annual Allerton Conference on Communication, Control, and Computing (Allerton)},
  pages={1016--1025},
  year={2017},
  organization={IEEE}
}

@inproceedings{gangrade2018two,
  title={Two-Sample Testing can be as Hard as Structure Learning in {I}sing Models: Minimax Lower Bounds},
  author={Gangrade, Aditya and Nazer, Bobak and Saligrama, Venkatesh},
  booktitle={2018 IEEE International Conference on Acoustics, Speech and Signal Processing (ICASSP)},
  pages={6931--6935},
  year={2018},
  organization={IEEE}
}

@article{kim2019two,
  title={Two-sample inference for high-dimensional markov networks},
  author={Kim, Byol and Liu, Song and Kolar, Mladen},
  journal={arXiv preprint arXiv:1905.00466},
  year={2019}
}

@article{bandeira18,
	Author = {Bandeira, Afonso S},
	Journal = {Foundations of Computational Mathematics},
	Number = {2},
	Pages = {345--379},
	Publisher = {Springer},
	Title = {Random Laplacian matrices and convex relaxations},
	Volume = {18},
	Year = {2018}}

@article{bodwin2018testing,
	Author = {Bodwin, Kelly and Zhang, Kai and Nobel, Andrew},
	Journal = {The Annals of Applied Statistics},
	Number = {2},
	Pages = {1180--1203},
	Publisher = {Institute of Mathematical Statistics},
	Title = {A testing based approach to the discovery of differentially correlated variable sets},
	Volume = {12},
	Year = {2018}}

@article{cai2019differential,
	Author = {Cai, TT and Li, H and Ma, J and Xia, Y},
	Journal = {Biometrika},
	Number = {2},
	Pages = {401--416},
	Publisher = {Oxford University Press},
	Title = {Differential Markov random field analysis with an application to detecting differential microbial community networks},
	Volume = {106},
	Year = {2019}}

@article{costanzo2010genetic,
	Author = {Costanzo, Michael and Baryshnikova, Anastasia and Bellay, Jeremy and Kim, Yungil and Spear, Eric D and Sevier, Carolyn S and Ding, Huiming and Koh, Judice LY and Toufighi, Kiana and Mostafavi, Sara and others},
	Journal = {science},
	Number = {5964},
	Pages = {425--431},
	Publisher = {American Association for the Advancement of Science},
	Title = {The genetic landscape of a cell},
	Volume = {327},
	Year = {2010}}

@article{ideker2012differential,
	Author = {Ideker, Trey and Krogan, Nevan J},
	Journal = {Molecular systems biology},
	Number = {1},
	Publisher = {John Wiley \& Sons, Ltd},
	Title = {Differential network biology},
	Volume = {8},
	Year = {2012}}

@article{mohammed2016integrative,
	Author = {Mohammed, Ali I and Gritton, Howard J and Tseng, Hua-an and Bucklin, Mark E and Yao, Zhaojie and Han, Xue},
	Journal = {Scientific reports},
	Pages = {20986},
	Publisher = {Nature Publishing Group},
	Title = {An integrative approach for analyzing hundreds of neurons in task performing mice using wide-field calcium imaging},
	Volume = {6},
	Year = {2016}}

@inproceedings{neuralconnectomics,
	Author = {Javier G. Orlandi and Bisakha Ray and Demian Battaglia and Isabelle Guyon and Vincent Lemaire and Mehreen Saeed and Alexander Statnikov and Olav Stetter and Jordi Soriano},
	Booktitle = {Proceedings of the Neural Connectomics Workshop at ECML 2014},
	Editor = {Demian Battaglia and Isabelle Guyon and Vincent Lemaire and Jordi Soriano},
	Pages = {1--22},
	Series = {Proceedings of Machine Learning Research},
	Title = {First Connectomics Challenge: From Imaging to Connectivity},
	Volume = {46},
	Year = {2015}}

@article{drton2017structure,
	Author = {Drton, Mathias and Maathuis, Marloes H},
	Journal = {Annual Review of Statistics and Its Application},
	Pages = {365--393},
	Publisher = {Annual Reviews},
	Title = {Structure learning in graphical modeling},
	Volume = {4},
	Year = {2017}}

@article{phizicky1995protein,
	Author = {Phizicky, Eric M and Fields, Stanley},
	Journal = {Microbiol. Mol. Biol. Rev.},
	Number = {1},
	Pages = {94--123},
	Publisher = {Am Soc Microbiol},
	Title = {Protein-protein interactions: methods for detection and analysis.},
	Volume = {59},
	Year = {1995}}

@article{yang2017vivo,
	Author = {Yang, Weijian and Yuste, Rafael},
	Journal = {Nature methods},
	Number = {4},
	Pages = {349},
	Publisher = {Nature Publishing Group},
	Title = {In vivo imaging of neural activity},
	Volume = {14},
	Year = {2017}}

@article{zhang2019diffnetfdr,
	Author = {Zhang, Xiao-Fei and Ou-Yang, Le and Yang, Shuo and Hu, Xiaohua and Yan, Hong},
	Journal = {Bioinformatics},
	Title = {DiffNetFDR: differential network analysis with false discovery rate control},
	Year = {2019}}

@article{shojaie_survey,
author = {Shojaie, Ali},
title = {Differential network analysis: A statistical perspective},
journal = {WIREs Computational Statistics },
year = {2020},
doi = {10.1002/wics.1508}
}

@article{negahban2012unified,
  title={A unified framework for high-dimensional analysis of $ M $-estimators with decomposable regularizers},
  author={Negahban, Sahand N and Ravikumar, Pradeep and Wainwright, Martin J and Yu, Bin},
  journal={Statistical Science},
  volume={27},
  number={4},
  pages={538--557},
  year={2012},
  publisher={Institute of Mathematical Statistics}
}

@article{adamczak2019note,
  title={A note on concentration for polynomials in the Ising model},
  author={Adamczak, Rados{\l}aw and Kotowski, Micha{\l} and Polaczyk, Bart{\l}omiej and Strzelecki, Micha{\l}},
  journal={Electronic Journal of Probability},
  volume={24},
  year={2019},
  publisher={The Institute of Mathematical Statistics and the Bernoulli Society}
}

@inproceedings{daskalakis2018tolerant,
  title={Which distribution distances are sublinearly testable?},
  author={Daskalakis, Constantinos and Kamath, Gautam and Wright, John},
  booktitle={Proceedings of the Twenty-Ninth Annual ACM-SIAM Symposium on Discrete Algorithms},
  pages={2747--2764},
  year={2018},
  organization={SIAM}
}

@article{griffiths1969rigorous,
  title={Rigorous results for Ising ferromagnets of arbitrary spin},
  author={Griffiths, Robert B},
  journal={Journal of Mathematical Physics},
  volume={10},
  number={9},
  pages={1559--1565},
  year={1969},
  publisher={American Institute of Physics}
}



\newpage

\begin{appendix}

{\LARGE{\textbf{Appendices}}}

\section{Appendix to \S\ref{sec:def}}\label{appx:appx_to_defs}

\subsection{Proof of Ordering of Sample Complexities}\label{appx:sample_comp_order_pf}

The proposition is argued by direct reductions showing how a solver of a harder problem can be used to solve a simpler problem. The main feature of the definitions that allows this is that the risks of $\SL$ and $\eofbb$ are defined in terms of a probability of error.

\begin{proof}[Proof of Proposition \ref{prop:sample_comp_order}]$ $

\emph{Reducing EoF to SL}: Suppose we have a $(s-1/2)$-approximate structure learner with risk $\delta$ that uses $n$ samples. Then we can construct the following $\eofbb$ estimator with the same sample costs. Take a dataset from $Q^{\otimes n}$, and pass it to the structure learner. With probability at least $1-\delta,$ this gives a graph $\widehat{G}$ that is at most $\lfloor s/2\rfloor$-separated from $G(Q)$. Now compute $G(P) \triangle \widehat{G}$ ($G(P)$ is determined because $P$ is given to the $\eofbb$ tester). By the triangle inequality applied to the adjacency matrices of the graphs under the Hamming metric, this identifies $G(P) \triangle G(Q)$ up to an error of $(s-1)/2$, and so, the EoF risk incurred is also $\delta$. Taking $\delta = 1/8$ concludes the argument.

\emph{Reducing GoF to EoF}: Suppose we have a $s$-EoF solver that uses $n$ samples with risk $ \delta$. Again, take a dataset from $Q^{\otimes n}$, and pass it to the EoF solver, along with $P$. With probability at least $1-\delta,$ this yields a graph $\widehat{G}$ such that $|\widehat{G} \triangle (G(P) \triangle G(Q)| \le (s-1)/2$. But then, if $G(Q) = G(P),$ $\widehat{G}$ can have at most $(s-1)/2$ edges, while if $|G(P) \triangle G(Q)| \ge s,$ then $\widehat{G}$ must have at least $(s+1)/2$ edges. Thus, thresholding on the basis of the number of edges in $\widehat{G}$ produces a GoF tester with both null and alternate risk controlled by $\delta,$ or total risk $2\delta$. Taking $\delta = 1/8$ then finishes the argument. \qedhere

\end{proof}

\subsection[Proof of upper bound on sample complexity of structure learning]{Proof of Upper Bound on $n_{\sll}$}\label{appx:pf_of_sl_upper}

This proof is essentially constructed by slightly improving upon the proof of \cite[Thm 3a)]{SanWai} due to Santhanam \& Wainwright, which analyses the maximum likelihood scheme. We use notation from that paper below.

\begin{proof}[Proof of Theorem \ref{thm:sl_upper}]
    \cite{SanWai} shows, in Lemmas 3 and 4, that if the data is drawn from an Ising model $P \in \ii_d$, and $Q \in \ii_d$ is such that $G(P) \triangle G(Q) = \ell,$ then \[ P^{\otimes n}(\mathscr{L}(P) \le \mathscr{L}(Q) ) \le \exp{-n\ell \kappa/8d},\] where $\mathscr{L}(P)$ denotes the likelihood of $P,$ i.e. if the samples are denoted $\{ X^{(k)}\}_{ k \in [1:n]},$ then  $\mathscr{L}(P) = \prod_{k = 1}^{n} P(X^{(k)})$, and \[ \kappa  = (3e^{2\beta d} + 1)^{-1} \sinh^2(\alpha/4) \ge \frac{\sinh^2(\alpha/2)}{4e^{2\beta d}}.\]

Now, for the max-likelihood scheme to make an error in approximate recovery, it must make an error of at least $s$ - i.e., an error occurs only if $\mathscr{L}(Q) \ge \mathscr{L}(P)$ for some $Q$ with $G(Q) \triangle G(P) \ge s$. Union bounding this as Pg.~4129 of \cite{SanWai}, we may control this as \begin{align*}
    P(\mathrm{err}) &\le \sum_{ \ell = s}^{pd} \binom{\binom{p}{2}}{\ell} \exp{ -n \ell \kappa /8d} \\
                    &\le \sum_{\ell = s}^{pd} \exp{ \ell \left(\log \frac{ep^2}{2\ell} - n \kappa /8d\right)} \\
                    &\le \sum_{\ell = s}^{pd} \exp{ \ell \left(\log \frac{ep^2}{2s} - n \kappa /8d\right)}.
\end{align*} 

Now, if $n \kappa/8d \ge 2\log \nicefrac{ep^2}{2s} = 2 \log \nicefrac{p^2}{s} + 2(1-\log(2)),$ and if $\exp{-ns\kappa/8d} \le \nicefrac{1}{2}$ then the above is bounded as $2\exp{ - ns\kappa/8d}$, which can be driven lower than any $\delta$ by increasing $n$ by an $O(s^{-1} \log(2/\delta))$ additive factor. It follows that \[ n_{\sll}(s, \ii) \le \frac{16d}{\kappa} \left( \log\frac{p^2}{s} + 2 + O(1/s)\right), \] and the claim follows by expanding out the value of $\kappa$. \end{proof}
\section{Appendix to \S\ref{sec:Main_lowbs}}\label{appx:lowbs}

\subsection{Expanded Proof Technique}

This section expands upon \S\ref{sec:pf_tech} in the main text, including a treatment of the method used for $\eofbb$ lower bounds, giving an expanded version of Lemma \ref{lem:lifting}, and a theorem collating the resulting method to construct bounds. Some of the text from \S\ref{sec:pf_tech} is repeated for the sake of flow of the presentation.

As discussed previously, the proofs proceed by explicitly constructing distributions with differing network structures that are statistically hard to distinguish. In particular, we measure hardness by the $\chi^2$-divergence. We begin with some notation.

\begin{defi}
A $s$-change ensemble in $\mathcal{I}$ is a distribution $P$ and a set of distributions $\mathcal{Q}$, denoted $(P,\mathcal{Q}),$ such that $P \in \mathcal{I}, Q \subseteq \mathcal{I},$ and for every $Q \in \mathcal{Q},$ it holds that $|G(P) \triangle G(Q)| \ge s.$
\end{defi}

Each of the testing bounds we show will involve a mixture of $n$-fold distributions over a class of distributions. For succinctness, we define the following symbol.
\begin{defi}
For a set of distributions $\mathcal{Q}$ and a natural number $n$, we define the mixture \[ \langle \mathcal{Q}^{\otimes n} \rangle := \frac{1}{|\mathcal{Q}|} \sum_{Q \in \mathcal{Q}} Q^{\otimes n}.\]
\end{defi}

Consider the case of $\gofbb$ testing, with the known distribution $P$. Suppose we provide the tester with the additional information that the dataset is drawn either from $P,$ or from a distribution picked uniformly at random from $\mathcal{Q},$ where $(P,\mathcal{Q})$ for a $s$-change ensemble. Clearly, the Bayes risk suffered by any tester with this side information must be lower than the minimax risk of $\gofbb$ testing. The advantage of this formulation is that the risks of these tests with the side information can be lower bounded by standard techniques - basically the Neyman-Pearson Lemma. The following generic bound, which is Le Cam's two point method \cite{Yu1997, ingster_suslina} captures this. 

\begin{mylem}{(Le Cam's Method)} \label{lem:lecam}
\[R^{\gof}(n , s, \mathcal{I}) \ge \sup_{(P, \mathcal{Q})} 1 - d_{\mathrm{TV}} (\langle \mathcal{Q}^{ \otimes n} \rangle, P^{\otimes n}) \ge \sup_{(P,\mathcal{Q})} 1 -\sqrt{\frac{1}{2} \log( 1 + \chi^2( \langle \mathcal{Q}^{ \otimes n} \rangle \| P^{\otimes n}) )}, \]
where the supremum is over $s$-change ensembles in $\mathcal{I}.$ 
\end{mylem}

Above,  $\chi^2(\cdot\| \cdot)$ is the $\chi^2$-divergence, which is defined for distributions $P,Q$ as follows  \[ \chi^2(Q\|P) := \begin{cases}  \mathbb{E}_P \left[ \left( \displaystyle\frac{\mathrm{d} Q }{\mathrm{d} P} \right)^2 \right] - 1 & \textrm{ if } Q \ll P \\ \infty & \textrm{ if } Q \not\ll P \end{cases}. \] 
Note that generally the method is only stated as the first bound, and the second is a generic bound on the total variation divergence which follows from Pinsker's inequality and the monotonicity of R\'{e}nyi divergences. The $\chi^2$-divergence is invoked becuase it yields a twofold advantage in that it both tensorises well, and behaves well under mixtures such as $\langle \mathcal{Q}^{\otimes n} \rangle $ above.

For the $\eofbb$ bounds, more care is needed. Recall that the $\eofbb$ problem only requires errors smaller than $s/2.$ To address this, we introduce the following. \begin{defi}
An $(s',s)$-packing change ensemble is an $s$-change ensemble $(P, \mathcal{Q})$ such that $\mathcal{Q}$ is an $s'$-packing under the Hamming metric on network structures, that is, for every $Q, Q' \in \mathcal{Q}, |G(Q) \triangle G(Q')| \ge s'.$
\end{defi} Clearly, if one can solve the $\eofbb$ problem, one can exactly recover the structures in a $(s/2, s)$-packing change ensemble. Thus, the following lower bound of Guntuboyina is applicable.
\begin{mylem}{\cite[Example II.5]{guntuboyina2011lower}} \label{lem:guntuboyina}
\[ R^{\eof}(n, s, \mathcal{I}) \ge \sup_{(P,\mathcal{Q})} 1 - \frac{1}{|\mathcal{Q}|} - \sqrt{\frac{\sum_{Q \in \mathcal{Q}} \chi^2( Q \| P)}{|\mathcal{Q}|^{2}}}, \] 
where the supremum is taken over $(s/2, s)$-packing change ensembles in $\mathcal{I}.$
\end{mylem}

Note that \cite{guntuboyina2011lower} shows a number of lower bounds of the above form. We use the $\chi^2$-divergence here primarily for parsimony of effort, in that the bounds on $\chi^2$-divergences we construct for the $\gofbb$ setting can easily extended to the $\eofbb$ case via the above.

Our task is now greatly simplified - we merely have to construct change ensembles such that $|\mathcal{Q}|$ is large, and $\chi^2(Q \|P)$ is small for every $Q \in P.$ Since it is difficult to directly construct large degree bounded graphs with tractable distributions, we will instead provide constructions on a small number of nodes, and lift these up to the whole $p$ nodes by the following lemma.

\begin{mylem}{(Lifting)}\label{lem:lifting_back} Let $P_0$ and $Q_0$ be Ising models with degree $\le d$ on $\nu \le p$ nodes such that $|G(P_0) \triangle G(Q_0)| = \sigma,$ and $\chi^2(Q_0^{\otimes n}\|P_0^{\otimes n}) \le a_n.$ Let $m := \lfloor p/\nu \rfloor.$ For $1 \le t < m/16e,$ there exists a $t\sigma$-change ensemble $(P,\mathcal{Q})$ over $p$ nodes such that $|\mathcal{Q}| = \binom{m}{t}$ and \[ \chi^2(\langle \mathcal{Q}^{ \otimes n} \rangle \| P^{\otimes n} ) \le \frac{1}{\binom{m}{t}} \sum_{k = 0}^t \binom{t}{k} \binom{m - t}{t-k} ( (1 +a_n)^k - 1) \le \exp{ \frac{t^2}{m} a_n} - 1. \] Further, there exists a $(t\sigma/2, t\sigma)$-packing change ensemble $(P, \widetilde{\mathcal{Q}})$ over $p$ nodes such that \[ |\widetilde{\mathcal{Q}}| \ge \frac{2}{t}\left(\frac{m}{8et}\right)^{t/2} \] and  \[ \forall\,Q \in \widetilde{\mathcal{Q}}, \chi^2(Q^{\otimes n} \| P^{\otimes n} )\le (1 +a_n)^t -1 .\]
\end{mylem}

We note that the proof of the above lemma constructs explicit change ensembles. We will abuse terminology and refer to \emph{the} change ensemble or \emph{the} packing change ensemble of Lemma \ref{lem:lifting_back}.

The above Lemma requires control on $n$-fold products of two distributions. However, since the $\chi^2$-divergence is conducive to tensorisation, control for $n=1$ is usually sufficient. The statement below captures this fact and gives an end to end lower bound on this basis. The statement amounts to collating the various facts described in this section.

\begin{myth}\label{thm:n_bound_tech}
Let $P_0$ and $Q_0$ be as in Lemma \ref{lem:lifting_back}. Suppose further that $\chi^2(Q_0\|P_0) \le \kappa$. Then for $1 \le t < m/16e,$ where $m = \lfloor p/\nu\rfloor,$ \begin{align*}
    n_{\gof}(t\sigma, \ii_d) &\ge \frac{1}{2\log(1+\kappa)} \log \left( 1 + \frac{m}{t^2} \right), \\
    n_{\eof}(t\sigma, \ii_d) &\ge \frac{1}{2\log(1+\kappa)} \log \left( \frac{m}{4000t}\right).
\end{align*} 

\end{myth}

The $4000$ in the above can be improved under mild assumptions, such as if $t \ge 8$, but we do not pursue this further. We conclude this section with proofs of the main claims above.

\subsubsection{Proof of Lifting Lemma}
\begin{proof}[Proof of Lemma \ref{lem:lifting_back}] Let $G_0, H_0$ be the network structures underlying $P_0, Q_0,$ and $A_0, B_0$ be the weight matrices of $G_0, H_0.$ Recall that these are graphs on $\nu$ nodes. Partition $[1:p]$ into $m+1$ pieces  $(\pi_1, \pi_2, \dots, \pi_m) = ([1:\nu], [\nu+1: 2\nu] ,\dots [(m-1)\nu+1:m\nu]) $ and $ \pi_{m+1} = [m\nu + 1:p]$, the last one being possibly empty. We may place a copy of $G_0$ on each of the first $m$ parts, and leave the final graph disconnected to obtain a graph $G$ with the block diagonal weight matrix $\mathrm{diag}(A_0, A_0, \dots, A_0, 0).$ We let $P$ be the Ising model on $G$. For any vector $\mathbf{v} \in  \{0,1\}^m$ of weight $t,$ let $Q_{\mathbf{v}} $ be the graph which places a copy of $B_0$ on $\pi_i$ for all $i : \mathbf{v}_i =1,$ and $A_0$ as before otherwise. Note the block independence across parts of $\pi$ induced by this. Concretely, we have \begin{align*}
        P(X=x) &= \prod_{i =1}^m P_0(X_{\pi_i}=x_{\pi_i}) \cdot 2^{-|\pi_{m+1}|},\\
        Q_\mathbf{v}(X = x) &= P(X = x) \cdot \prod_{i: \mathbf{v}_i = 1} \frac{Q_0(X_{\pi_i} = x_{\pi_i})}{P_0(X_{\pi_i} = x_{\pi_i})}.
    \end{align*} 
    
    Now, let $\mathcal{V}_t$ be the $t$-weighted section of the cube $\{0,1\}^m,$ and $\mathcal{V}'_t$ be a maximal $t/2$ packing of $\mathcal{V}_t$. 
    
    We let $\mathcal{Q} := \{Q_{\mathbf{v}}, \mathbf{v} \in \mathcal{V}_t\}$ and $\mathcal{Q}' := \{ Q_{\mathbf{v}}, \mathbf{v} \in \mathcal{V}'_t\}.$ Since $(P_0, Q_0)$ had symmetric difference $\sigma,$ and since we introduce $t$ differences of this form in $\mathcal{Q},$ $(P, \mathcal{Q})$ forms a $t\sigma$-change ensemble. Further, $\mathcal{Q}'$ inherits the packing structure of $\mathcal{V}_t'$, $(P,\mathcal{Q}')$ forms a $(t\sigma/2, t\sigma)$-packing change ensemble. Next note that $|\mathcal{Q}| = \binom{m}{t}$ trivially. Further, since $|\mathcal{Q}|' = |\mathcal{V}_t|,$ it suffices to lower bound the latter to show that $\mathcal{Q}$ is as big as claimed. Since $\mathcal{V}'_t$ is maximal, its cardinality must exceed the $t/2$-covering number of the $t$-section of the cube. But then, by a volume argument, \[ |\mathcal{V}'_t| \ge \frac{\binom{m}{t}}{\sum_{k = 0}^{t/2} \binom{t}{k}\binom{m-t}{k}} \ge \frac{\binom{m}{t}}{(t/2) 2^t \binom{m}{t/2}} \ge \frac{2}{t} \left( \frac{m}{t}\right)^t 2^{-t} \left(\frac{2em}{t}\right)^{-t/2} = \frac{2}{t} \left( \frac{m}{8e t} \right)^{t/2} \] where we have used $t \le m/4.$ 
    
    Next, note that for any $Q_{\mathbf{v}} \in \mathcal{Q},$ and hence any $Q_{\mathbf{v}} \in \mathcal{Q}',$ we have \[1 + \chi^2(Q^{\otimes n}\|P^{\otimes n}) = \mathbb{E}_{P^{\otimes n}} \prod_{\mathbf{v}_i = 1} \frac{Q_0^{\otimes n}}{P_0^{\otimes n}}(X_{\pi_i}^n) = \left(1 + \chi^2(Q_0^{\otimes n}\| P^{\otimes n}) \right)^t. \] Finally, \begin{align*}
        1 + \chi^2( \langle \mathcal{Q}^{\otimes n} \rangle \| P^{\otimes n}) &= \frac{1}{|\mathcal{Q}|^2} \sum_{\mathbf{v}, \mathbf{v}' \in \mathcal{V}_t} \mathbb{E}_{P^\otimes n} \left[ \frac{Q_{\mathbf{v}}^{\otimes n} Q_{\mathbf{v}'}^{\otimes n}}{(P^{\otimes n})^2}(X^n) \right]\\
                                                                        &= \frac{1}{\binom{m}{t}^2} \sum_{\mathbf{v}, \mathbf{v}'\in \mathcal{V}_t } \prod_{i : \mathbf{v}_i = \mathbf{v}'_i = 1} \mathbb{E}_{P_0^{\otimes n}} \left[ \frac{(Q_0^{\otimes n})^2}{(P_0^{\otimes n})^2} (X_{\pi_i}^n) \right] \\
                                                                        &\le \frac{1}{\binom{m}{t}^2} \sum_{\mathbf{v, v'} \in \mathcal{V}_t} (1 + a_n)^{|\{i:\mathbf{v}_i = \mathbf{v}'_i = 1\}|} \\
                                                                        &= \frac{1}{\binom{m}{t}} \sum_{j = 0}^{t} \binom{t}{j} \binom{m-t}{t-j} (1+a_n)^j.
    \end{align*}
    
    Finally, note that the final expression can be written as $\mathbb{E}[ (1 + a_n)^{\mathscr{H}}]$ where $\mathscr{H} \sim \mathrm{Hyp}(m,t,t)$. Since hypergeometric random variables are stochastically dominated by the corresponding binomial random variables, we may upper bound the above by the moment generating function of a $\mathrm{Bin}(t, t/m)$ random variable at $(1+a_n)$ to yield that \[1 + \chi^2( \langle \mathcal{Q}^{\otimes n} \rangle \| P^{\otimes n}) \le \left( 1 + (t/m) ( (1 + a_n) - 1) \right)^t \le \exp{\frac{t^2}{m} a_n} . \qedhere\]
\end{proof}

\subsubsection{Proof of Theorem \ref{thm:n_bound_tech}}

\begin{proof}
It is a classical fact that the $\chi^2$-divergence tensorises as \[ \chi^2(Q^{\otimes n}_0 \| P^{\otimes n}_0) = (1 + \chi^2(Q_0\| P_0))^n - 1.\] The reason for this is that due to independence, $1 + \chi^2(Q_0^{\otimes n}\| P_0^{\otimes n})$ amounts to a product of second moments of relative likelihoods ($Q/P$) of iid samples.

Thus, since $\chi^2(Q_0\|P_0) \le \kappa,$ we may set $a_n = (1+\kappa)^n - 1$ in Lemma \ref{lem:lifting_back}. Now, by LeCam's method (Lemma \ref{lem:lecam}), we know that if $R_{\gof}(t\sigma) < 1/4$ for a given $n$, then using ensemble from Lemma \ref{lem:lifting_back}, it must hold that \begin{align*}
    \frac{1}{4} &\ge 1 - \sqrt{\frac{1}{2} \log \left( 1 + \exp{\frac{t^2}{m}a_{n}} -  1\right) }\notag \\
    \iff\, a_{n} &\ge 2 (3/4)^2\frac{m}{t^2} \notag\\
    \implies\, (1+\kappa)^n -1 &\ge \frac{m}{t^2} \\
    \iff\, n &\ge \frac{1}{\log(1 + \kappa)} \log \left( 1 + \frac{m}{t^2}\right)
\end{align*}

Thus, the smallest $n$ for which we can test $t\sigma$-changes in $\ii_d$ must exceed the above lower bound, giving the stated claim.

The $\eofbb$ claim follows similarly. Using the packing change ensemble from Lemma \ref{lem:lifting_back}, and the lower bound Lemma \ref{lem:guntuboyina}, if the risk is at most $1/4$ for some $n$, then we find that 
\begin{align*}
    \frac14 &\ge 1- \frac{1}{|\widetilde{\mathcal{Q}}| } - \sqrt{ \frac{(1 + a_n)^t -1 }{|\widetilde{\mathcal{Q}}|}} \notag \\ 
                \iff\, (1 + a_n)^t &\ge 1 + |\widetilde{\mathcal{Q}}| \left( \frac34 - \frac1{|\widetilde{\mathcal{Q}}|}\right)^2 \\
                \iff\, (1+\kappa)^{nt} &\ge 1 + |\widetilde{\mathcal{Q}}| \left( \frac34 - \frac1{|\widetilde{\mathcal{Q}}|}\right)^2 \\\
                \iff\, n &\ge \frac{1}{t \log(1 + \kappa)} \log \left( |\widetilde{\mathcal{Q}}| \left( \frac34 - \frac1{|\widetilde{\mathcal{Q}}|}\right)^2 \right)
\end{align*}

Now, since $1\le  t \le m/16e,$ we observe that \[ |\widetilde{\mathcal{Q}}| \ge \frac{2}{t}\left( \frac{m}{8et}\right)^{t/2} \ge \frac{2}{t} \cdot 2^{t/2} \ge 2.5.\] Thus, $(\nicefrac{3}{4} - \nicefrac{1}{|\widetilde{\mathcal{Q}}}|)^2 \ge 1/9,$ and the term in the final log above is at least $\log |\widetilde{\mathcal{Q}}|/9$, which in turn is lower bounded by Lemma \ref{lem:lifting_back}. Thus continuing the above chain of inequalities, we observe that \[n \ge \frac{1}{\log(1+\kappa)} \cdot \frac{1}{t}\left( \frac{t}{2} \left( \log\left(\frac{m}{8et}\right) - \frac{(2\log(t/2) + 4\log(3))}{t} \right)\right) \]

Finally, since $\log(x)/(x/2) \le 1/e$, we may $-2(\log(t/2) + 4\log(3))/t \ge -5.$ Folding this $-5$ into the $\log$ gives $8e^6 \le 4000$ in the denominator. Finally, again, this tells us that the infimum of the $n$ for which the $\eofbb$ risk is small is at least the above lower bound, yielding the claim. \end{proof}

\subsection{Expanded Lower Bound Theorem Statements and Proofs}

We give slightly stronger theorem statements than those in the main text, and give the proofs of the claimed bounds. In all cases the proofs involve the use of Lemma \ref{lem:lifting_back} - we describe which widgets are used, and what values of $\sigma, t$ are needed. Then we simply invoke Theorem \ref{thm:n_bound_tech} repeatedly to derive the results.

\subsection[Moderately Small Changes]{The case $d \le s \le c p$}

\begin{proof}[Proof of Theorem \ref{thm:lb_small_s}]$ $\\
\textbf{High Temperature Bound} This is shown by using the Triangle construction of \S\ref{sec:triangle}. This construction amounts to $\sigma = 1$ and $m = \lfloor p/3\rfloor.$ Thus taking $t = s,$ $\mu = \alpha, \lambda = \beta$ and invoking both Proposition \ref{prop:triangle} and Theorem \ref{thm:n_bound_tech}, we find that so long as $ p/6 \ge 16e s,$ the bounds \[ n_{\gof}(s, \ii_d) \ge \frac{1}{C\tanh^2(\alpha) e^{-2\beta}} \log \left( 1 + \frac{p}{Cs^2} \right),\] and similarly  \[ n_{\eof}(s, \ii_d) \ge \frac{1}{C\tanh^2(\alpha) e^{-2\beta}} \log \left(\frac{p}{Cs} \right).\]

\textbf{Low Temperature Bound} Let $\beta d \ge \log d$. We show this for even $d$ - odd $d$ follows by reducing $d$ by one. We use the Emmentaler clique versus the full clique of \S\ref{sec:emmen} with $\ell = 1$. This corresponds to $\sigma = d/2$ and $m = \lfloor p/d+1 \rfloor \ge p/2d$. Now take $t = \lceil 2s/d \rceil \le 4s/d.$ Note that the total number of changes is at least $s$ and at most $d/2 \lceil 2s/d\rceil \le 2s$. Notice that $t \le m$ holds so long as $s \le p/K$ for some $K \ge 400.$ Invoking Proposition \ref{lem:emmen_vs_full} in the case of $\mu = \alpha, \lambda = \beta,$ and then Theorem \ref{thm:n_bound_tech} with the stated $m, \sigma, t,$ gives us the bound \begin{align*} n_{\gof} &\ge \frac{1}{C d^2 \min(1, \mu^2 d^4) e^{-2\beta(d- 3)}} \log \left( 1 + \frac{1}{C} \frac{(p/2d)}{(4s/d)^2}\right) \\ &\ge \frac{e^{2\beta(d- 3)}}{C' d^2 \min(1, \mu^2 d^4) } \log \left( 1 + \frac{1}{C'} \frac{p d}{s^2}\right),\end{align*} where the $(d-3)$ in the exponent arises as $(d-1) -1 - \ell,$ and $d-1$ occurs since we may reduce $d$ by $1$ to make it even. Similarly \[ n_{\eof} \ge \frac{e^{2\beta(d- 3)}}{C' d^2 \min(1, \mu^2 d^4) } \log \left( 1 + \frac{1}{C'} \frac{p }{s}\right). \]

\textbf{Integrating the bounds.} We now note that if $\beta d \le 3\log d,$ then \[ \frac{e^{2\beta(d- 3)}}{d^2 \min(1, \mu^2 d^4) } \le \frac{e^{2\beta}}{\tanh^2(\alpha)}.\] Indeed, in this case, $e^{2\beta (d-3)}$ is bounded as $d^6,$ and so the left hand size is at most $d^4\min(1, \alpha^2 d^4)^{-1} \le \alpha^{-2}$, which is dominated by the right hand side. 

On the other hand even if $\beta d \ge 3\log d,$ we may still use the high temperature bound since this is shown unconditionally. Thus, at least so long as we replace the $pd/s^2$ in the low temperature bound by $p/s^2,$ we may take the maximum of the expressions in the above bounds to get a concise lower bound - the low temperature term itself only becomes active when $\beta d \le 3\log d,$ in which case it is known to be true. The claim thus follows.
\end{proof}

\subsection[Relatively Large Changes]{The case $cp \le s \le c pd^{1-\zeta}$}

We first state the commensurate $\eofbb$ bound - \begin{myth}\label{thm:large_s_eof}
In the setting of Theorem \ref{thm:lb_large_s}, we further have that \begin{enumerate}[leftmargin = .5in, parsep = 3pt, topsep = 2pt] 
    \item If $\alpha d^{1-\zeta} \le 1/32$ then \(\displaystyle n_{\eof} \ge C \frac{1}{d^{2-2\zeta} \alpha^2} \log \Big(1 + C\frac{pd^{1-\zeta}}{s}\Big).\)
    \item If $\beta d \ge 4\log(d-4)$ then \( \displaystyle n_{\eof} \ge C\frac{e^{2\beta d(1 - d^{-\zeta})}}{d^2 \min(1, \alpha^2 d^4)} \log \Big(1 + C\frac{pd^{1-\zeta}}{s}\Big). \)
\end{enumerate}
\end{myth}

\begin{proof}[Proofs of Thms.~\ref{thm:lb_large_s} and \ref{thm:large_s_eof}] $ $\\
\textbf{High Temperature Bounds} Suppose $s = pd^{1-\zeta_0}/K$ for any $\zeta_0 \in (0,1]$. We invoke the widget of a full $d^{1-\zeta_0}$-clique as $Q_0$ versus an empty graph as $P_0$, i.e. the construction of \S\ref{sec:clique_high_temp}. This corresponds to taking $\sigma = d^{2-2\zeta_0}/2 + O(d),$ $m \ge p d^{-(1-\zeta_0)}/2$ and $t = \lfloor 2s d^{-(2-2\zeta_0)} \rfloor$, with the total edit made being at most $2s$. Invoking Proposition \ref{prop:high_temp_clique} with $\mu = \alpha,$ and then Theorem \ref{thm:n_bound_tech} gives the bounds on noting that \begin{align*}
    \frac{m}{t^2} &\ge C\frac{p d^{-(1-\zeta_0)} }{(s d^{-(2-2\zeta_0)})^2} = C \frac{p d^{3 - 3\zeta_0}}{s^2}, \\
    \frac{m}{t} &\ge C\frac{p d^{-(1-\zeta_0)} }{(s d^{-(2-2\zeta_0)})} = C \frac{p d^{1 - \zeta_0}}{s^2}
\end{align*} 

and then finally setting $\zeta_0 \ge \zeta$ to derive the claim.

\textbf{Low Temperature Bounds} Again fix a $\zeta_0$. We invoke the Emmentaler clique v/s full clique widget of \ref{sec:emmen}, but this time with $\ell = d^{1- \zeta_0}.$ This gives $\sigma \approx d^{2-\zeta_0}/2$, $m = \lfloor p/d \rfloor$ and $t = \lceil 2s d^{-2 - \zeta_0}\rceil.$ The bound now follows similarly to the above section upon invoking Propositions \ref{lem:emmen_vs_full} with $\lambda = \beta, \mu = \alpha$ and then Theorem \ref{thm:n_bound_tech} with the stated $m,t,\sigma.$ We only track the terms in the log, which are 
\begin{align*}
    \frac{m}{t^2} &\ge C\frac{p d^{-1} }{(s d^{-(2-\zeta_0)})^2} = C \frac{p d^{3 - 2\zeta_0}}{s^2}, \\
    \frac{m}{t} &\ge C\frac{p d^{-1} }{(s d^{-(2-\zeta_0)})} = C \frac{p d^{1 - \zeta_0}}{s^2}. \qedhere
\end{align*} 
\end{proof}

\subsection[Very Small Changed]{Proofs in the setting $s\le d$}

The catch in this section is that the Emmentaler clique construction of the proofs above can no longer be employed, since setting even $\ell = 1$ in these induces $\Omega(d)$ changes. We instead turn to the clique with a large hole construction of \S\ref{sec:cliquehole}.

\begin{proof}[Proof of Theorem \ref{thm:very_small_s}]$ $\\
\textbf{High Temperature Bound} This is the same as the high temperature bound of Thm. \ref{thm:lb_small_s}, and that proof may be repeated.

\textbf{Low Temperature Bound} Suppose $\beta d \ge 3\log d$. We use the clique with a large hole construction of \S\ref{sec:cliquehole} with the choice of $\ell = \lceil \sqrt{2s} \rceil$. This amounts to $s\le \sigma = s + O(\sqrt{s}) \le 2s,$ and $m = \lfloor p/d \rfloor$. We then simply set $t = 1$ in Theorem \ref{thm:n_bound_tech}. Now invoking Proposition \ref{prop:clique_hole_P||Q}, we find that \[ n_{\gof} \ge \frac{1}{C\sqrt{s} \sinh^2(\alpha \sqrt{s}) e^{-2\beta (d - 1-2\sqrt{s})}} \log \left(1 + \frac{p}{Cd} \right) \ge  \frac{e^{2\beta (d - 1-2\sqrt{s})}}{C d^6\sinh^2(\alpha \sqrt{s}) } \log \left(1 + \frac{p}{Cd} \right) \] and the same lower bound for $n_{\eof}$ since in this case $m/t^2 = m/t = 1$ (the $d^6$ is introduced to make the following easy).

\textbf{Integrating the bounds} Similarly to the proof of Thm.~\ref{thm:lb_small_s}, note that for $\beta d \le 3\log d, e^{2\beta d} d^{-6} \le 1,$ allowing us to rewrite the low-temperature bound as the $\max$ expression in the theorem statement. Giving up space in the logarithm to $p/s^2 \wedge p/d$ then yields the stated claim for $\gofbb$. For $\eofbb,$ we follow the same procedure, but note that since $s \le d,$ $(p/s \wedge p/d) = p/d$.
\end{proof}

\section{Appendix to \S\ref{sec:testing_deletions}} 

\subsection{Testing Deletions in Forests, and Changes in Trees}\label{appx:forests}

\subsubsection{Proofs of Lower Bounds}

\begin{proof}[Proof of Lower bounds from Theorem \ref{thm:forest_testing}]

First note that $n \ge 1$ is necessary, since testing/estimation with no samples is impossible. To derive the second term in the converse for $\gofbb$ and the converse for $\eofbb$, we plug in the single-edge widget of \S\ref{sec:single_edge} with $\mu = \alpha$ into Theorem \ref{thm:n_bound_tech}. The widget corresponds to $\nu = 2$ and $\sigma = 1$. Thus, setting $t = s$ and $m = \lfloor p/2 \rfloor \ge p/3,$ we obtain both the claimed bounds. 
\end{proof}

\subsubsection{Proof of Upper Bound of Theorem \ref{thm:forest_testing}, and of Theorem \ref{thm:tree_testing}}

We give the proof for $\alpha > 0.$ The proof for $\alpha < 0$ follows identically.

We use $u$ as a short hand for a pair $(i,j)$ with $i < j,$ and set $Z_u = X_i X_j$. We exploit two key properties of forest structured graphs \begin{enumerate}
    \item For any $u = (i,j)$, if nodes $i$ and $j$ are connected via the graph, then $\mathbb{E}[Z_u] = \prod_{v \in \mathrm{path}(u)} \tanh(\theta_v),$ where for $u = (i,j)$ $\mathrm{path}(u)$ is the unique path connecting $i$ and $j.$ If $i$ and $j$ are not connected, then $\mathbb{E}[Z_u] = 0$.
    \item For any $u \neq v$, $\mathbb{E}[Z_u Z_v] = \mathbb{E}[Z_u] \mathbb{E}[Z_v]$, that is, the $Z_u$s are pairwise uncorrelated.
\end{enumerate}

The above are standard properties, and are shown by exploiting the fact that conditioning on any node in the forest breaks it into two uncorrelated forests. See, e.g.~\cite{breslerkarzand} for proofs.

\begin{proof}[Proof of Upper Bound in Theorem \ref{thm:forest_testing}]\label{forest_upper_pf}

Recall the statistic $\mathscr{T} = \sum_{\ell = 1}^n \sum_{u \in G(P)} Z_u^{\ell}/n,$ where the outer sum is over samples. Suppose $G(P)$ has $k$ edges. Let $\tau := \tanh(\alpha)$. We propose the test \[ \mathscr{T} \overset{\mathrm{Null}}{ \underset{\mathrm{Alt}}{\gtrless}} (k - s/2) \tau.\]

Since the sum is over all edges in $p$, and since all edges have the same weight $\alpha,$ we note that \[ \mathbb{E}_P[\mathscr{T}] = k \tau.\] 
Now consider an alternate $Q_\Delta$ that deletes some $\Delta \ge s$ of these edges. Since a deletion of an edge in the forest disconnects the nodes at the end of the edges (the path connecting two nodes in a forest is unique, if it exists, and we've just removed that unique path by deleting the edge), \[ \mathbb{E}_{Q_\Delta} [\mathscr{T}] = (k- \Delta) \tau.\]

Next, we consider the variance of the statistic. Due to uncorrelation of $Z_u$s, under any forest structured Ising model we have in the case of $n = 1$ \[ \mathrm{Var}[\mathscr{T}] = \sum_{u \in G(P)} (1 - (\mathbb{E}[Z_u])^2,\] where we have used that $Z_u^2 = (\pm 1)^2 = 1$ always. Using the standard behaviour of variances under averaging of independent samples, \begin{align*}
    \mathrm{Var}_{P^{\otimes n}}[\mathscr{T}] &= \sum_{u \in G(P)} \frac{1 - \tau^2}{n} = \frac{k(1 - \tau^2)}{n},\\
    \mathrm{Var}_{Q_\Delta^{\otimes n}}[\mathscr{T}] &= \sum_{u \in G(P)\cap G(Q_{\Delta})} \frac{1 - \tau^2}{n} + \sum_{ u \in G(P)\setminus G(Q_{\Delta})} 1/n  = \frac{ k(1 - \tau^2) + \Delta \tau^2}{n}.
\end{align*}

Using Tchebycheff's inequality, we then observe that for a given constant $C > 1$, the following hold with probability at least $7/8:$ \begin{align*}
    \textrm{ Under $P^{\otimes n}$: }& \qquad \mathscr{T} \ge k \tau - C\sqrt{ \frac{k(1-\tau^2)}{n}},\\
    \textrm{ Under any $Q_\Delta^{\otimes n}$:}& \quad \mathscr{T} \le (k -\Delta) \tau + C\sqrt{\frac{ k(1-\tau^2) + \Delta \tau^2}{n}}.
\end{align*}

Thus, the test has false alarm and size both at most $1/8,$ irrespective of $P$ and $Q_\Delta$, so long as \[ (k -\Delta) \tau + C\sqrt{\frac{ k(1-\tau^2) + \Delta \tau^2}{n}} < (k-s/2) \tau < k \tau - C\sqrt{ \frac{k(1-\tau^2)}{n}}. \]

Solving out the upper bound on $(k-s/2)\tau$ yields \[ n > 4C^2 \frac{k}{s^2} (\tau^{-2} - 1), \] while for the lower bound, since $\Delta \ge s,$ the same must hold if \begin{align*}
    (k - \Delta) \tau + C\sqrt{\frac{ k(1-\tau^2) + \Delta \tau^2}{n}} < (k-\Delta/2) \tau,
\end{align*}
which may be rearranged to \[ n > 4C^2 \left( \frac{1}{\Delta} + \frac{k}{\Delta^2} (\tau^{-2}-1) \right),\] which in turn must hold if \[ n > 4C^2 \left( 1 + \frac{k}{s^2}(\tau^{-2} - 1) \right), \] where the final inequality again utilises $\Delta \ge s$. 

Thus, forests with $k$ edges can be tested with risk at most $1/4$ as long as we have at least \[ 4C^2 \left( 1 + \frac{k}{s^2}(\tau^{-2} - 1) \right) + 1 \le C' \max\left(1 , \frac{k}{s^2} (\tau^{-2} - 1) \right)\] samples, where $C' \le 8C^2 + 1$ is a constant. Since forests on $p$ nodes have at most $p-1$ edges, replacing $k$ by $p$ yields an upper bound on the sample complexity of testing deletions in forests.

Finally, since $\tau = \tanh(\alpha),$ we note that \( \tau^{-2} - 1  = {\sinh^{-2}(\alpha)},\) concluding the proof.
\end{proof}

\subsubsection*{Some Observations}

\begin{itemize}[leftmargin = 0.16in]
    
    \item While the above proof is for uniform edge weights, this can be relaxed with little change. However, the above proof does strongly rely on the edge weights all having the same sign. If this is not the case, then we may encounter edit the same number of positively and negatively weighted edges, and the statistic $\mathscr{T}$ becomes uninformative.    
    
    \item The statistic $\mathscr{T}$ similarly loses power in the general setting of testing both additions and deletions in forests. This is because while the variance remains controlled as $k (1-\tau^2),$ the means under the alternates may not move if the only changes being made are additions.
        
    \item On the other hand, if we consider testing only of full trees, i.e. $P$ such that $G(P)$ has the full $(p-1)$ edges, and further the altered $Q$ are also trees, then something interesting emerges - at least in the setting of uniform weights. Since at least $s$ edges were changed from $G(P)$ to $G(Q)$, and one cannot add an edge to $G(P)$ without introducing a cycle, it must be the case that $G(Q)$ effects at least one edge-deletion for every edge it adds, and so it must make at least $\ge s/2$ deletions. In this case, the statistic discussed above \emph{is} powerful. This, of course, was the point of Theorem \ref{thm:tree_testing} in the main text, which we are now ready to prove
    
\begin{proof}[Proof of Upper Bound from Theorem \ref{thm:tree_testing}] Assume that $\alpha >0$. The proof proceeds similarly for $\alpha < 0$. We use the statistic $\mathscr{T}$ from the proof of the upper bound of Thm.~\ref{thm:forest_testing} above, and also reuse the notation of $\tau, \Delta$ and $Q_{\Delta}$ from the above. The claim relies on the above observation that if $\Delta$ edges are changed, then at least $\Delta/2 \ge s/2$ edges must be deleted.

In this case, the mean and the variance of $\mathscr{T}$ under $P$ remain unchanged. On the other hand, under $Q_{\Delta},$ for any edge $u \in G(P)$ that was deleted in $G(Q_{\Delta})$, we must have $|\mathbb{E}_{Q_{\Delta}}[Z_u]| \le \tau^2, $ since the distance between the end points of these edges is now at least $2$. Further, since $G(Q)$ is a tree, the variance of the statistic under $Q_{\Delta}$ (for $n = 1$) is \begin{align*} \mathrm{Var}_{Q_\Delta}[\mathscr{T}] &= \sum_{u \in G(P)} (1 - \mathbb{E}_{Q_\Delta}[Z_u]^2) \\ &\le (p-1 -\Delta)(1- \tau^2) + \Delta \\ &= (p-1)(1-\tau^2) + \Delta \tau^2.\end{align*}

At this point the argument from the earlier proof of Thm.~\ref{thm:forest_testing} can be used. The test needs to be updated to declaring for the null only when $\mathscr{T} > (p-1)\tau - s\tau(1-\tau)/4$.
\end{proof}

We conclude by showing the lower bound in Theorem \ref{thm:tree_testing}. This requires a mild departure from the previously discussed lower bounds, in that the lifting trick is not applicable - this fundamentally constructs disconnected graphs, while trees need to be connected. However, pretty much the same approach is used.

\begin{proof}[Proof of the Lower Bound from Theorem \ref{thm:tree_testing}]
    We use Le Cam's method, as before. The construction is as follows: Let $p$ be odd, and let $m = (p-1)/2.$ Take $P$ to be the Ising model with uniform weights $\alpha$ on the graph with the edge set \[ G(P) = \{ (p,i): i \in [1:m] \} \cup \{ (i, m+i): i \in [1:m]\}.\] This is a `two-layer star' - one node is singled out as central. Half the remaining nodes are incident on it, and the other half are each incident on one of these `inner' nodes.
    
    Let $t = \lceil s/2\rceil$, assumed smaller than $m$. For each $S \subset [1:m]$ such that $|S| = t,$ we define $Q_S$ to be the Ising model with uniform weights $\alpha$ on the following graph \[G(Q_S) = \{(p,i): i \in [1:m]\setminus S\} \cup \{(p,m+1): i \in S\} \cup \{(i,m+i): i \in [1:m]\}. \] In words, $Q_S$ detaches node $i$ from node $p$ and attaches node $(m+i)$ to node $p$ for $i \in S$, thus switching some of the inner nodes to being outer and vice versa. Notice that in total, $2|S| = 2t \in \{s, s+1\}$ edges have been changed.
    
    We directly argue that for $P$ as defined above, and $\mathcal{Q} = \{Q_S : S \subset [1:m], |S| = t\},$ it holds that \[ \chi^2\left( \langle \mathcal{Q}^{\otimes n} \rangle, P^{\otimes n}\right) \le \exp{ \frac{t^2}{m} \left(  (1 + 2\tanh^2\alpha)^n -1 \right)} -1. \] This, along with Le Cam's method implies the claim upon noting that $m/t^2 \ge 2(p-1)/(s+1)^2$ which is in turn larger than $p/s^2$ for $s \ge 4, p \ge 9.$
    
    Let us proceed to show the above claim. By direct computation, \[1 + \chi^2\left( \langle \mathcal{Q}^{\otimes n} \rangle, P^{\otimes n}\right) = \frac{1}{|\cal Q|^2}\sum_{S, \tilde{S}} \left(\mathbb{E}_P\left[ \frac{Q_S(X) Q_{\tilde{S}}(X)}{P(X)^2} \right]\right)^n.\]
        We invoke the following calculation \begin{mylem}\label{lem:subsidiary_tree_lower_bound}
            \[ \mathbb{E}_P\left[ \frac{Q_S(X) Q_{\tilde{S}}(X)}{P(X)^2} \right]  \le (1 + 2\tanh^2\alpha)^{|S \cap \tilde{S}|}.\] 
        \end{mylem}
        
        Let $\varphi := (1 + 2\tanh^2(\alpha))^n$. Plugging the above result into the expression for the $\chi^2$-divergence, we find that \begin{align*}
            1 + \chi^2\left( \langle \mathcal{Q}^{\otimes n} \rangle, P^{\otimes n}\right) &\le \frac{1}{|\cal Q|^2}\sum_{S, \tilde{S}} \varphi^{|S \cap \tilde{S}|} \\
                                                                                           &= \sum_{k = 0}^t \frac{\binom{t}{k} \binom{m-t}{t-k}}{\binom{m}{t}} \varphi^k \\
                                                                                           &= \mathbb{E}[ \varphi^{\mathscr{H}}]\\
                                                                                           &\le \mathbb{E}[\varphi^{\mathscr{B}}]\\
                                                                                           &= \left( 1 + \frac{t}{m}(\varphi - 1) \right)^t \le \exp{\frac{t^2}{m}(\varphi - 1)},
        \end{align*}
        
        where we have used the fact that $\cal Q$ is parametrised by all subsets of size $t$ of a set of size $m,$ and then proceeded similarly to the proof of the first part in Lemma \ref{lem:lifting_back}, with $\mathscr{H}$ being a $(m,t,t)$-hypergeometric random variable, and $\mathscr{B}$ being a $(t,t/m)$-binomial random variable. It remains to show the above Lemma, which is argued below.\qedhere        
\end{proof}

\begin{proof}[Proof of Lemma \ref{lem:subsidiary_tree_lower_bound}]

    Notice that \begin{align*}
        P(x) &= \frac{1}{2^p \cosh^{p-1}(\alpha)} \exp{\alpha \left( x_p\sum_{i = 1}^m x_i + \sum_{i = 1}^m x_i x_{m+i} \right)}\\
        Q_S(x) &= \frac{1}{2^p \cosh^{p-1}(\alpha)} \exp{\alpha \left( x_p\sum_{i \in S^c} x_i + x_p \sum_{i \in S} x_{m+i} + \sum_{i= 1}^m x_ix_{m+i}\right)}
    \end{align*}
Where the partition functions are directly calculated. As a consequence, \begin{align*}
   2^p \cosh^{p-1}(\alpha) \mathbb{E}_P\left[ \frac{Q_S(X) Q_{\tilde{S}}(X)}{P(X)^2}\right] &= 2^p \cosh^{p-1}(\alpha) \sum_x \frac{Q_S(x) Q_{\tilde{S}}(x)}{P(x)} \\ &= \sum_x \exp{\alpha\left(  x_p \sum_{i \in (S \cup \tilde{S})^c} x_i + \sum_{i \in (S \cup \tilde{S})^c} x_i x_{m+i} \right)} \\ &\quad \times \exp{\alpha\left( x_p \sum_{i \in S \triangle \tilde{S}}x_{m+i} + \sum_{i \in S \triangle \tilde{S}} x_i x_{m+i} \right)} \\ &\qquad \times  \exp{ \alpha \left( x_p \sum_{i \in S \cap \tilde{S}} (2x_{m+i} - x_i) + \sum_{i \in S\cap \tilde{S}} x_i x_{m+i}\right)}.
\end{align*}

Observe that upon fixing a value of $x_p$, the product above completely decouples into $m$ groups over $(x_i, x_{m+i}),$ which can then be summed separately. Indeed, \begin{align*}
    2^p \cosh^{p-1}(\alpha) \mathbb{E}_P\left[ \frac{Q_S(X) Q_{\tilde{S}}(X)}{P(X)^2}\right] &= \sum_{x_p} \prod_{i \in (S \cup \tilde{S})^c} \left( \sum_{x_i, x_{m+i}} \exp{\alpha (x_p x_i + x_i x_{m+i}) }\right) \\&\quad \times \prod_{i \in S \triangle \tilde{S}} \left( \sum_{x_i, x_{m+i}} \exp{\alpha( x_p x_{m+i} + x_i x_{m+i} }\right) \\ &\qquad \times \prod_{i \in S \cap \tilde{S}} \left( \sum_{x_i, x_{m+i}} \exp{\alpha( x_p(2x_{m+i} - x_i) + x_i x_{m+i}}\right).
\end{align*} 

There are three types of $i$ - those that lie in neither of $S, \tilde{S},$ those that lie in only one of these, and those that lie in both, which is how the above has been separated. We will explicitly compute the sum over $(x_i, x_{m+i})$ for each type separately.

\begin{enumerate}
    \item $i \in (S \cup \tilde{S})^c$: \begin{align*}
        \sum_{x_i, x_{m+i}} \exp{\alpha (x_p x_i + x_i x_{m+i}) } &= e^{\alpha(x_p + 1)} + e^{\alpha (x_p - 1)} + e^{\alpha(-x_p - 1)} + e^{\alpha(-x_p + 1)} \\
                                                                  &= 2e^{\alpha}\cosh(\alpha x_p) + 2e^{-\alpha} \cosh(\alpha x_p)\\
                                                                  &= 4\cosh^2(\alpha),
    \end{align*}
    where we have utilised the fact that $x_p \in \pm 1,$ and that $\cosh$ is an even function.
    \item $i \in S \triangle \tilde{S}$: This case is very similar to the above:\begin{align*}  \sum_{x_i, x_{m+i}} \exp{\alpha (x_p x_{m+i} + x_i x_{m+i}) } &= e^{\alpha(x_p + 1)} + e^{\alpha (-x_p - 1)} + e^{\alpha(x_p - 1)} + e^{\alpha(-x_p + 1)} \\
                                                                  &= 4\cosh^2(\alpha)\end{align*} 
    \item Finally, for $i \in S \cap \tilde{S}$,\begin{align*}
        \sum_{x_i, x_{m+i}} \exp{\alpha ( x_p(2x_{m+i} - x_i) + x_i x_{m+i} } &= e^{\alpha(x_p + 1)} + e^{\alpha (-3x_p - 1)} + e^{\alpha(3x_p - 1)} + e^{\alpha(-x_p + 1)} \\
                                                                              &= 2(e^{\alpha} \cosh(\alpha) + e^{-\alpha}\cosh(3\alpha)),
    \end{align*}                                                          
\end{enumerate}

Plugging the above calculations in, we find that \begin{align*} \mathbb{E}_P\left[ \frac{Q_S(X) Q_{\tilde{S}}(X)}{P(X)^2}\right] &= \sum_{x_p} \frac{(4\cosh^2\alpha)^{|(S \cup \tilde{S}^c| + |S \triangle \tilde{S}|} (2(e^{\alpha} \cosh(\alpha) + e^{-\alpha}\cosh(3\alpha)))^{|S\cap \tilde{S}|}}{2^p \cosh^{p-1}(\alpha)} \\ 
&= 2 \cdot \frac{(2\cosh(\alpha))^{2(m - |S \cap \tilde{S}|} (2(e^{\alpha} \cosh(\alpha) + e^{-\alpha}\cosh(3\alpha)))^{|S\cap \tilde{S}|}}{2^p \cosh^{p-1}(\alpha)} \\
&= \left( \frac{e^{\alpha} \cosh(\alpha) + e^{-\alpha}\cosh(3\alpha)}{2\cosh^2(\alpha)} \right)^{|S \cap \tilde{S}|},\end{align*} where we have used the fact that $(S \cup \tilde{S})^c, S\triangle \tilde{S}, S\cap \tilde{S}$ form a partition of $[1:m],$ and that $2m = p-1.$

To finish, we observe that \begin{align*}
    e^x \cosh(x) + e^{-x}\cosh(3x) - 2\cosh^2(x) &= \frac{e^{2x} + e^{-4x} -e^{-2x} - 1}{2}\\
                                                 &= e^{-x} \frac{ e^(3x) + e^{-3x} - (e^x + e^{-x})}{2}\\
                                                 &= e^{-x} (\cosh(3x) - \cosh(x))\\
                                                 &= e^{-x}(4\cosh^3(x) - 3\cosh(x) -\cosh(x))\\
                                                 &= 4e^{-x}\cosh(x) (\cosh^2(x) - 1) \\
                                                 &\le 4\sinh^2(x),
\end{align*} where the final relation uses $x \ge 0$.
\end{proof}
\end{itemize}

\subsubsection{Tolerant Testing of Forest Deletions, and of Trees}\label{appx:tolerant_testing}

\begin{proof}[Proof of Theorem \ref{thm:tolerant_tree}]

We repeatedly reuse the notation from the proof of Theorem \ref{thm:forest_testing} above. 

For the forest deletion setting, suppose $|G(P)| = k,$ and let $\widetilde{P}_{\Delta_0}$ be such that it's network structure is a deletion of most $\Delta_0 \le \varepsilon s$ edges from $G(P).$ It follows from the mean and variance calculations before, that, for any $\Delta \ge s,$ \begin{align*} 
\mathbb{E}_{\widetilde{P}^{\otimes n}_{\Delta_0}}[\mathscr{T}] &= (k - \Delta_0)\tau \ge (k - \varepsilon s )\tau, \\
\mathrm{Var}_{\widetilde{P}^{\otimes n}_{\Delta_0}}[\mathscr{T}] &= \frac{k(1-\tau^2) + \Delta_0 \tau^2}{n} \le \frac{k(1-\tau^2) + \Delta \tau^2}{n}.
\end{align*}

Consider the test which rejects the null hypothesis when $\mathscr{T} < (k - \frac{1+\varepsilon}{2}s)\tau$. Comparing the above to a $Q_{\Delta}$ as in the proof of Theorem \ref{thm:forest_testing}, and proceeding as in it, we find that the risk is appropriately controlled so long as the following relations hold for every $\Delta_0 \le \varepsilon s,$ and $\Delta \ge s$, where $C$ is an absolute constant: \begin{align*}
    n & \ge C\frac{k(\tau^{-2} -1) + \Delta_0 }{ \left(\frac{1 + \varepsilon}{2}s - \Delta_0\right)^2} \\
    n & \ge C\frac{k(\tau^{-2} -1) + \Delta }{ \left(\Delta -\frac{1 + \varepsilon}{2}s\right)^2} \\
\end{align*}

The right hand sides of the first and second equations above respectively increase and decrease with $\Delta_0$ and $\Delta$. Thus, setting $\Delta_0 = \varepsilon s$ and $\Delta = s,$ and taking the maximum possible $k = p$ tells us that the conditions will be met so long as \begin{align*}
    n \ge 4C \frac{ (p-1) \sinh^{-2}(\alpha) + s }{(1-\varepsilon)^2 s^2} 
\end{align*}

For the tree case, the same argument follows but with a small change - in the null case, a change of $\Delta_0$ edges can reduce the mean of $\mathscr{T}$ by $\Delta_0 \tau,$ but in the alternate, there may exist changes of $\Delta$ edges which only drop the mean of $\mathscr{T}$ by $\Delta/2 (\tau - \tau^2)$. Thus, we use the test \[ \mathscr{T} \overset{\mathrm{Null}}{\underset{\mathrm{Alt.}}{\gtrless}} (p-1)\tau - \frac{1 + 2\varepsilon}{4}s\tau + \frac{s}{4}\tau^2. \]

Continuing similarly, and keeping in mind that the variance of $\mathcal{T}$ after $\Delta$ changes is at most $(p-1)(1-\tau^2) + \Delta \tau^2,$ we find that risk of the above test is controlled so long as for every $\Delta_0 \le \varepsilon s,$ and for every $\Delta \ge s,$ the following relations hold \begin{align*}
    n &\ge \frac{C}{s^2} \frac{p(\tau^{-2} - 1) + \Delta_0}{ \left(  1 + 2\varepsilon - \tau - 4\Delta_0/s) \right)^2} \\
    n &\ge \frac{C}{s^2} \frac{p(\tau^{-2} - 1) + \Delta}{ \left( 2\Delta/s (1-\tau)  - (1 + 2\varepsilon-\tau)\right)^2}
\end{align*}

It is a matter of straightforward computation that if $\varepsilon \le \frac{1-\tau}{2},$ then the right hand sides of the first and second inequality above respectively increase and decrease with $\Delta_0$ and $\Delta$. Thus, setting $\Delta_0 = \varepsilon s$ and $\Delta =s,$ the above holds if \begin{align*}
    n &\ge \frac{C}{(1 - 2\varepsilon - \tau)^2}\left( \frac{p(\tau^{-2} -1)}{s^2} + \frac{1}{s} \right).\qedhere
\end{align*}

\end{proof}

\subsection{Testing Deletions in High-Temperature Ferromagnets}\label{appx:high_temp_ferro}

\subsubsection{Proof of achievability}
\begin{proof}[Proof of the upper bound of Theorem \ref{thm:high_temp_ferr_del}]

We follow the strategy laid out in the main text. The proposed test statistic is $\mathscr{T}(\{X^{(i)}\};P) := \widehat{\mathbb{E}}[ \sum_{(i,j) \in G(P)} X_iX_j],$ where the $\{X^{(i)}\}$ are the samples, and $\widehat{E}$ indicates the empirical mean over this data. Concretely, the test is to threshold $\mathscr{T}$ as \[ \mathscr{T} \overset{\mathrm{Null}}{\underset{\mathrm{Alt.}}{\gtrless}} \mathbb{E}_P[\mathscr{T}] - Cs \alpha,\] where $C$ the constant left implicit in Lemma \ref{lem:high_temp_ferr_sep}.

The analysis relies on two facts:
\begin{mylem}\label{lem:high_temp_ferr_sep}
    Let $P,Q \in \mathcal{H}_d^\eta(\alpha),$ and $G(Q) \subset G(P),$ with $|G(P) \triangle G(Q)| \ge s$. For every $\eta < 1,$ there exists a constant $C>0$ that does not depend on $(p,s,\alpha)$ such that \[ \mathbb{E}_P[\mathscr{T}] - \mathbb{E}_Q[\mathscr{T}] \ge 2C s \alpha.\]
\end{mylem}

\begin{mylem}\label{lem:high_temp_ferr_variance}
    For any $P,Q \in \mathcal{H}_d^\eta(\alpha),$ which may be equal, \[ \mathrm{Var}_Q\left[ \sum_{(i,j) \in G(P)} X_i X_j\right] \le C_\eta pd,\] where $C_\eta$ may depend on $\eta,$ but not otherwise on $(p,d,s,\alpha)$. 
\end{mylem}

Applying the variance contraction over $n$ independent samples, we find via a use of Tchebycheff's inequality that the following event have probability at least $1/8$ for the respective hypotheses: \begin{align*}
    \textit{Null:}& \quad \mathscr{T} \ge \mathbb{E}_P[ \mathscr{T}] - C_\eta \sqrt{\frac{8pd}{n}}, \\
    \textit{Alt:}& \quad \mathscr{T} \le \mathbb{E}_P[\mathscr{T}] - Cs \alpha + C_\eta \sqrt{\frac{8pd}{n}}.
\end{align*}

Thus, taking $n$ so large that $Cs \alpha > C_\eta \sqrt{\frac{8pd}{n}},$ the false alarm and missed detection probabilities are both controlled below $1/8,$ yielding the claimed result. \qedhere

It of course remains to argue the above lemmata. These are both essentially utilisations of existing results. 

\begin{proof}[Proof of Lemma \ref{lem:high_temp_ferr_sep}]
    We use the fact that in ferromagnetic models, the correlations between any pair of nodes increases as the weights increase (or contrapositively, if weights are deleted, then correlations must decrease). This is classically shown via (a special case of) Griffith's inequality \cite{griffiths1969rigorous}, which claims that for any $u,v,i,j,$ in a ferromagnetic Ising model, $\mathbb{E}[X_uX_v X_i X_j] \ge \mathbb{E}[X_u X_v] \mathbb{E}[X_i X_j]$. This is relevant here due to the fact that \begin{align*} \partial_{\theta_{ij}} \mathbb{E}_{P_\theta}[X_u X_v] &= \partial_{\theta_{ij}} \frac{1}{Z_\theta}  \sum_x x_u x_v \exp{\sum_{s< t} \theta_{st} X_s X_t} \\ 
                                                         &\overset{a}{=} \frac{1}{Z_\theta} \sum_x x_u x_v x_i x_j \exp{\sum_{s< t} \theta_{st} X_s X_t} \\ &\quad - \frac{1}{Z_\theta^2} \left( \sum_x x_u x_v \exp{\sum_{s< t} \theta_{st} X_s X_t}\right) \left( \sum_x x_u x_v \exp{\sum_{s< t} \theta_{st} X_s X_t} \right) \\
                                                         &= \mathbb{E}[X_uX_vX_iX_j] - \mathbb{E}[X_uX_v] \mathbb{E}[X_iX_j] \ge 0.\end{align*}

Above, equality $(a)$ is a consequence of the quotient rule, and the fact that $Z_\theta = \sum_x \exp{\sum_{s < t} \theta_{st} x_s x_t}$.

Next, we utilise the following structural lemma, due to Santhanam and Wainwright. While we cite it as a variation on their Lemma 6 below, more accurately this arises via a correction of a subsidiary part of the proof of the same lemma. In particular, we are utilising a corrected version of the unlabelled inequality on Page 4131 that follows the inequality (51), with further specialisation to the high-temperature deletion with a uniform edge weight context.

\begin{mylem} \emph{(A variation of Lemma 6 of \cite{SanWai})} \label{key_lemma_high_temp_ferr}
    Let $P \in \mathcal{H}_d^\eta(\alpha),$ and $Q$ be obtained by removing the edge $(a,b)$ from $P$. Then \[ \mathbb{E}_P[ X_aX_b] - \mathbb{E}_{Q}[ X_aX_b] \ge \frac{\alpha}{400}. \]
\end{mylem}

With this in hand, we develop our result by arguing over each deleted edge in a sequence. For a given $P$ and $Q$, such that $Q$ occurs by deleting $\Delta \ge s$ edges from $P$, take a chain of laws $P = Q_0, Q_1, Q_2, \dots, Q_\Delta = Q,$ where each $Q_{t+1}$ is obtained by deleting one edge from $Q_t$. Let $(i_{t+1}, j_{t+1})$ be the edge deleted in going from $Q_t$ to $Q_{t+1}$  Since each model is ferromagnetic, and each $Q_{t+1}$ deletes an edge from $Q_t$, we find that \begin{align*} \mathbb{E}_{Q_t}\left[ \sum_{(i,j) \in G(P)} X_iX_j\right] - \mathbb{E}_{Q_{t+1}}\left[ \sum_{(i,j) \in G(P)} X_iX_j\right] &\ge \mathbb{E}_{Q_t}\left[X_{i_{t+1}}X_{j_{t+1}}\right] -  \mathbb{E}_{Q_{t+1}}\left[X_{i_{t+1}}X_{j_{t+1}}\right] \\ &\ge \frac{\alpha}{400}. \end{align*}

Summing up the left hand side over $t = 0$ to $\Delta-1$ leads to a telescoping sum, while $\Delta \ge s$ copies of the right hand side get added, directly leading to our conclusion \begin{align*} \mathbb{E}_{P}\left[ \sum_{(i,j) \in G(P)} X_iX_j\right] - &\phantom{=} \mathbb{E}_{Q}\left[ \sum_{(i,j) \in G(P)} X_iX_j\right] \\ &= \mathbb{E}_{Q_0}\left[ \sum_{(i,j) \in G(P)} X_iX_j\right] - \mathbb{E}_{Q_{\Delta}}\left[ \sum_{(i,j) \in G(P)} X_iX_j\right] \\
&= \sum_{t = 0}^{\Delta-1} \mathbb{E}_{Q_t}\left[ \sum_{(i,j) \in G(P)} X_iX_j\right] - \mathbb{E}_{Q_{t+1}}\left[ \sum_{(i,j) \in G(P)} X_iX_j\right]\\
&\ge \sum_{t = 0}^{\Delta-1} \frac{\alpha}{400} = \Delta\frac{\alpha}{400} \ge s \frac{\alpha}{400}.\qedhere \end{align*}
\end{proof}

To complete the proof, we prove the key lemma used in the above argument.

\begin{proof}[Proof of Lemma \ref{key_lemma_high_temp_ferr}] We note that this proof assumes familiarity with the proof of Lemma 6 of \cite{SanWai}. The main reason is that the proof really consists of fixing an equation in the proof of this result, and then utilising the ferromagnetic properties a little. As a result, there is no neat way to make this proof self contained (reproducing the proof of the aforementioned lemma is out of the question, since this is a long and technical argument in the original paper). With this warning out of the way, let us embark.

Let $\partial a$ and $\partial b$ be the neighbours of, respectively, $a$ and $b$ in $G(P)$ (which, since $G(Q)$ only deletes $(a,b)$ from $G(P)$, contain all the neighbours of $a$ and $b$ in $G(Q)$ as well).

Before proceeding, we must first point out a (small) error in the proof of Lemma 6 in \cite{SanWai}. The clearest way to see this error is to note the inequality following equation (51) in the text, which claims that if $(a,b) \in G(P) \triangle G(Q),$ then some quantity ($J$ in the paper) known to be positive is upper bounded by \[ J \le \sum_{u \in \partial a \setminus \{b\}} ( \{\mathbb{E}_P - \mathbb{E}_Q\}[X_uX_a] ) (\theta^P_{ua} - \theta^Q_{ua}) + \sum_{v \in \partial b \setminus \{a\}} ( \{\mathbb{E}_P - \mathbb{E}_Q\}[X_vX_b] ) (\theta^P_{vb} - \theta^Q_{vb}).\]
Note that we have specialised the above to the case where $G(Q) \subset G(P)$. Now, observe than when the only change made is in the edge $(a,b)$, then the above upper bound is $0$. Indeed, $\theta^P_{ua} = \theta^Q_{ua}$ for every $u \in \partial a \setminus \{b\},$ since none of these edges have been altered, making the first sum $0$, and similarly the second, contradicting the claim that the sum is bigger than $J$ (which is positive). The error actually lies a few lines up, in the decomposition for the term $\Delta(\theta,\theta')$, which along with the claimed terms, should also include the term $(\{\mathbb{E}_P - \mathbb{E}_Q\}[X_aX_b])(\theta^P_{ab} - \theta^Q_{ab}),$ which is missing from the text of \cite{SanWai}. This term is present since the $P_{\theta[x_C]}$ and $P_{\theta'[x_C]}$ are, of course, laws on $X_a$ and $X_b$, and thus have $\theta^P_{ab}x_ax_b$ and $\theta^Q_{ab}x_ax_b$ in the Ising potentials.\footnote{note however that exactly one of $\theta^P_{ab}$ and $\theta^Q_{ab}$ is zero, since $(a,b)$ lies in one but not the other graph.} Putting this term back in, the correct equation is that \begin{align*} \kappa &\le (\{\mathbb{E}_P - \mathbb{E}_Q\}[X_aX_b])(\theta^P_{ab} - \theta^Q_{ab}) + \sum_{u \in \partial a \setminus \{b\}} ( \{\mathbb{E}_P - \mathbb{E}_Q\}[X_uX_a] ) (\theta^P_{ua} - \theta^Q_{ua}) \\ &\qquad + \sum_{v \in \partial b \setminus \{a\}} ( \{\mathbb{E}_P - \mathbb{E}_Q\}[X_vX_b] ) (\theta^P_{vb} - \theta^Q_{vb}),\end{align*} where $\kappa$ is the lower bound on $J$, that is (specialised to our case of uniform weights), \[ \kappa = \frac{\sinh^2(\alpha/4)}{1 + 3\exp{\alpha d}}. \]

We note that the conclusion of Lemma 6 of \cite{SanWai} is not affected by the above error\footnote{The expression $2\alpha d\max_{u \in \{a,b\}, v \in V} |\mu_{uv} - \mu'_{uv}|$ already accounts for the extra term we add, since it allows us to take $u = a, v = b.$}.

With this out of the way, we may now argue our point. In our case, we know that since only the edge $(a,b)$ has been altered, the second and third terms in the updated sum are $0$. Further, we know that $\theta^P_{ab} = \alpha \ge 0,$ and $\theta^Q_{ab} = 0.$ Thus, we conclude that \[ \mathbb{E}_P[X_aX_b] - \mathbb{E}_Q[X_aX_b] \ge \frac{\kappa}{\alpha} \ge \frac{\sinh^2{\alpha/4}}{\alpha (1 + 3\exp{2\alpha d})}.\] 

Finally, we use our high temperature condition. Firstly, note that $\alpha d \le \eta < 1,$ and thus $(1 + 3\exp{2\alpha d}) \le 1 + 3e^2 \le 24.$ Next, since $\sinh(x) \ge x,$ $\sinh^2(\alpha/4) \ge \alpha^2/16.$ Putting these together, we find that \[ \mathbb{E}_P[X_aX_b] - \mathbb{E}_Q[X_aX_b] \ge \frac{\alpha^2/16}{\alpha \cdot 24} = \frac{\alpha}{384} \ge \frac{\alpha}{400} \qedhere\] 
\end{proof}

\begin{proof}[Proof of Lemma \ref{lem:high_temp_ferr_variance}]
    We directly utilise the concentration result  \cite[Ex. 2.5]{adamczak2019note}, which shows that for bilinear forms $f(X) = \langle A , XX^{\tpose}\rangle,$ where the inner product is the Frobenius dot product, and for a high temperature Ising model $P$, there exists a $C_\eta$ depending only on $\eta$ such that\footnote{Instead of the Frobenius norm $\|A\|_F,$ the bound of \cite{adamczak2019note} features the Hilbert-Schmidt norm of $A$. These are the same thing for finite dimensional operators.} \[ P(|f - \mathbb{E}[f]| \ge t) \le 2\exp{- \frac{t}{C_\eta\|A\|_F}}.\]
    Via the standard integral representation $\mathbb{E}[(f-\mathbb{E}[f])^2] = \int_0^\infty P(|f- \mathbb{E}[f]|^2 \ge r)\mathrm{d}r$ and the above upper bound, we directly obtain that the variance of any $f$ such as the above is bounded by $3\|A\|_F^2 C_\eta^{2}$. 
    
    Now, out statistic is a bilinear function of the above form. Indeed, \[ \sum_{(i,j) \in G(P)} X_i X_j = \langle G(P)/2, XX^\tpose\rangle,\] where we treat $G(P)$ as it's adjacency matrix, and thus we immediately obtain that the variance is bounded by $1.5 C_\eta^2 \|G(P)\|_F^2.$ Notice that $\|G(P)\|_F^2$ is merely twice the number of edges in $G(P)$, and since this has degree at most $d$, this number is at most $2pd$. The claim follows.
\end{proof}

\end{proof}

\subsubsection{Proof of Lower Bounds}

The lower bounds are argued using Thm. \ref{thm:n_bound_tech}, with the widget(s) that consist of comparing a full clique to an empty graph, which of course satisfy the constraint that the alternate models are derived by deleting edges from the null graph. Concretely, we use the bound of Proposition \ref{prop:clique_vs_empty}, to show the following result \begin{myprop}
    Suppose $s \le pd/K$ for large enough $K$ and $\alpha d \le \eta \le 1/32$. Then there exists a $C$ independent of all parameters such that \begin{align*}
        n_{\gof, \delete}(s, \mathcal{H}_d^\eta(\alpha))&\ge \max_{s/Kp \le  k \le d} \frac{1}{Ck^2\alpha^2} \log \left( 1 +  \frac{pk^3}{Cs^2} \right),\\
        n_{\gof, \delete}(s, \mathcal{H}_d^\eta(\alpha)) &\ge \max_{s/Kp \le  k \le d} \frac{1}{Ck^2\alpha^2} \log \left( 1 +  \frac{pk}{Cs} \right),
    \end{align*}
    where the maximisation is over integers $k \ge 2$ in the stated ranges. In particular, the bounds in the main text correspond to taking $k = d$.
    
\begin{proof}
The proof relies on the fact that if $\alpha d \le 1/32,$ then $\alpha k \le 1/32$ for any $k\le d$ as well, which allows us to utilise Prop.~\ref{prop:clique_vs_empty} for each $k$. For each valid choice of $k$, we take $P_0$ to be the Ising model on the complete graph on $k$ nodes with uniform edge weight $\alpha$, and $Q_0$ to be the Ising model on the empty graph on $k$ nodes. The relevant quantities are $\sigma = \binom{k}{2},$ $m = \lfloor p/k \rfloor,$ and $t = \lceil s/\binom{k}{2}\rceil ,$ with the total number of changes lying between $s$ and $2s$. Repeated use of Thm.\ref{thm:n_bound_tech} concludes the argument. 
\end{proof}  

\end{myprop}

\subsection{Simulation Details}\label{appx:exp}

Details about the generation of Figure \ref{fig:forest_plot} are as follows:
\begin{itemize}[leftmargin = 0.16in]
    \item \textbf{Sampling from Ising Models} Samples from Ising models were generated by running Glauber dynamics for 1600 steps. This number is chosen to be four times the `autocorrelation time', which is the time at which the autocorrelation of the test statistic $\langle XX', G\rangle/2$ drops to near $0$, and serves as a proxy for the mixing time of the dynamics (at least for the relevant statistics). Note that raw samples were outputted from the dynamics (i.e., we did not take ergodic averages).
    \item \textbf{Constructing $P$s and $Q$s} Throughout, $P$ was the Ising model on a complete binary tree on $127$ nodes. For each value of $s$ and each experiment, $s$ random edges from this tree were deleted.
    \item \textbf{Experiment Structure} For each value of $s \in \{3,6,\dots, 60\}$ and $n \in \{20, 40, \dots, 480\},$ we carried out a simulation of the GoF testing risk of our statistic for $s$ deletions using $n$ samples. We refer to each of these as an experiment. Each experiment was carried out by running 100 independent tests (on independent data), which each consisted of two parts - first we generated samples from $P$, and declared a false alarm if $\mathscr{T}$ fell below $(p-1-s/2)\tanh(\alpha)$ for this. Next, we generated a $Q$ by deleting $s$ edges, and then generated samples from $Q$, and finally declared a missed detection if $\mathscr{T}$ was above the same threshold. Risks were computed by adding up the total number of false alarm and missed detection events in these 100 runs, and dividing them by 100. 
    \item \textbf{Structure of Figure \ref{fig:forest_plot}} Each box in the figure corresponds to a simulation for $s$ changes and $n$ nodes, where $(s,n)$ are the coordinates of the upper right corner of the box. The boxes are coloured according to their empirical risk - if this was greater than $0.35,$ then the box was coloured black; if smaller than $0.15,$ then coloured white, while if it was between these values, the box was coloured orange. 
\end{itemize}

Additionally, we note that structure learning performs very poorly for this setup. This is best illustrated by the Figure \ref{fig:chow-liu-error}, which shows the number of edge-errors (i.e. $|G(P) \triangle \hat{G}|$) versus the sample size when the Chow-Liu algorithm was run on data generated by the null model (i.e., the full binary tree). The Chow-Liu algorithm was run by computing the covariance matrix, and computing the weighted maximum spanning tree for it via the library methods in MATLAB. The number of errors is again averaged over 100 trials. This demonstrates that the na\"{i}ve scheme of recovering the graph and testing against it is infeasible for $s \le 60 $ if $n \le 1500$, empirically demonstrating the separation between structure learning and testing. 
\begin{figure}
    \centering
    \includegraphics[width = 0.5\textwidth]{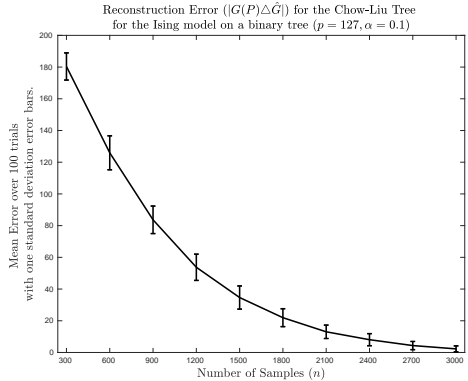}
    \caption{Reconstruction Error of the Chow-Liu Tree for the Ising model on a complete Binary Tree with $p = 127, \alpha = 0.1$.}
    \label{fig:chow-liu-error}
\end{figure}

\section{Widgets}\label{appx:widgets}

As discussed in the previous section, we will utilise Lemma \ref{lem:lifting}, in order to do which we need to provide specific instances of $(P_0, Q_0)$ that are close in $\chi^2$-divergence. We will abuse terminology and call this pair an ensemble. This section lists a few such pairs of graphical models, along with the $\chi^2$-divergence control we offer for the same, proofs for which are left to \S\ref{appx:widget_proofs}. Throughout, we will use $\lambda$ and $\mu$ as weights of edges, with $\lambda \ge |\mu| >0.$ I the proofs of the theorems, we will generally set $\lambda = \beta$ and $\mu = \alpha,$ but retaining these labels aids in the proofs of $\chi^2$-divergence control offered for these widgets.

\subsection{High-Temperature Obstructions}

The following graphs are used to construct obstructions in high temperature regimes. The first is the triangle graph, as described in \S\ref{sec:pf_tech}. The second is a full clique in high temperatures. The latter section is derived from the bounds of \cite{cao2018high}. 

\subsubsection{The Triangle}\label{sec:triangle}

We start simple. Let $P_{\textrm{Triangle}}$ be the Ising model on $3$ nodes with edges $(1,2)$ and $(2,3),$ each with weight $\lambda,$ and $Q_{\textrm{Triangle}}$ be the same with the edge $(1,3)$ of weight $\mu$ appended (see Figure \ref{fig:triangle}). The bound below follows from an explicit calculation, which is tractable in this small case.

\begin{figure}[ht]
    \centering
    \includegraphics[width = .5\textwidth]{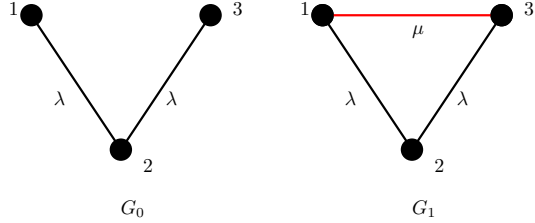}
    \caption{Ensemble used for Proposition \ref{prop:triangle} }
    \label{fig:triangle}
\end{figure}

\begin{myprop}\label{prop:triangle}
    For $\lambda \ge |\mu| >0,$ \[ \chi^2( Q_{\mathrm{Triangle}} \| P_{\mathrm{Triangle}} ) \le  {8 e^{-2\lambda}}{ \tanh^2\mu } . \]
\end{myprop}

\subsubsection{Full Clique versus Empty Graph}\label{sec:clique_high_temp}

\cite{cao2018high} shows the remarkable fact that high-temperature cliques are difficult to separate from the empty graph. We present this result below.
\begin{myprop}\label{prop:high_temp_clique}
    Let $P$ be the Ising model on the empty graph with $k$ nodes, and let $Q$ be the Ising model on the $k$-clique, with uniform edge weights $\mu$. If $32\mu k \le 1,$ then \[\chi^2(Q\|P) \le  3k^2 \mu^2. \]
\end{myprop}

In the notation of \cite{cao2018high}, this is the bound at the bottom of page 22, instantiated with $G = G'$ and the $\mathcal{R,B}, \Gamma$ values as determined in the proof of Example 2.7.

We will also utilise the following reversed $\chi^2$-divergence bound. This is not formally shown in \cite{cao2018high}, and thus, we include a proof of the same, using the techniques of the cited paper, in \S\ref{appx:pf_of_clique_versus_empty}.
\begin{myprop}\label{prop:clique_vs_empty}
    Let $P$ be the Ising model on a clique on $m$ nodes with uniform edge weights $\mu$, and let $Q$ be the Ising model on the empty graph on $m$ nodes. If $32\mu m \le 1,$ then \[ \chi^2(Q\|P) \le 8( \mu m)^2.\] 
\end{myprop}

\subsubsection{Fan Graph}\label{sec:fan}

This widget is not required for the main text, although it may serve as a more involved construction to show the bounds of Thms.~\ref{thm:lb_small_s} and \ref{thm:very_small_s}. Its main use is in Appendix \ref{appx:prop_test}, where it is used to show an obstruction to testing of maximum degree in a graph.

Generalising the triangle of the previous section, we may hang many triangles from a single vertex, getting a graph that resembles an axial fan with many blades. Using such a graph, we may demonstrate high-temperature obstructions to determining the maximum degree of a graph. 

Concretely, for a natural $B$ we define a fan with $B$ blades to be the graph on $2B+1$ nodes where, nodes $[1:2B]$ are each connected to the central node $2B+1,$ and further, for $i\in [1:B],$ nodes $2i$ and $2i - 1$ are connected. We call the edges incident on the central node ($B+1$) axial, and the remaining edges peripheral.

Treating $\ell$ as a parameter, the Ising models $P_{\ell, \mathrm{Fan}}$ and $Q_{\ell, \mathrm{Fan}}$ are determined as followed: 

\begin{itemize}
    \item $Q_{\ell, \mathrm{Fan}}$ places a weight $\lambda$ on each peripheral edge, and a weight of $\mu$ on each axial edge.
    \item $P_{\ell, \mathrm{Fan}}$ `breaks in half' $\ell$ of the blades in the graph - concretely, for $i \in [1:\ell],$ we delete the edges $\{2i-1, 2B+1\}.$ 
\end{itemize}

\begin{figure}[ht]
    \centering
    \includegraphics[width = .4\textwidth]{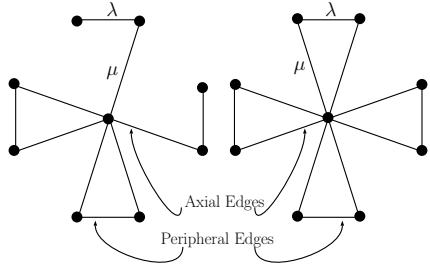}
    \caption{The Fan graphs for $P_{\ell, \mathrm{Fan}}$ (left) and $Q_{\ell, \mathrm{Fan}}$ (right) in the setting $B = 4, \ell = 2$.}
    \label{fig:fan}
\end{figure}

Viewing $P$ as the null graph, note that in $Q$ we have added an excess of $\ell$ edges, and increased the degree of the central node from $2B - \ell$ to $2B$. The fan graph serves as a high-temperature obstruction to determining the maximum degree of the graph underlying an Ising model via the following claim.

\begin{myprop}\label{prop:fan}
    For $\ell \le B,$  if $\lambda\mu \ge 0,$ then \[ \chi^2( Q_{\ell, \mathrm{Fan}} \| P_{\ell, \mathrm{Fan}} ) \le \left( 1 + {16 e^{-2\lambda}}\tanh^2\mu \right)^{\ell} - 1.\]
\end{myprop}

\subsubsection{Single Edge}\label{sec:single_edge}

This construction is possibly the simplest, and is used to show the lower bound in Thm.~\ref{thm:forest_testing}. We consider the two possible Ising models on two nodes - $P$ is the one with an edge, of weight $\mu$, while $Q$ has no edges. The characterisation is trivial, and we omit the proof. \begin{myprop}\label{prop:edge_vs_empty}
    \(\chi^2(Q\|P) = \sinh^2(\mu) \).
\end{myprop}

\subsection{Low-Temperature Obstructions via Clique-Based Graphs}

The computations in this and the subsequent cases are rather more complicated that in the previous case, and will intimately rely on a `low temprature' assumption. The basic unit is that of a clique on some $d+1 \gg 1$ nodes, in the setting of temperature $\lambda d \ge \log d$. 

The intuition behind these is rather simple - Ising models on cliques tend to `freeze' in low temprature regimes, i.e.~ the distribution concentrates to the states $\pm (1, 1, \dots, 1)$ with probability $1 - \exp{ - \Omega(\beta d))}$ for $\beta d \gg 1.$ This effect is fairly robust, and dropping or adding even a large number of edges does not alter it significantly. Thus, one has to collect an exponential in $\beta d$ number of samples merely to obtain some diversity in the samples, which will be necessary to distinguish any of these variations of a clique from the full thing. 

While generic arguments can be offered for each of the settings below on the basis of the above intuition, these tend to be lossy in how they handle the effect of very low edge weights. To counteract this, we individually analyse each setting, and while the arguments have structural similarities, the particulars vary.

It is worth noting that our bounds rely on  below diverge from the classical literature in the low temperature condition we impose. We generally demand conditions like $\beta d \ge \log d,$ while most other lower bounds demand that $\beta d \ge 1.$ This extra room allows us to tighten the exponents in the sample complexity bounds as opposed to previous work, but has the obvious disadvantage of reduced applicability. We note, however, that in most settings, this only yields a lost factor of $d$ in the resulting bounds, and frequently not even that. Functionally, thus, there is little to no loss in the use of this stronger low-temperature condition.\footnote{This effect is linked to the concentration of the Ising model on the clique we mentioned before. Notice that the probability of a uniform state is as $1 - \exp{-\Omega(\beta d)}$. For this to be appreciable, i.e., at least polynomially close to $1$, a condition like $\beta d = \Omega(\log d)$ is in fact necessary.} A similar notion of low temperature has appeared in e.g.~\cite{bezakova_lowbd}.

\subsubsection{Clique with a deleted edge}\label{sec:1edgeclique}

This calculation is the simplest demonstration of our bounding technique, and all following settings are analysed in a similar way. While it is superseded by the section immediately following it, the bound is thus important for the reasons of comprehension if nothing else.

We consider graphs on $d+1$ nodes, and let $P_{\mathrm{Clique}}$ be the Ising model on the complete graph on $d+1$ nodes, with edge $(1,2)$ of weight $\mu,$ and every other edge of weight $\lambda.$ $Q_{\mathrm{Clique}}$ is formed by deleting the edge $(1,2)$ in $P_{\mathrm{Clique}}$\begin{figure}[ht]
    \centering
    \includegraphics[width = .75\textwidth]{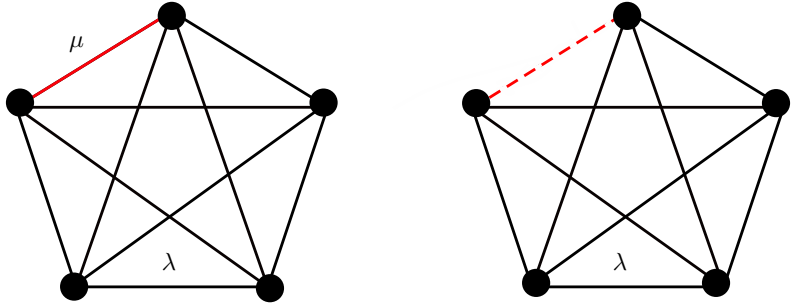}
    \caption{The clique with uniform weight $\lambda$ barring one edge, and the same edge deleted. Here $d = 4.$}
    \label{fig:clique_one_edge_gone}
\end{figure} Note that such underlying constructions feature in nearly every lower bound on structural inference on degree bounded Ising models.

With the exposition out of the way, we state the bound below. \begin{myprop}\label{prop:1edgeclique}
    Suppose $\lambda d > \log d.$ Then \[ \chi^2( Q_{\mathrm{Clique}} \|P_{\mathrm{Clique}}) \le  16e^{-2\lambda (d-1)} \sinh^2\mu . \]  
\end{myprop}

\subsubsection{The clique with a large hole}\label{sec:cliquehole}

To allow for a greater number of changes, we modify the previous construction by removing a large subclique from the $K_{d+1}$ used above, instead of just one edge. More formally, for some $ \ell < d/8,$ let $K_\ell$ be the complete graph on nodes $[1:\ell].$ We set   $P_{\ell, \mathrm{Clique}}$ to the the Ising model on $K_{d+1}$ such that the edges in $K_\ell$ have weight $\mu,$ and all other edges have weight $\lambda,$ while $Q_{\ell, \mathrm{Clique}}$ instead deletes the edges in $K_\ell$. Note that as a conseuquence, we have effected a deletion of $\sim \ell^2/2$ edges from the original model. 

\begin{myprop}\label{prop:clique_hole_Q||P}
    If $\ell +1 \le d/8,$ $\lambda \ge |\mu|$ and $\lambda d > 3\log d,$ then \[\chi^2(Q_{\ell,\mathrm{Clique}} \|P_{\ell,\mathrm{Clique}}) \le  32\ell e^{-2\beta (d + 1- \ell)}  \sinh^2 (\mu(\ell-1) ). \]
\end{myprop}

Note that the bound of the previous subsection (up to some factors) can be recovered by setting $\ell = 2$ in the above.

Control on the $\chi^2$-divergence with $P$ and $Q$ exchanged is also useful. 
\begin{myprop}\label{prop:clique_hole_P||Q}
    If $\ell +1 \le d/12,$ $\lambda \ge |\mu|$ and $\lambda d > 3\log d,$ then \[\chi^2(P_{\ell,\mathrm{Clique}} \|Q_{\ell,\mathrm{Clique}}) \le  64\ell e^{-2\beta (d + 1- \ell)} \sinh^2(2\mu(\ell -1)) . \]
\end{myprop}

\subsubsection{Emmentaler Clique}\label{sec:emmen}

As a development of the Clique with a large hole, we may in fact put in many large holes, leading to a pockmarked clique reminiscent of a Swiss cheese. Concretely, let $\ell$ be a number such that $B:= d/(\ell + 1)$ is an integer. We define a graph on $d$ nodes in the following way: Divide the nodes into $B$ groups of equal size, $V_1, \dots, V_B$. Form the complete graph on $d$ nodes, and then delete the $\ell+1$-sublique on $V_i$ for each $i$. Note that equivalently, the graph above is the complete symmetric $B$-partite graph on $d$ nodes. The graph effects a deletion of $\sim d \ell/2$ edges from a clique.

\begin{figure}[h]
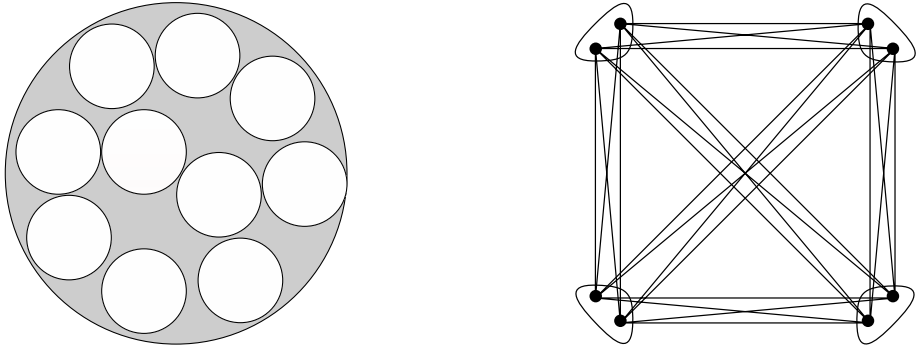

    \centering
    \includegraphics[height = .2\textheight]{Emmen_2.eps}
    ~\hspace{1in}~
    \includegraphics[height = .2\textheight]{Emmen_as_l-partite.eps}
    \caption{Two views of the Emmentaler cliques. The left represents the base clique as the large grey circle, while the uncoloured circles within represent the groups $V_i$ with no edges within (this should be viewed as $\ell \gg 1, B = 10$). This view is inspiration for the name. On the right, we represent the Emmentaler as the graph $K_{\ell + 1, \ell + 1, \dots, \ell + 1}$ - here $d = 8$ and $\ell = 1$ is shown.}
    \label{fig:emmen}
\end{figure}

The key property of the Emmentaler is that it still freezes at a exponential rate, and it has sufficient `space' in it to accommodate significantly more edges. In particular, the graph is regular and the degrees of each node are uniformly $d- \ell - 1.$ We use this in two ways:

\paragraph{Emmentaler with one extra node} We show that determining the degree of a node connected to many of the nodes of an Emmentaler is hard. Concretely, we construct the following two graphs on $d+1$ nodes:

Construct an Emmentaler Clique on the first $d$ nodes. Next, connect the node $d+1$ to each node in $\bigcup_{i = 1}^{B-1} V_i.$ Notice that node $d+1$ is not connected to one of the parts of the Emmentaler. We choose $P_{\ell}$ to be the Ising model with uniform weight $\lambda$ on the this graph. For $Q_{\ell}$, we additionally add edges between node $d+1$ and each node in $V_B$ with weight $\mu.$ The following result holds.

\begin{myprop}\label{prop:emmen_with_extra_node}
    If $2\le \ell+1 \le d/4$ and $\lambda (d-4) \ge 3\log d,$ and $|\mu| \le \lambda,$ then \[ \chi^2(Q_{\ell} \| P_{\ell} )  \le  32de^{-2\lambda (d-1 -\ell)} .\]
\end{myprop}

Notice that the above proposition does not show a $\mu$ dependence. This is due to inefficiencies in our proof technique. We strongly conjecture that a bound of the form $(1 + C d\tanh^2(\mu(\ell + 1)) e^{-2\lambda(d - \ell - 1)})^n$ holds.

\paragraph{Emmentaler v/s Full Clique} We show that it is difficult to distinguish between an Emmentaler and a full clique. Concretely, we let $P_\ell$ be an Emmentaler as above, and in $Q_\ell$, we add back the deleted subcliques  to each $V_i,$ but with weight $\mu.$ \begin{myprop}\label{prop:emmen_vs_full}
    If $\ell + 1 \le d/4$ and $\lambda(d-4) \ge 3 \log d,$ then \[ \chi^2(Q_{\ell}\| P_{\ell} )  \le  d^2\min(1, \mu^2d^4)e^{-2\lambda (d-1 -\ell)}.\]
\end{myprop}

\section{Miscellaneous}

\subsection{Using statistical formulations to test structural changes}\label{appx:stat_div}

The main text makes the case that statistical formulations of $\gofbb$ do not give us the whole story when one is interested in structural changes. Concretely, though, this only directly affects the lower bounds. On the other hand, when we restrict alternate hypotheses in the $\gofbb$ problem to make a lot of changes, then one may expect that tests under statistical formulations are powerful.

Intuitively, this expectation is rendered plausible by the fact that the notion of being close to a given model is similar under the statistical and the structural formulations - equality under one is also equality under the second, at least in the setting of unique network structures, and mere continuity suggests that, at least locally, setting some value of $s(P, \varepsilon)$ or $\varepsilon(s, P)$ should allow one to translate tests from the statistical to the structural notions of changes and vice versa.\footnote{It should be noted that this analogy is flawed - while the notions of being close are indeed similar, the notion of being far from a model is significantly different under the two formulations. The main text mentions an example illustrating this - if a small group of disconnected nodes is bunched into a clique, a large statistical change is induced due to the marked difference in the marginal law of this group, but the structural change is tiny. Of course, being close and far are ultimately related concepts, and some shadow of this effect must be cast on the closeness argument we have just presented.} However, this strategy does not work too well, at least with our current understanding of Ising models. More concretely - utilising statistical tests for structural testing in a sample efficient way requires a \emph{local} understanding of the distortion of the edge-Hamming distance of the graph under the map $(\theta, \theta') \mapsto \mathrm{SKL}{(\theta\|\theta')},$ which is not available as of now. Global constraints on the same are available, and are unhappily both rather pessimistic, and essentially tight. This means that using the methods developed for testing for statistical divergences in the setting of structural identity testing is problematic.

Some details - the best available results that translate edge-differences to symmetrised KL divergence is via Lemma 4 of \cite{SanWai}. The Bhattacharya coefficient of two distributions is $\mathrm{BC}(P,Q) := \sum_x \sqrt{P(x) Q(x)}.$ The cited lemma argues that under $s$ changes, \[\mathrm{BC} \le \exp{-C s \sinh^2(\alpha) e^{-2\beta d}/d }.\] Let $-\varphi$ denote the exponent in the above, for conciseness. Since $ -2\log \mathrm{BC} \le \mathrm{KL},$ this induces $D_{\mathrm{SKL}} \gtrsim \varphi,$ and similarly, since $1 - \mathrm{BC} \le \mathrm{TV},$ this tells us also that $\mathrm{TV} \ge 1 - \exp{-\varphi}.$ Since $1 - e^{-z} \le z,$ this means that the best lower bound we can possibly derive this way is $\mathrm{TV} \ge \varphi.$  

Now, the best known upper bounds for statistical testing under SKL is $(\beta pd /\varepsilon)^2$ up to log factors \cite{daskalakis2016testing}, and under TV for ferromagnets this may be improved to $(pd/\varepsilon)^2$ \cite{bezakova_lowbd}. Plugging in the values of $\varepsilon$ implicit in the above, the first of these then requires about \[ \left( \frac{\beta pd}{\varphi} \right)^2 \sim \frac{e^{4\beta d}}{\alpha^4} \left( \frac{\beta pd^2}{s}\right)^2,\] which is worse than the testing by first recovering the underlying network. Similarly, under TV, a similar number is required, but without an extra $\beta$-factor, which has little effect in light of terms like $e^{\beta d}$ showing up. So, na\"{i}vely using this structural characterisation does not give promising results.

Further, unfortunately, the characterisation of $\mathrm{BC}$, and indeed of KL and TV divergences offered through this is essentially tight. This essentially follows from our results providing control on the $\chi^2$-divergences in various construction, and the control this imposes on $\mathrm{KL, TV}$ via the monotonicity of R\'{e}nyi divergences and Pinsker's inequality. It may be the case that in some special cases, tight bounds for structural testing may be derived via the statistical testing approach above. We have not explored this possibility in detail.

\subsection{Lower Bounds on Property Testing}\label{appx:prop_test}

In passing, we mention that our constructions improve upon lower bounds for some of the property tests studied in \cite{neykov2019property}. For instance, the triangle construction provides an obstruction to cycle testing that does not require explicit control on $\alpha$ as in \cite{neykov2019property}. Similarly, the Clique with a hole, and the Emmentaler clique with an extra node constructions may serve as obstructions to testing the size of the largest clique, and to testing the value of the maximum degree of the network structures in low temperatures. In high temperatures, the Fan graph construction shows that testing maximum degree is hard. In each case this either improves upon the lower bounds of \cite{neykov2019property} by either improving the exponent from $\beta d/4$ to $2\beta d(1 - o_d(1)),$ or by removing an explicit high-temperature condition that is enforced in the lower bound.


\section{Proofs of Widget Bounds}\label{appx:widget_proofs}

\textbf{An Observation} For Ising models $P,Q,$ \[1 + \chi^2(Q\|P) = \sum_x \frac{Q(x)^2}{P(x)} = \sum_x \frac{Z_P}{Z_Q^2} \exp{x^T 2\theta_Q x - x^T \theta_P x} = \frac{Z_P Z_{2Q-P}}{Z_Q^2}, \] where $Z_{2Q-P} := \sum_x \exp{x^T (2\theta_Q - \theta_P) x}$ is yet another partition function. We will repeatedly use this form of the $\chi^2$-divergence, without further comment, in the following.

\subsection{Star-Based Widgets}

\subsubsection{Triangle} 

\begin{proof}[Proof of Proposition \ref{prop:triangle}]

Let $P = P_{\mathrm{Triangle}}, Q = Q_{\mathrm{Triangle}}.$ Note that 

\begin{align*} P(x) &= \frac{1}{Z_P}e^{\lambda x_2(x_1 + x_3)} \\ Q(x) &= \frac{1}{Z_Q(\mu)}e^{\lambda x_2(x_1 + x_3)} e^{\mu x_1x_3} \end{align*}

Where the partition functions may simply be computed to obtain the expressions below: \begin{align*} Z_P &= 2^3 \cosh^2\lambda = 4(\cosh2 \lambda + 1) \\ Z_Q(\mu) &= 4(e^\mu \cosh 2\lambda + e^{-\mu}).\end{align*}

Further, we have that \[ W:= \mathbb{E}_P[(Q/P)^2] = \left(\frac{Z_P}{Z_Q(\mu)}\right)^2 \cdot \frac{1}{Z_P} \cdot \sum e^{\lambda x_2 (x_1 + x_3)} e^{2\mu x_1x_3} = \frac{Z_P Z_Q(2\mu)}{Z_Q(\mu)^2}. \]

Inserting the previous computed values of these partition functions, we have \begin{align*}
    W &= \frac{(\cosh 2\lambda + 1) (e^{2\mu} \cosh 2\lambda + e^{-2\mu})}{(e^{\mu} \cosh 2\lambda + e^{-\mu})^2} \\
      &= \frac{ e^{2\mu} \cosh^2 2\lambda + e^{-2\mu} + \cosh 2\lambda (e^{2\mu} + e^{-2\mu})}{(e^{\mu} \cosh 2\lambda + e^{-\mu})^2} \\
      &= 1 + \frac{\cosh 2\lambda (e^{\mu}  - e^{-\mu})^2}{(e^{\mu} \cosh 2\lambda + e^{-\mu})^2}\\
      &\le 1 + \frac{(e^{\mu} - e^{-\mu})^2}{e^{2\mu} \cosh 2\lambda} \\
      &\le1 +  \frac{4\sinh^2 \mu}{\cosh^2\mu \cosh 2\lambda} \\
      &\le 1 + 8e^{-2\lambda}  \tanh^2 \mu
\end{align*}

where the second and third inequalities both use that $e^x \ge \cosh x \ge e^x/2,$ for $x \ge 0$.
\end{proof}

\subsubsection{Fan with deletions}

In keeping with the rest of the text, these proofs will set $2B = d.$ Note that the value of $B$ does not enter the resulting bounds.

\begin{proof}[Proof of Proposition \ref{prop:fan}]

Let 
\begin{align*}
 P_{\ell, \eta, \mu, \lambda}(x) := \frac{1}{Z(\ell, \eta, \mu, \lambda)} &\exp{\lambda x_{d+1} (\sum_{i = 1}^{d/2} x_{2i})  + \mu x_{d+1} (\sum_{i = \ell + 1}^{d/2} x_{2i-1})} \\ & \cdot \exp{\eta x_{d+1} (\sum_{i = 1}^{\ell} x_{2i-1}) +  \lambda (\sum_{i = 1}^{d/2} x_{2i} x_{2i-1})}.
\end{align*}

Then $P_{\ell, \mathrm{Fan}} = P_{\ell, 0, \mu, \lambda},$ $Q_{\ell , \mathrm{Fan}} = P_{\ell, \mu, \mu, \lambda}.$ Further, $Z_{2Q - P} = Z(\ell, 2\mu, \mu, \lambda).$

Here again the partition function is simple to compute. In essence, the groups $(x_{2i- 1} , x_{2i})$ across $i$ are independent given $x_{d+1}$, and the expressions, unsurprisingly, are invariant to value of $x_{d+1}$. 

Unfortunately the calculations get a little messy. If one is not interested in the results on property testing in \S\ref{appx:prop_test}, then the following may be safely skipped. We do note that the steps below are elementary, it is just the form of the expressions that is long.

\begin{align*} &\quad Z(\ell,\eta, \mu, \lambda) \\ &= \sum_{x_{d+1}} \sum \exp{\lambda x_{d+1} (\sum_{i = 1}^{d/2} x_{2i})  + \mu x_{d+1} (\sum_{i = \ell + 1}^{d/2} x_{2i-1}) + \eta x_{d+1} (\sum_{i = 1}^{\ell} x_{2i-1}) +  \lambda (\sum_{i = 1}^{d/2} x_{2i} x_{2i-1})} \\
&=  \sum_{x_{d+1}} \prod_{i =1}^{\ell} \sum_{x_{2i-1}, x_{2i}}  e^{ x_{d+1}( \eta x_{2i-1} + \lambda x_{2i}) + \lambda x_{2i}x_{2i-1}  }  \cdot \prod_{i =\ell +1}^{d/2} \sum_{x_{2i-1}, x_{2i}}  e^{ x_{d+1}( \mu x_{2i-1} + \lambda x_{2i}) + \lambda x_{2i}x_{2i-1}}\\
&= \sum_{x_{d+1}} \left( 2e^{\lambda} \cosh( (\lambda + \eta)x_{d+1}) + 2e^{-\lambda} \cosh( (\lambda - \eta) x_{d+1})  \right)^{\ell} \\ &\qquad \qquad \qquad \cdot  \left( 2e^{\lambda} \cosh( (\lambda + \mu)x_{d+1}) + 2e^{-\lambda} \cosh( (\lambda - \mu) x_{d+1})  \right)^{d/2-\ell} \\
&= 2^{d+1} \left( e^{\lambda} \cosh( \lambda + \eta) + e^{-\lambda} \cosh( \lambda - \eta)  \right)^{\ell} \left( e^{\lambda} \cosh( \lambda + \mu) + e^{-\lambda} \cosh( \lambda - \mu)   \right)^{d/2 - \ell}\end{align*}

Thus, \begin{align*}
    1 + \chi^2(Q\|P) &= \frac{Z(\ell, 0, \mu, \lambda) Z(\ell, 2\mu, \mu, \lambda)}{ Z(\ell, \mu, \mu, \lambda)^2}\\
                         &= \left( \frac{ \left( e^{\lambda} \cosh( \lambda) + e^{-\lambda} \cosh( \lambda)  \right) \left( e^{\lambda} \cosh( \lambda +2\mu) + e^{-\lambda} \cosh( \lambda - 2\mu)  \right)}{ \left( e^{\lambda} \cosh( \lambda+\mu) + e^{-\lambda} \cosh( \lambda-\mu)  \right)^2} \right)^{\ell} \\
                         &=: U^{\ell}. 
\end{align*}

We proceed to estimate $U$. \begin{align*}
    U &= \frac{ \left( e^{\lambda} \cosh( \lambda) + e^{-\lambda} \cosh( \lambda)  \right) \left( e^{\lambda} \cosh( \lambda +2\mu) + e^{-\lambda} \cosh( \lambda - 2\mu)  \right)}{ \left( e^{\lambda} \cosh( \lambda+\mu) + e^{-\lambda} \cosh( \lambda-\mu)  \right)^2}\\
      &= \frac{e^{2\lambda} \cosh \lambda \cosh(\lambda + 2\mu) + e^{-2\lambda} \cosh \lambda \cosh(\lambda - 2\mu) + \cosh(\lambda) \cosh( \lambda + 2\mu) +  \cosh(\lambda) \cosh( \lambda - 2\mu) }{e^{2\lambda} \cosh^2(\lambda + \mu) + e^{-2\lambda} \cosh^2(\lambda - \mu) + 2\cosh( \lambda + \mu)\cosh( \lambda - \mu) }
\end{align*}

By eliminating one factor of the denominator from the numerator above, we obtain the sequence of relations that follows below.
\begin{align*}
    U &\overset{(a)}{=} 1 + \frac{ (e^{2\lambda} + e^{-2\lambda}) \sinh^2 \mu + \sinh(\mu) \left(\sinh(2\lambda + \mu) - \sinh(2\lambda - \mu) \right)}{e^{2\lambda} \cosh^2(\lambda + \mu) + e^{-2\lambda} \cosh^2(\lambda - \mu) + 2\cosh( \lambda + \mu)\cosh( \lambda - \mu)} \\
      &\overset{(b)}= 1 + \frac{2\cosh(2\lambda) \sinh^2 \mu + 2 \cosh(2\lambda)\sinh^2\mu}{ \left( e^{\lambda} \cosh( \lambda+\mu) + e^{-\lambda} \cosh( \lambda-\mu)  \right)^2} \\
      &= 1 + \frac{ 4\sinh^2(\mu) \cosh(2\lambda)}{e^{2\lambda} \cosh^2(\lambda + \mu) + e^{-2\lambda} \cosh^2(\lambda - \mu) + 2\cosh( \lambda + \mu)\cosh( \lambda - \mu)} \\
      &\overset{(c)}{\le} 1 + 4\frac{\sinh^2\mu}{ \cosh^2(\lambda + \mu)} \le 1 + 4\frac{\sinh^2 \mu}{\cosh^2 \lambda \cosh^2 \mu}\\
      &\le 1 + 16 e^{-2\lambda} \tanh^2\mu,
\end{align*}

where $(a)$ follows by the identities \begin{align*}
    \cosh(u) \cosh(u + 2v)  - \cosh^2(u + v) &= \sinh^2 v \\
    \cosh(u) \cosh(u + 2v) - \cosh(u+v) \cosh(u-v) &=  \sinh(v) \sinh(2u + v),
\end{align*} 

$(b)$ uses \[ \sinh(2u + v) - \sinh(2u - v) = 2\cosh (2u) \sinh u,\] and $(c)$ follows by dropping all terms but the first in the denominator, and observing that $e^{2\lambda} \ge \cosh (2\lambda).$ Finally, the inequality $\cosh(\lambda + \mu) \ge \cosh \lambda \cosh \mu$ holds because $\lambda , \mu \ge 0.$

\end{proof}

\subsection{Clique-based Widgets}

The method for showing the bounds is developed in the case of the Clique with a single edge deleted. While there are variations in the proofs of the following two cases, the basic recipe remains the same.

We begin with a technical lemma that is repeatedly used in the following. \begin{mylem}\label{lemma:concave_fn} Let $\tau:[a,b] \to \mathbb{R}$ be a function differentiable on $(a,b)$ such that $\tau'$ is strictly concave. If $\tau(a) < 0$ and $\tau(b) > 0,$ then $\tau$ has exactly one root in $(a,b)$ \end{mylem}
\begin{proof}\label{lem:concave}
Since $\tau'$ is concave, it can have at most two roots in $(a,b).$ Indeed, if there were three roots $a < x_1 < x_2 < x_3 < b,$ then $\exists t \in (0,1): x_2 = t x_1 + (1-t) x_3,$ and $0 = f(x_2) = t f(x_1) + (1-t)f(x_3)$ violates strict concavity. Further, between its roots, $\tau'$ must be positive, again by concavity.

Thus, we can break $[a,b]$ into three intervals $(I_1, I_2, I_3)$, some of them possibly trivial\footnote{i.e.~of cardinality $0$ or $1$. More precise characterisation can be obtained by casework on the number of roots of $\tau'$.}, of the from $([a, x_1), [x_1, x_2], (x_2, b])$, such that $\tau$ is monotone decreasing on the interiors of $I_1, I_3$ and monotone increasing on the interior of $I_2.$

Note that $\tau$ has at least one root by the intermediate value theorem. We now argue that it cannot have more than one. Since $\tau$ is falling on $I_1,$ it follows that $\sup_{x \in I_1} \tau(x) = \tau(a) < 0,$ and there is no root in $I_1.$ Similarly, since $\tau$ is falling on $I_3$, $\tau(b) = \inf_{x \in I_3} \tau(x) > 0,$ and there is no root in $I_3.$ This leaves $I_2$, and since $\tau$ is monotone on $I_2$, it has at most one root on the same.\qedhere
\end{proof}

\subsubsection{Clique with a single edge deleted}\label{sec:clique_single_pf}

\begin{proof}[Proof of Proposition \ref{prop:1edgeclique}]

Let $P = P_{\mathrm{Clique}}$ and $Q = Q_{\mathrm{Clique}}$ as defined in the main text. For given $\lambda,\eta,$ let \begin{align*}
    P_{\lambda, \eta} (x) := \frac{1}{Z(\lambda, \eta)} e^{\frac{\lambda}{2} \left( (\sum x_i)^2 - (d+1) \right)} e^{-\eta x_1x_2}
\end{align*}

Note that $P = P_{\lambda, \lambda -\mu},$ and $Q = P_{\lambda, \lambda}.$ Further, \[ W:= \mathbb{E}_P[(Q/P)^2] = \frac{Z(\lambda, \lambda - \mu)Z(\lambda, \lambda + \mu)}{Z(\lambda, \lambda)^2}.\] We begin by writing $Z$ in a convenient form, derived by breaking the configurations into bins depending on the number of $x_i$s that take the value $-1$: 

\begin{align*}
    Z(\lambda, \eta) &= \sum_{j = 0}^{d-1} \binom{d-1}{j} \left\{ e^{-\eta} \left( e^{\frac{\lambda}{2} (d + 1 - 2j)^2 - (d+1)} + e^{\frac{\lambda}{2} (d -3 - 2j)^2 - (d+1)}\right) + 2e^{\eta} e^{\frac{\lambda}{2} ( (d - 1 - 2j)^2 - (d+1))} \right\} .
\end{align*}

Notice above that since $(d - 3 - 2(d-1 - j))^2 = ( d + 1 - 2j)^2,$ and $\binom{d-1}{j} = \binom{d- 1}{d - 1 - j},$ it follows that the sums over the first two terms above are identical. Thus,

\begin{align*}
    Z(\lambda, \eta) &= 2 \sum \binom{d-1}{j} e^{-\eta}  e^{\frac{\lambda}{2} (d + 1 - 2j)^2 - (d+1)} +  2\sum e^{\eta} e^{\frac{\lambda}{2} ( (d - 1 - 2j)^2 - (d+1))} \\
    \iff \underbrace{\frac{Z(\lambda, \eta)}{2 e^{\lambda/2 (d^2 - d)}}}_{=: \widetilde{Z}(\lambda, \eta)} &= e^{\lambda d -\eta} 
    \underbrace{\sum \binom{d-1}{j} e^{-2\lambda j ( d + 1 - j)}}_{=: S_1(\lambda)} + e^{-(\lambda d - \eta)} 
    \underbrace{\sum \binom{d-1}{j} e^{-2\lambda j (d - 1 -j)}}_{=: S_2(\lambda)} \\
    \iff \widetilde{Z}(\lambda, \eta) &= e^{\lambda d - \eta} S_1(\lambda) + e^{-\lambda d + \eta} S_2(\lambda).
\end{align*}

Since the term appears often, we set $d' = d-1.$ As a consequence of the above, we have \begin{align*}
    W &= \frac{Z(\lambda, \lambda - \mu)Z(\lambda, \lambda + \mu)}{Z(\lambda, \lambda)^2}  = \frac{\widetilde{Z}(\lambda, \lambda - \mu)\widetilde{Z}(\lambda, \lambda + \mu)}{\widetilde{Z}(\lambda, \lambda)^2}\\
     &= \frac{(e^{\lambda d' + \mu} S_1(\lambda) + e^{-\lambda d' -\mu} S_2(\lambda) ) (e^{\lambda d' - \mu} S_1(\lambda) + e^{-\lambda d' + \mu} S_2(\lambda))}{(e^{\lambda d'} S_1(\lambda) + e^{-\lambda d'} S_2(\lambda))^2} \\
     &= 1 + 4\sinh^2 \mu \frac{S_1 S_2}{(e^{\lambda d'} S_1 + e^{-\lambda d'} S_2)^2}  \\
     &\le 1 + 4\sinh^2 \mu \frac{ e^{-2\lambda d'} S_2 (\lambda) }{S_1 (\lambda)}.
\end{align*}

The bounds are now forthcoming by controlling $S_1, S_2$ as in the following \begin{mylem}\label{lem:1edgeclique} If $d \ge 5$ and $\lambda (d-4) \ge \log(d)$, then \begin{align*}
    S_1(\lambda) &\ge 1 \\
    S_2(\lambda) &\le 2 + 3d e^{-2\lambda(d-2)} \le 2 + 3/d.
    \end{align*}

\end{mylem}

The bound follows directly from the control offered above.  \end{proof}

This proof describes closely the structure of the forthcoming proofs \begin{itemize}
    \item Begin by introducing one free parameter, $\eta$ varying which yields Ising models that interpolate between $P$ and $Q.$
    \item Express the $\chi^2$ divergence as a ratio of partition functions.
    \item Exploit the symetries of the mean field Ising model to more conveniently write these partition functions.
    \item Control the terms arising via a `ratio trick' as in the proof of Lemma \ref{lem:1edgeclique}. At time this is used more than once, or a more direct form of this trick is used instead.
\end{itemize} 

We conclude by showing Lemma \ref{lem:1edgeclique}.

\begin{proof}[Proof of Lemma \ref{lem:1edgeclique}]
$ S_1 \ge 1$ follows trivially, since all terms in the sum are non-negative and the first term is $\binom{d-1}{0} e^{0} = 1.$

Concentrating on $S_2,$ let $T_j := \binom{d-1}{j} e^{-2\lambda j (d - 1 -j)}$. Note that $S_2 = \sum T_j,$ and that $T_j = T_{d - 1 - j}$ for every $j$. Further, for $j \in [0:d-2],$ \[ \frac{T_{j+1}}{T_j} =  \frac{d-1-j}{j+1} e^{-2\lambda (d - 2 - 2j)}.\]

Treating $j$ as a real number in $[0,d-2],$ define \[\tau(j) = \log(d - 1 - j) - \log(j + 1) - 2\lambda(d - 2- 2j). \]

We have \begin{align*}
    \tau'(j) &= -\frac{1}{d - 1 - j} - \frac{1}{j + 1} + 4\lambda \\
    \tau''(j) &= -\frac{1}{(d - 1 - j)^2} + \frac{1}{(j+1)^2} \\
    \tau'''(j) &= - \frac{2}{(d -1 - j)^3} - \frac{2}{(j+1)^3} <0.
\end{align*}

We may now note that $\tau'$ is a strictly concave function on the relevant domain. Further, note that since  $\log(d- 1) \le 2\lambda(d-2)$ follows from our conditions, $\tau(0)< 0,$ and similarly, $\tau(d-2) > 0.$ By Lemma \ref{lem:concave}, $\tau$ has exactly one root in $[0, d-2]$ - in particular, this lies at $j = d/2 - 1$. But since $T_{j+1}/T_j = e^{\tau(j)},$ we obtain that for $j \le d/2-1, T_{j+1} \le T_j,$ and for $j\ge d/2 -1, T_{j+1} \ge T_j.$

Since $T$s are decreasing until $d/2-1$ and increasing after $d/2,$ it follows that for all $j \in [2:d-3],$ $T_j \le \max(T_2, T_{d-3}) = T_2.$ Now, under the conditions of the theorem, \begin{align*}
    \frac{T_2}{T_1} &= \exp{\tau(1)} = \exp{\log(d-2) - \log 2 - 2\lambda(d-4)} \\
                    &\le \exp{ \log(d-2) - \log 2 - 2\log(d) } \le 1/d,
\end{align*} where we have used the assumption $\lambda(d-4) \ge \log d$. Thus, \begin{align*} S_2 &= T_0 + T_1 + \sum_{j = 2}^{d-3} T_j + T_{d-2} + T_{d-1} \\ &\le 1 + T_1 + \frac{d-4}{d} T_1  + T_1 + 1 \\ &\le  2 + 3d \exp{-2\lambda (d-2)} \le 2 + 3/d. 
\qedhere \end{align*}

We call this method of estimating sums such as $S_2$ the \emph{ratio trick}, since they control the values of the sums by controlling the ratios of subsequent terms.
\end{proof}

\subsubsection{Clique with Large Hole}

The computations of this section are in essence a deepening of the previous section, and we will frequently make references to the same.

\begin{proof}[Proof of Proposition \ref{prop:clique_hole_Q||P}]
Once again condensing notation, let $P := P_{\ell, \mathrm{Clique}}, Q := Q_{\ell, \mathrm{Clique}}.$

Further, let \[P_{\ell, \lambda, \eta}(x) := \frac{1}{Z_\ell(\lambda, \eta)} e^{ \frac{\lambda}{2} \left( \sum_{1 \le i \le d+1} x_i \right)^2 - (d+1) } e^{ -\frac{\eta}{2} \left( \sum_{1 \le i \le \ell} x_i \right)^2 - \ell } \]

Again, $P = P_{\ell, \lambda, \lambda - \mu}, Q = P_{\ell, \lambda, \lambda}$ holds. $Z_\ell$ is the central object for this section, and has the following expression. This is derived by tracking the number of negative $x_i$s in both the bulk of the clique and the single `hole' separately.
\begin{align*} Z_\ell(\lambda, \eta) &:=  \sum_{ \{\pm 1\}^{d+1}} e^{ \frac{\lambda}{2} \left( \sum_{1 \le i \le d+1} x_i \right)^2 - (d+1) } e^{ -\frac{\eta}{2} \left( \sum_{1 \le i \le \ell} x_i \right)^2 - \ell }  \\
&= \sum_{i,j} \binom{\ell}{i} \binom{d+1 - \ell}{j}   e^{\frac{\lambda}{2} (d + 1 - 2i -2j)^2 - (d+1)} e^{\frac{-\eta}{2} (\ell - 2i)^2 - \ell}\end{align*} 

We normalise $Z_\ell$ by $e^{ \lambda/2 ( (d+1)^2 - (d+1))} e^{-\eta/2 (\ell^2 - \ell)},$ and put a $\thicksim$ over the normalised version\footnote{Unlike in \S\ref{sec:clique_single_pf}, we include the factor due to $\eta$ in the normalisation. This does not affect the further calculations since these factors cancel in the expression for $W$ below. More  importantly, the normalisation includes a factor of $e^{\lambda/2( (d+1)^2 - (d+1) )}$ instead of $e^{\lambda/2 (d^2 - d)}$. While the latter lent the formulae in the $\ell = 2$ case of the previous section a pleasant symmetry, the former yields more convenient expressions when dealing with $\ell$ abstractly. Due to this, the terms are further reduced by a common factor of $e^{\lambda d}.$ We highlight this here because of the cosmetic differences arising from these changes\textemdash for instance, the leading term in $\widetilde{Z}_{\ell}$ is just $S_1$ instead of $e^{\lambda d - \eta} S_1$ as in the \S\ref{sec:clique_single_pf}\textemdash which may irk the careful reader at first glance. } to get 
\begin{align*}
    \widetilde{Z}_\ell(\lambda, \eta) &:= \sum_{i,j} \binom{\ell}{i} \binom{d+1 - \ell}{j}   e^{-2\lambda j (d+1 - 2i -j)} e^{ 2\eta i (\ell-i)} e^{-2\lambda i (d+1 - i)} \\
                                   &=: \sum_{i = 0}^\ell \binom{\ell}{i} e^{ 2\eta i (\ell-i)} e^{-2\lambda i (d+1 - i)} S_i(\lambda)
\end{align*}

where \[ S_i(\lambda) := \sum_{j} \binom{d+1 - \ell}{j}  e^{-2\lambda j (d +1 - 2i -j) }.\] Notice that $S_i \ge 0$ for every $i$.

As before, we are interested in controlling \[W:= \frac{{Z}_\ell(\lambda, \lambda - \mu) {Z}_\ell(\lambda, \lambda + \mu)}{{Z}_\ell(\lambda, \lambda)^2}  =\frac{\widetilde{Z}_\ell(\lambda, \lambda - \mu) \widetilde{Z}_\ell(\lambda, \lambda + \mu)}{\widetilde{Z}_\ell(\lambda, \lambda)^2}. \]

To this end, note first that $ 2\lambda i(\ell-i) - 2\lambda i(d+1 - i) = -2\lambda (d+1 - \ell),$ and so, for instance, \[ \widetilde{Z}_\ell(\lambda, \lambda + \mu) = \sum_{i} \binom{\ell}{i} e^{ 2\mu i (\ell-i)} e^{-2\lambda i (d+1 - \ell)} S_i(\lambda). \] 

Collecting like terms in expressions of the above form, we obtain that \[ \frac{\widetilde{Z}_\ell(\lambda, \lambda - \mu)}{\widetilde{Z}_\ell(\lambda, \lambda)} = 1 + \frac{\sum_{i=1}^{\ell-1} \binom{\ell}{i} \left(e^{-2\mu i (\ell-i)} -1 \right) e^{-2\lambda i (d+1 - \ell)} S_i(\lambda)}{\widetilde{Z}_\ell(\lambda, \lambda)} \] and 

\[ \frac{\widetilde{Z}_\ell(\lambda, \lambda + \mu)}{\widetilde{Z}_\ell(\lambda, \lambda)} = 1 + \frac{\sum_{i=1}^{\ell-1} \binom{\ell}{i} \left(e^{ 2\mu i (\ell-i)} -1 \right) e^{-2\lambda i (d+1 - \ell)} S_i(\lambda)}{\widetilde{Z}_\ell(\lambda, \lambda)}, \] where the terms involving $i =0$ and $i = \ell$ in the numerator drop out because $e^{2\mu i(\ell - i)} = 1$ in these cases.

Now, if $\mu \ge 0$ the second terms in the above two expressions are respectively negative and positive, while if $\mu < 0,$ they are respectively positive and negative. It is a triviality that for $A<0<B, (1+A)(1+B) \le 1 + A + B.$ We thus have the upper bound \begin{align}\label{ineq:w:clique_hole}
    W &\le 1 + \frac{\sum_{i=1}^{\ell-1} \binom{\ell}{i} 2\left( \cosh 2\mu i(\ell-i) - 1 \right) e^{-2\lambda i (d+1 - \ell)} S_i(\lambda)}{\widetilde{Z}_\ell(\lambda,\lambda)} \notag \\
      &= 1 + 4\frac{\sum_{i=1}^{\ell-1} \binom{\ell}{i}  \sinh^2 ( \mu i(\ell-i)) e^{-2\lambda i (d+1 - \ell)} S_i(\lambda)}{\widetilde{Z}_\ell(\lambda,\lambda)}
\end{align}

While we will provide full proofs in the sequel, it may help to see where we are going first. Roughly, we argue via the ratio trick in the proof of Lemma \ref{lem:1edgeclique} in the previous section, that $S_i$ is bounded by $2 (1 + e^{-2\lambda (\ell-2i)(d+1 - \ell)})$, under conditions such as $\lambda(d+1 - 2\ell) \ge \log d + 1 - 2\ell.$ Plugging in this upper bound, and noting that after multiplication with $e^{-2\lambda i(d+1 - \ell)}$ we have a sum that is completely symmetric under $i \mapsto \ell-i,$ we can bound $W$ as \[W \le 1 + 16   \frac{\sum_{i=1}^{\ell-1} \binom{\ell}{i}  \sinh^2 ( \mu i(\ell-i)) e^{-2\lambda i (d+1 - \ell)} }{\widetilde{Z}_\ell(\lambda,\lambda)} . \] We then show that under the conditions of the proposition, the first term in the above sum dominates all the remaining terms, in the process utilising the condition $|\mu| \le \lambda$. Finally, using the trivial bound $\widetilde{Z}_\ell(\lambda, \lambda) \ge 1,$ we get the claied upper bound.

Let us then proceed. The control on the $S_i$s is offered below. \begin{mylem}\label{lem:clique_with_hold_bounding_Si}
If $\lambda(d+ 1 - 2\ell) \ge \log(d + 1 - 2\ell)$ and $d \ge 4\ell,$ then for every $i \in [1:\ell - 1],$ \[ S_i(\lambda) \le 2 + 2e^{-2\lambda(\ell - 2i) (d+1 - \ell)}. \]
\end{mylem}

Incorporating the above lemma into (\ref{ineq:w:clique_hole}), we have \begin{align}\label{ineq:w_clique_hole_postS}
    W &\le 1 + 8\frac{\sum_{i=1}^{\ell-1} \binom{\ell}{i}  \sinh^2 ( \mu i(\ell-i)) e^{-2\lambda i (d+1 - \ell)} \left(1 + e^{-2\lambda(\ell - 2i )(d+1-\ell)}\right)}{\widetilde{Z}_\ell(\lambda,\lambda)} \notag \\
      &\le 1 + 8 \frac{\sum_{i=1}^{\ell-1} \binom{\ell}{i}  \sinh^2 ( \mu i(\ell-i))  \left(e^{-2\lambda i (d+1 - \ell)} + e^{-2\lambda(\ell - i)(d+1-\ell)}\right)}{\widetilde{Z}_\ell(\lambda,\lambda)} \notag \\
      &\overset{(a)}{=} 1 + 16\frac{\sum_{i=1}^{\ell-1} \binom{\ell}{i}  \sinh^2 ( \mu i(\ell-i))  e^{-2\lambda i (d+1 - \ell)}}{\widetilde{Z}_\ell(\lambda,\lambda)} \notag \\
      &= 1 + \frac{16}{{\widetilde{Z}_\ell(\lambda,\lambda)}} \left( \sinh^2(\mu(\ell - 1)) e^{-2\lambda(d + 1 - \ell)} +  \sum_{i=2}^{\ell-1} \binom{\ell}{i}  \sinh^2 ( \mu i(\ell-i))  e^{-2\lambda i (d+1 - \ell)} \right) \notag \\
      &\overset{(b)}{\le} 1 + \frac{16}{{\widetilde{Z}_\ell(\lambda,\lambda)}} \left( \sinh^2(\mu(\ell - 1)) e^{-2\lambda(d + 1 - \ell)} + \sum_{i=2}^{\ell-1} \binom{\ell}{i}    e^{2|\mu| i \ell -2\lambda i (d+1 - \ell)} \right) \notag \\
      &\overset{(c)}{\le} 1 + \frac{16}{{\widetilde{Z}_\ell(\lambda,\lambda)}} \left( \sinh^2(\mu(\ell - 1)) e^{-2\lambda(d + 1 - \ell)} + \sum_{i=2}^{\ell-1} \binom{\ell}{i}    e^{ -2\lambda i (d+1 - 2\ell)} \right)
\end{align}

where the equality $(a)$ follows since each term in the sum is invariant under the map $i \mapsto \ell - i$, $(b)$ follows since $\sinh x \le e^{x},$ and $(c)$ used $\lambda \ge |\mu|.$ .

For $i\in[2:\ell],$ let $V_i$ denote the term corresponding to $i$ in the summation above, and let $V_1 =  \sinh^2(\mu(\ell - 1) e^{-2\lambda(d + 1 - \ell)}$. We will argue that $V_1$ dominates $V_i$ for every $i$ by using a weakened ratio trick. 

Note that \[ V_1 \ge e^{-2\lambda(d + 1 - \ell) - 2|\mu|(\ell - 1)} \ge e^{-2\lambda d}.\] Further,  \[\frac{V_i}{V_1} \le  \exp{ i\log \ell + 2\lambda d - 2\lambda i (d + 1 - 2\ell)}.\] This is smaller than $1/\ell$ so long as for every $i$, \[ i(2\lambda(d + 1 - 2\ell) - \log \ell ) > 2\lambda d + \log(\ell),\] which hold if the following conditions are true: \begin{align*}
    2\lambda(d + 1 - 2\ell) &> \log \ell \\
    4\lambda(d + 1 - 2\ell) &> 3\log \ell + 2\lambda d.
\end{align*} 

The above hold if $\lambda(d + 2 - 4\ell) \ge 3/2\log\ell$, which is true under the conditions of the proposition since $\ell < d/8,$ and since $\lambda (d+ 2  - 4\ell) \ge  \lambda d/2 \ge 3/2\log d.$

Finally, it remains to show that $\widetilde{Z}_\ell(\lambda,\lambda)$ is non-trivially large. But note that \( \widetilde{Z}_\ell(\lambda,\lambda) \ge S_0(\lambda) \ge 1.\)

Thus, we have shown that \[ W \le 1 + 32\ell \sinh^2(\mu(\ell - 1)) e^{-2\lambda (d + 1-\ell)}.\]

\end{proof}

\begin{proof}[Proof of Lemma \ref{lem:clique_with_hold_bounding_Si}]
For $j \in [0: d + 1 - \ell],$ let \[T_j :=  \binom{d+1 - \ell}{j}  e^{-2\lambda j (d + 1 - 2i -j) }.\] Recall that $S_i = \sum T_j$.
We will use the ratio trick again. To this end, observe that \[ \frac{T_{j+1}}{T_j} = \frac{d +1  - \ell - j}{j+1} \exp{ - 2\lambda (d - 2i - 2j}.\] Again treating $j$ as a real number in $[0:d - \ell],$ let \[ \tau(j) := \log(d + 1 -\ell - j) - \log(1+j) - 2\lambda(d - 2i - 2j).\]
By considerations similar to the previous section, $\tau$ is strictly concave, and by Lemma \ref{lem:concave}, $\tau$ has exactly one root so long as $\tau(0) < 0$ and $\tau(d - \ell) > 0.$ In this setting these conditions translate to \begin{align*}
    \log(d + 1 - \ell) &< 2\lambda(d - 2i) \\
    \log(d + 1 - \ell) &< -2 \lambda( d  - 2i -2(d -\ell)) = 2\lambda( d - 2(\ell - i) ).
\end{align*}

The above hold for every $i$ so long as $\log(d + 1 - \ell) < 2\lambda(d +2 - 2\ell).$

Since $\tau$ has a single root and is initially negative, we again find that for all $j \in [2:d-1-\ell],$ $T_j \le \max(T_2, T_{d -1 - \ell}).$ Further, \begin{align*}
    \frac{T_2}{T_1} &= \frac{d-\ell}{2} \exp{-2\lambda(d - 2 - 2i)} \le \frac{d - \ell}{2}\exp{-2\lambda(d  - 2\ell)} \le \frac{1}{d - \ell}\\
    \frac{T_{d-\ell-1}}{T_{d-\ell}} &= \frac{d-\ell}{2} \exp{-2\lambda(  d - 2(\ell - i)} \le \frac{1}{d-\ell}.
\end{align*}

Further, \begin{align*}
    \max\left(\frac{T_1}{T_0}, \frac{T_{d-\ell}}{T_{d+1 - \ell}}\right) \le (d + 1 - \ell)e^{-2\lambda (d - 2\ell)} \le 1/2.
\end{align*}

Thus, \begin{align*} S_1 &\le T_0 + T_{d+1-\ell} + (1 + (d-\ell-2)/(d-\ell)) \max(T_1,T_{d-\ell}) \\ &\le T_0 + T_{d+1 - \ell} + 2\max(T_1, T_{d-\ell}) \\ &\le 2(T_0 + T_{d + 1 - \ell}) .\end{align*} 

Now notice that \begin{align*}
    T_0 &= 1\\
    T_{d-\ell + 1} &= \exp{-2\lambda (d + 1 - \ell)(d + 1 - 2i - d - 1 + \ell)} = \exp{-2\lambda (\ell - 2i )(d+1-\ell)},
\end{align*}
and thus the claim follows.
\end{proof}

We now prove the reverse direction, i.e. control on $\chi^2(P\|Q)$. This is essentially a small variation on the previous setting.

\begin{proof}[Proof of Proposition \ref{prop:clique_hole_P||Q}]
    Referring to the previous proof, we instead need to control \[ W' = \frac{\widetilde{Z}_\ell(\lambda, \lambda) \widetilde{Z}_\ell(\lambda, \lambda + 2\mu)}{\widetilde{Z}_\ell(\lambda, \lambda + \mu)^2}. \] Proceeding in the same way, we may conntrol \[ W' \le 1 + \frac{\sum_{i=1}^{\ell-1} \binom{\ell}{i} \left( \cosh(4\mu i(\ell-i)) - 2\cosh(2\mu(i(\ell - i)) +  1 \right) e^{-2\lambda i (d+1 - \ell)} S_i(\lambda)}{\widetilde{Z}_\ell(\lambda, \lambda + \mu) }\]
    
    For succinctness, let $f(x) := \cosh(4\mu x) - 2\cosh(2\mu x) + 1.$ Note that $ 1 \le f(x) \le e^{4|\mu| x}.$ Since the $S_i$ are identical to the previous case, Lemma \ref{lem:clique_with_hold_bounding_Si} applies, and \begin{align}
    W' &\le 1 + 8\frac{\sum_{i=1}^{\ell-1} \binom{\ell}{i}  f(i(\ell - i)) e^{-2\lambda i (d+1 - \ell)} \left(1 + e^{-2\lambda(\ell - 2i )(d+1-\ell)}\right)}{\widetilde{Z}_\ell(\lambda,\lambda + \mu)} \notag \\
      &\le  1 + 16\frac{\sum_{i=1}^{\ell-1} \binom{\ell}{i}  f( i(\ell-i))  e^{-2\lambda i (d+1 - \ell)}}{\widetilde{Z}_\ell(\lambda,\lambda + \mu)} \notag \\
     &\le 1 + \frac{16}{{\widetilde{Z}_\ell(\lambda,\lambda+ \mu)}} \left( f(\ell - 1)  e^{-2\lambda(d + 1 - \ell)} + \sum_{i=2}^{\ell-1} \binom{\ell}{i}    e^{4|\mu| i \ell -2\lambda i (d+1 - \ell)} \right) \notag \\
      &{\le} 1 + \frac{16}{{\widetilde{Z}_\ell(\lambda,\lambda + \mu)}} \left( f(\ell - 1) e^{-2\lambda(d + 1 - \ell)} + \sum_{i=2}^{\ell-1} \binom{\ell}{i}    e^{ -2\lambda i (d+1 - 3\ell)} \right) \notag
\end{align}

Notice the distinction that the exponent in the second sum contains a $-3\ell$ instead of a $-2\ell$. Using $f(x) \ge 1,$ the same control on the relative values of $S_i$ and the summation holds as long as \[ 4\lambda(d + 1 - 3\ell) > 3\log \ell + 2\lambda d.\] This translates to demanding that $2\lambda(d - 6\ell ) > 3/2\lambda d,$ which holds for $\ell \le d/12$. Finally, $\widetilde{Z}_\ell(\lambda, \lambda + \mu) \ge 1$ as well, and thus, \[ W' \le 1 + 32\ell e^{-2\lambda(d + 1 - \ell)} \left( \cosh(4\mu(\ell-1)) - 2\cosh( 2\mu(\ell - 1) ) +  1 \right). \]

Finally, we note that for any $x,$ \begin{align*}
    \cosh(4x) - 2\cosh(2x) + 1 &= \sinh^2(2x) + (\cosh(2x) - 1)^2 \\ &= 4\sinh^2x \cosh^2 x + 4\sinh^4x = 4\sinh^2x \cosh^2 x (1 + \tanh^2 x) \\ &\le 2\sinh^2(2x). \qedhere
\end{align*} 
\end{proof}

\subsubsection{Emmentaler Cliques} 

\begin{proof}[Proof of Proposition \ref{prop:emmen_with_extra_node}]

Recall the setup - $d+1$ nodes are divided into $B = d/(\ell + 1)$ groups of $\ell + 1$ nodes each, denoted $V_1, \dots, V_B$, and the final node $d+1$ is kept separate. Recall that for a set $S,$ $x_S := \sum_{u \in S} x_u$. Define 

\[ P_{\ell, \lambda, \eta} = \frac{1}{Z_\ell(\lambda, \eta)} \exp{ \lambda/2 \left(\sum_{i = 1}^{B} x_{V_i} \right)^2 - \lambda/2\sum_{i=1}^B (x_{V_i}^2) + \lambda x_{v} \sum_{i = 2}^{B} x_{V_i} +  \eta x_v x_{V_1}}.  \]

Then $P = P_{\ell, \mathrm{Emmentaler}} = P_{\ell, \lambda, 0}, Q = Q_{\ell, \mathrm{Emmentaler}} = P_{\ell, \lambda, \mu}$ and $Z_{2Q - P} = Z_\ell(\lambda, 2\mu)$
Marginalising over $x_v,$ we get
\begin{align*}
    Z_\ell(\lambda, \eta) &= 2\sum_x \exp{ \lambda/2 \left(\sum_{i = 1}^{B} x_{V_i} \right)^2 - \lambda/2\sum_{i=1}^B (x_{V_i}^2)} \cosh\left(\lambda \sum_{i = 2}^{B} x_{V_i} +  \eta x_{V_1}\right) \\
             &\le 2\cosh(\lambda (d - \ell - 1) + \eta (\ell + 1) )  \sum_x \exp{ \lambda/2 \left(\sum_{i = 1}^{B} x_{V_i} \right)^2 - \lambda/2\sum_{i=1}^B (x_{V_i}^2)},
             \end{align*}
    
    while dropping all terms for which $|\sum_i x_{V_i} | < d,$ we get \begin{align*}             
       Z_\ell(\lambda, \eta) &\ge 4\cosh(\lambda (d - \ell - 1) + \eta (\ell + 1) ) e^{\lambda/2  (B^2 - B)(\ell + 1)^2} \\ &= 4\cosh(\lambda (d' - \ell - 1) + \mu (\ell + 1) ) e^{\lambda/2  ( d^2  - d(\ell + 1))} .
\end{align*}

To control $Z_\ell$ from above, it is necessary to control the partition function of the Emmentaler graph on $d$ nodes (i.e., with only the groups $V_1,\dots V_{B},$ and without the extra node from above. We set this equal to $Y_{\ell}(\lambda)$. Then, similarly tracking configurations by the number of negative $x_i$s in each part, \begin{align*}
    Y_\ell :=&  \sum_x \exp{ \lambda/2 \left(\sum_{i = 1}^{B} x_{V_i} \right)^2 - \lambda/2\sum_{i=1}^B (x_{V_i}^2)}. \\
            =& \sum_{j_1, \dots, j_B} \prod \binom{\ell + 1}{j_i} \cdot \exp{\lambda/2 \left( (d - 2\sum j_i)^2 - \sum (\ell + 1 - 2j_i)^2\right)}\\ 
            =& e^{\lambda/2(d^2 -d(\ell + 1))} \sum_{j_1, \dots, j_B} \prod \binom{\ell + 1}{j_i} \cdot \exp{-2\lambda \left(   (d- \ell - 1)(\sum j_i) + \sum j_i^2 - (\sum j_i)^2 \right)} 
\end{align*} 

For succinctness, let $d' := d - \ell - 1.$ We establish the following lemma after concluding this argument \begin{mylem} \label{lem:emmen_control_y} If  $\ell \le d/4$ and $\lambda (d-4) \ge 3\log (d),$ then \[ Y_\ell \le {2 e^{\lambda/2(d^2 -d(\ell + 1))}} \left(1 + 2d e^{-2\lambda d'}\right) \] \end{mylem}

Invoking the above lemma and the previously argued control on $Z_\ell,$ we get that \begin{align*}
    W := \mathbb{E}_P[(Q/P)^2] &= \frac{Z_{\ell}(\lambda, 0)Z_{\ell}(\lambda, 2\mu)}{Z_{\ell}(\lambda, \mu)^2} \\
                               &\le \frac{\cosh(\lambda d' )\cosh(\lambda d' + 2\mu(\ell + 1) )}{\cosh^2(\lambda d' + \mu(\ell + 1) )} \left(\frac{2Y_\ell}{4e^{\lambda/2  ( d^2  - d(\ell + 1))}}\right)^2 \\
                               &\le \left(1 + \frac{\sinh^2(\mu(\ell + 1))}{\cosh^2(\lambda d' + \mu(\ell + 1) )} \right) \left(1 + 2d e^{-2\lambda d'}\right)^2 \\
                               &\le \left(1 + 4\tanh^2(\mu(\ell + 1))e^{-2\lambda d'}\right)\left(1 + 2d e^{-2\lambda d'}\right)^2
\end{align*} 

Under the conditions of the theorem, both $4\tanh^2(\mu(\ell + 1))e^{-2\lambda d'}$ and $2d e^{-2\lambda d'}$ are smaller than $1/4.$ But for $x,y  ,$ it holds that $(1 + x)^2 < 1 + 3x$ and $(1 + 3x)(1 + y) < 1 + 4(x+y) \le 1 + 8\max(x,y).$ Lastly, $4\tanh^2 x \le 4 \le d,$ and thus, we have shown the bound \[ W \le 1 + 32 de^{-2\lambda (d - \ell - 1)}. \qedhere\] \end{proof}

\begin{proof}[Proof of Lemma \ref{lem:emmen_control_y}]

Fix a vector $(j_1, \dots, j_B)$ and let $k:= \sum j_i$. We will argue the claim by controlling the terms in $Y_\ell$ with a given value of $k$. \begin{mylem}\label{lem:this_got_really_bloody_nested_didnt_it}
If $\sum j_i = k \in [2:d-2],$ $\ell + 1 \le d/4$ and $\lambda (d-4) \ge 3\log (d),$ then \[ \prod \binom{\ell + 1}{j_i} \cdot \exp{-2\lambda \left(   d'(\sum j_i) + \sum j_i^2 - (\sum j_i)^2 \right)} \le \frac{1}{d^{\min(k, d-k)}} e^{-2\lambda d'}. \]
\end{mylem} 

Thus, we have the bound \[ \frac{Y_\ell}{ e^{\lambda/2(d^2 -d(\ell + 1))}} \le 2\left(1 + B(\ell  + 1) e^{-2\lambda d'}\right) + \sum_{k =2}^{d-2} \frac{N_k}{d^{\min(k, d-k)}} e^{-2\lambda d'}, \] where \[ N_k = \left|\left\{ j \in [0:\ell + 1]^B : \sum j_i = k \right\}.\right| \] Notice that $N_k = N_{d - k}.$
Further, for $k \le d/2,$ by stars and bars, \[ N_k \le \binom{k + B - 1}{k} \le (1 + (B-1)/k)^{k-1} \le B^k \le d^k \]

Consequently, $N_k \le d^{\min(k, d-k)},$ and we have established the upper bound \[ \frac{Y_\ell}{2 e^{\lambda/2(d^2 -d(\ell + 1))}} \le 1 + 2d e^{-2\lambda d'}.\qedhere  \]
\end{proof}

\begin{proof}[Proof of Lemma \ref{lem:this_got_really_bloody_nested_didnt_it}]
Note that $\binom{n}{m} \le n^{\min(m, n-m)}.$ Therefore, \[ \prod \binom{\ell + 1}{j_i} \le \exp{\min(k, d-k) \log(\ell + 1)}.\] Next, by Cauchy-Schwarz, \[ \sum j_i^2 \ge \frac{(\sum j_i)^2}{B} = k^2 \left( 1 - \frac{d'}{d} \right).\]

Let $\mathrm{LHS}, \mathrm{RHS}$ be the left and right hand sides of the inequality claimed in the Lemma. Using the above, \begin{align*} \log \frac{\mathrm{LHS}}{\mathrm{RHS}} &\le  \min(k, d-k) \log(d(\ell + 1)) - 2\lambda  \left( d'k + k^2 d'/d  -d'  \right) \\ &=  \min(k, d-k) \log(d(\ell + 1)) - 2\lambda  \frac{d'}{d}  \left( k(d -k) - d\right). \end{align*}

Let $f(k)$ be the upper bound above. Notice that $f(k) = f(d-k)$. Thus, it suffices to show that $f(u) \le 0$ for every real number $u \in [2, d/2].$

For a real number $u \in [2,d/2),$ it holds that $  f''(u) = 4\lambda > 0.$ It follows that $f$ attains its maxima on $\{2, d/2\}.$ Since $\ell + 1 < d/4,$ we have $d'/d \ge 3/4,$ and thus \begin{align*}
    f(2)& = 2\log(d(\ell + 1)) - 2\lambda \frac{d'}{d} (d-4) \le 4\log(d) - \frac{3}{2} \lambda (d-4) < 0\\
    f(d/2)& = \frac{d}{2} \left( \log(d(\ell + 1) - 2\lambda\frac{d'}{d} \cdot \frac{ (d - 4)}{2}\right) = \frac{d}{4} f(2) < 0.\qedhere
\end{align*}
\end{proof}

\subsubsection{Emmentaler v/s Full Clique} 
\begin{proof}[Proof of Proposition \ref{prop:emmen_vs_full}]Let 

\[ P_{\ell, \lambda, \eta}(x) := \frac{1}{Z_\ell(\lambda, \eta)} \exp{ \lambda/2 \left(\left(\sum_{i = 1}^{B} x_{V_i} \right)^2  - d\right)- (\lambda - \eta)/2\sum_{i=1}^B (x_{V_i}^2 - (\ell + 1))}.\]

Then $P_\ell = P_{\ell,\lambda, 0}, Q_\ell = P_{\ell, \lambda, \mu}$. Let $d' = d- 1 - \ell$.
Developing this a little, one can write \[ Z_{\ell}(\lambda, \eta) = C_{\ell, \lambda, \eta}\sum_{j_1, \dots, j_B} \prod \binom{\ell + 1}{j_i} \cdot e^{ -2\lambda \left(d'\sum j_i + \sum j_i^2 - (\sum j_i)^2 \right) -2\eta \left( (\ell +1) \sum j_i - \sum j_i^2 \right)}, \]

where \[ C_{\ell, \lambda, \eta} = \exp{\lambda/2 (d^2 - d(\ell + 1)) + \eta d(\ell + 1)/2}.\] Notice that \[ \frac{C_{\ell, \lambda, 0}C_{\ell, \lambda, 2\mu}}{C^2_{\ell, \lambda, \mu}} = 1,\] and thus \[ W := \mathbb{E}_P[(Q/P)^2] = \frac{{Z}_\ell(\lambda, 0) {Z}_\ell(\lambda, 2\mu)}{{Z}_\ell(\lambda, \mu)^2}  = \frac{\widetilde{Z}_\ell(\lambda, 0) \widetilde{Z}_\ell(\lambda, 2\mu)}{\widetilde{Z}_\ell(\lambda, \mu)^2},\] where \[ \widetilde{Z}_{\ell}(\lambda, \eta) := \frac{Z_\ell(\lambda, \eta)}{C_{\ell \lambda, \eta}} = \sum_{k= 0}^{d} e^{ -2\lambda \left(d' k - k^2 \right) - 2\eta (\ell + 1) k } \sum_{\substack{j_1, \dots, j_B\\ \sum j_i = k}} \prod \binom{\ell + 1}{j_i} \cdot e^{ -2(\lambda - \eta) \sum j_i^2}. \]

Let $T_k$ be the $k$\textsuperscript{th} term in the above. It holds that $T_{k} = T_{d- k}.$ Indeed, the original terms are invariant under the map $x \mapsto -x,$ and for $j = (j_1, \dots, j_B)$, this maps to $(\ell + 1) \mathbf{1} - j$ which has the sum $d- k.$

Further, since \[ \sum j_i^2 \le \max_i (j_i) \sum j_i \le (\ell + 1) \sum j_i,\] it holds that each term, which depends on $\eta$ as $e^{-2\eta ((\ell + 1) \sum j_i -\sum j_i^2}$ decreases as $\eta$ increases (or equivalently, $\frac{\partial}{\partial \eta}  \widetilde{Z}_{\ell}(\lambda, \eta) \le 0$)

Due to the above, for $\mu >0,$  \begin{align*}
    \rho_1 := \frac{\widetilde{Z}_\ell(\lambda, 0) - \widetilde{Z}_\ell(\lambda, \mu)}{\widetilde{Z}_\ell(\lambda, \mu)} &\ge 0 \\
    \rho_2 := \frac{\widetilde{Z}_\ell(\lambda, 2\mu) - \widetilde{Z}_\ell(\lambda, \mu)}{\widetilde{Z}_\ell(\lambda, \mu)} &\le 0,
\end{align*} 

yielding, \[ W = \frac{\widetilde{Z}_\ell(\lambda, 0) \widetilde{Z}_\ell(\lambda, 2\mu)}{\widetilde{Z}_\ell(\lambda, \mu)^2} \le 1 + \rho_1 + \rho_2. \] (For $\mu < 0,$ the signs of both $\rho_1$ and $\rho_2$ are flipped, giving the same bound.)

We now offer control on $\rho_1 + \rho_2,$ to complete the argument. To this end, note that \[ 1 - 2e^{-2\mu \left((\ell + 1)k - \sum j_i^2 \right)} + e^{-4\mu \left((\ell + 1)k - \sum j_i^2 \right)} = \left( 1 - e^{-2\mu \left((\ell + 1)k - \sum j_i^2 \right)} \right)^2, \] and thus \begin{align*} \widetilde{Z}_\ell(\lambda, \mu) (\rho_1 + \rho_2) &= \sum_{k = 1}^{d-1} \sum_{j: \sum j_i = k} \prod\binom{\ell+1}{j_i} e^{-2\lambda ( d'k - k^2 + \sum j_i^2)} \left( 1 - e^{-2\mu \left((\ell + 1)k - \sum j_i^2 \right)} \right)^2 \\
&\le 2\sum_{k = 1}^{\lfloor d/2\rfloor} \sum_{j: \sum j_i = k} \prod\binom{\ell+1}{j_i} e^{-2\lambda ( d'k - k^2 + \sum j_i^2)} \left( 1 - e^{-2\mu \left((\ell + 1)k - \sum j_i^2 \right)} \right)^2,\end{align*}

where we have used the symmetry of the $T_ks$ above.

We argue below that the first term in the above strongly dominates all subsequent terms. \begin{mylem}\label{lem:emmen_vs_full}
If $\sum j_i = k \in [2:\lfloor d/2\rfloor],$ $\ell + 1 \le d/4$ and $\lambda (d-4) \ge 3\log (d),$ then \[   \prod\binom{\ell+1}{j_i} e^{-2\lambda ( d'k - k^2 + \sum j_i^2)}  \le \frac{1}{d^k} e^{-2\lambda d'}. \]
\end{mylem}

Using the above, along with $\sum j_i^2 \ge \sum j_i$ and the fact that the number of $B$-tuples of whole numbers that sum up to $k$ is at most $\binom{k+B-1}{k} \le (eB)^k \le d^k,$ we immediately have \[ \widetilde{Z}_\ell(\lambda, \mu) (\rho_1 + \rho_2) \le 2d e^{-2\lambda d'} \sum_{k = 1}^{d/2} \left( 1- e^{-2\mu\ell  k}\right)^2 . \] 

We bound the sum above in two ways - firstly, each term is $\le 1,$ and so the sum is at most $d/2$. Further, using $1 - e^{-x} \le x,$ the sum is at most $4\sum \mu^2 \ell^2 k^2 \le \mu^2 d^5.$ This gives , \begin{align*}
    \widetilde{Z}_\ell(\lambda, \mu) (\rho_1 + \rho_2) \le 2d^2 \min(1, \mu^2 d^4) e^{-2\lambda (d - 1 - \ell)}
\end{align*}
The bound on $W$ now follows since $\widetilde{Z}_\ell(\lambda, \mu) \ge 2$ trivially. 
\end{proof}

\begin{proof}[Proof of Lemma \ref{lem:emmen_vs_full}] This is essentially the same as Lemma \ref{lem:emmen_control_y}, and may be proved similarly. \end{proof}

\subsubsection{The Clique versus the Empty Graph in High Temperatures}\label{appx:pf_of_clique_versus_empty}

\begin{proof}[Proof of Proposition \ref{prop:clique_vs_empty}] 

This proof heavily relies on techniques we encountered in \cite{cao2018high}. The principal idea is via the following representation of the law of an Ising model with uniform edge weights, and the subsequent expression (and upper bound) for its partition function, both of which we encountered in the cited paper.

Let $\tau = \tanh(\mu)$. Then the law of the Ising model on a $m$-vertex graph $G$ with uniform weights $\alpha$ is \[P(X = x) = \frac{\prod_{(i,j) \in G} (1 + \tau X_iX_j)}{2^m \mathbb{E}_0[ \prod_{(i,j) \in G} (1 + \tau X_iX_j)},\] where $\mathbb{E}_0$ denotes expectation with respect to the uniform law on $\{-1, 1\}^m.$ This is shown by noticing that $\exp{x} = \cosh(x) (1 + \tanh(x))$, and then observing that for $x = \mu X_i X_j,$ since $X_iX_j = \pm 1,$ the same is equal to $\cosh(\mu) (1 + \tanh(\mu) X_iX_j)$. The $\cosh(\mu)$ term is fixed for all entries, and thus vanishes under the normalisation. The denominator is simply a restatement of $\sum_{\{-1,1\}^m} \prod_{(i,j) \in G} (1 + \tau X_iX_j).$

Let the denominator of the above be denoted $2^m \Phi(\tau ; G)$. We further have the expansion \[ \Phi(\tau ;G) = \sum_{u \ge 0} \mathscr{E}(u,G) \tau^u,\] where $\mathscr{E}(j,G)$ denotes the number of `Eulerian subgraphs of $G$', where we call a graph Eulerian if each of its connected components is Eulerian (and recall that a connected graph is Eulerian if and only if each of its nodes has even degree). This arises by expanding the above product out to get \[ \Phi(\tau;G) = \sum_{u \ge 0} \tau^u \cdot \sum_{ \textrm{choices of $u$ edges $(i_1, j_1), (i_2, j_2), \dots (i_u, j_u)$}} \mathbb{E}_0[ X_{i_1}X_{j_1} \dots X_{i_u}X_{j_u}]. \]

Now, due to the independence, if any node of the $X_i$s or the $X_j$s appears an odd number of times in the product, the expectation of that term under $\mathbb{E}_0$ is zero. If they all appear an even number of times, the value is of course $1$. Thus the inner sum, after expectation, amounts to the number of groups of $u$ edges such that each node occurs an even number of times in this set of edges, which corresponds to the number of Eulerian subgraphs of $G$, defined in the above way. 

A further subsidiary lemma controls the size of $\mathscr{E}(u,G)$ as follows, where we abuse notation and use $G$ to denote the adjacency matrix of the graph $G$.\[ \mathscr{E}(u,G) \le (2\|G\|_F)^u.\] The idea behind this is to first control the number of length-$v$ closed walks in a graph, by noticing that the total number of length $v$ walks from $i$ to $i$ is $(G^v)_{i,i}$, summing which up gives an upper bound on the number of closed length $v$ walks of $\mathrm{Tr}(G^v) \le \|G\|_F^v.$ Next, we note that to get an Eulerian subgraph of $G$ with $u$ edges, we can either take a closed walk of length $u$ in $G$, or we can add a closed walk of length $v \le u-2$ to an Eulerian subgraph with $u - v$ edges. This yields a Gr\"{o}nwall-style inequality that the authors solve inductively. Please see \cite[Lemma A.1]{cao2018high}.

Now, let $P$ be the Ising model $K_m$ with uniform weight $\alpha,$ and let $Q$ be the Ising model on the empty graph on $m$ nodes.  Using the above expression for the law of an Ising model, we have \[ 1 + \chi^2(Q\|P) = \mathbb{E}_Q[Q/P] = \mathbb{E}_0[ \prod_{i < j}(1 +  \tau X_iX_j)] \mathbb{E}_0[  \prod_{i<j}(1 + \tau X_iX_j)^{-1}] ,\] which, by multiplying and dividing each term in the second expression by $1 - \tau X_iX_j,$ and noting that $X_i^2X_j^2 = 1,$ may further be written as \begin{align*} 1 + \chi^2(Q\|P) &= \mathbb{E}[\prod_{i < j}(1 + \tau X_iX_j)] \mathbb{E}\left[\frac{\prod_{i < j}(1 - \tau X_iX_j)}{(1- \tau^2)^{-\binom{m}{2}}}\right] \\ &= \Phi(\tau;K_m)\Phi(-\tau;K_m) (1- \tau^2)^{-\binom{m}{2}}.\end{align*}

Since the above expression is invariant under a sign flip of $\tau,$ we may assume, without loss of generality, that $\tau \ge 0.$ Next, notice, due to the expansion in terms of $\mathscr{E}$ of $\Phi,$ that $\Phi(-\tau;K_m) \le \Phi(\tau;K_m)$ for $\tau \ge 0.$ Further, for $\tau \ge 0,$ using the bound on $\mathscr{E}(u,G)$, \[ \Phi(\tau;K_m) \le \mathscr{E}(0;K_m) + t\mathscr{E}(1;K_m) + t^2\mathscr{E}(2; K_m) + \sum_{u \ge 3} (2t\|K_m\|_F)^u. \]

Now notice that $\mathscr{E}(0;K_m) = 1,$ and $\mathscr{E}(1;K_m) = \mathscr{E}(2;K_m) = 0$. The first of these is because there is only a single empty graph, while the other two follow since $K_m$ is a simple graph. Further, $\|K_m\|_F = \sqrt{m(m-1)} \le m$. Thus, we have \[ \Phi(\tau;K_m) \le 1 + \sum_{u \ge 3} (2tm)^u.\] Now, since $2\tanh(\alpha)m \le 2\alpha m \le 1/16 < 1/2,$  we sum up and bound the geometric series to conclude that $\Phi(\tau;K_m) \le 1 + 16 (tm)^3 \le 1 + (tm)^2,$ and as a consequence, \[ \Phi(\tau;K_m)^2 \le (1 + (tm)^2)^2 \le 1 + 3(tm)^2 \le \exp{3(tm)^2}.\]

Further, since $\tau m < 1/32,$  and $m \ge 1,$ we have $\tau < 1/32,$ which in turn implies that $(1- \tau^2)^{-1} \le \exp{2\tau^2}$. Thus, we find that \[ 1 + \chi^2(P\|Q) \le \exp{3 (\tau m)^2} \cdot (\exp{2\tau ^2})^{m^2/2} \le \exp{4(\tau m)^2}\le 1 + 8(\tau m)^2,\] where the final inequality uses the fact that for $x < \ln(2),$ $e^x \le 1 + 2x,$ which applies since $4(\tau m)^2 \le 4/(32)^2 < \ln(2).$\qedhere
\end{proof}

It is worth noting that Proposition \ref{prop:high_temp_clique} is also shown in the above framework by \cite{cao2018high}. The main difference, however, is that in the $\chi^2$ computations, the square of $\prod (1 + \tau X_iX_j)$ appears. The technique the authors use is to extend the notion of $\mathscr{E}$ to multigraphs, and show the same expansion for these, along with the same upper bound for $\mathscr{E}(u,G)$, this time with the entries of $G$ denoting the number of edges between the corresponding nodes.

\end{appendix}
\end{document}